\documentclass[10pt, journal]{IEEEtran}

\usepackage{amsmath, amsthm, amsfonts, amssymb}
\usepackage{bm}

\newif\iffullpaper
\fullpapertrue
\usepackage[subrefformat=parens, labelformat=parens]{subfig}
\usepackage[font=small]{caption}
\usepackage[export]{adjustbox}

\usepackage{xcolor}
\usepackage{soul}
\usepackage{rotating}
\usepackage{algorithm, algorithmicx, algpseudocode}
\usepackage[shortlabels]{enumitem}
\usepackage{graphicx}
\usepackage{upgreek}
\usepackage{tabularx}
\usepackage{amsmath}
\DeclareMathOperator*{\argmax}{arg\,max}

\usepackage{graphicx}
\usepackage{epstopdf}

\newtheorem{lemma}{Lemma}

\newtheorem{theorem}{Theorem}

\newtheorem{definition}{Definition}

\newtheorem{remark}{Remark}

\usepackage{url}
\begin{document}

\sloppy

\iffullpaper
\title{On the Optimization and Stability of \\ Sectorized Wireless Networks}
\else
\title{Optimizing Sectorized Wireless Networks: \\ Model, Analysis, and Algorithm}
\fi

% \titlenote{Produces the permission block, and copyright information}
% \subtitle{Extended Abstract}
% \subtitle{ACM MobiHoc 2019 Submission \#XXX, 10-page}
% \subtitlenote{The full version of the author's guide is available as \texttt{acmart.pdf} document}

\author{Panagiotis~Promponas, %~\IEEEmembership{Student~Member,~IEEE,}
Tingjun~Chen, %~\IEEEmembership{Member,~IEEE,}
and~Leandros~Tassiulas%,~\IEEEmembership{Senior~Member,~IEEE}
\vspace{-\baselineskip}
\thanks{This work was supported in part by NSF grants CNS-2128530, CNS-2128638, CNS-2146838, CNS-2211944, AST-2232458, ECCS-2434131, and by ARO MURI grant W911NF2110325.
A partial and preliminary version of this paper appeared in ACM MobiHoc'23, Oct. 2023~\cite{promponas2023}.}
\thanks{P. Promponas and L. Tassiulas are with the Department of Electrical and Computer Engineering, Yale University, New Haven, CT, USA (email: \{panagiotis.promponas, leandros.tassiulas\}@yale.edu).}
\thanks{T. Chen is with the Department of Electrical and Computer Engineering, Duke University, Durham, NC, USA (email: tingjun.chen@duke.edu).}
}% <-this % stops a space

% \author{Panagiotis Promponas$^\dag$, Tingjun Chen$^\natural$, Leandros Tassiulas$^\dag$}
% % \authornote{Dr.~Trovato insisted his name be first.}
% % \orcid{1234-5678-9012}
% \affiliation{%
%   \institution{\vspace*{0.5ex}$^\dag$Electrical Engineering, Yale University, $^\natural$Electrical and Computer Engineering, Duke University}
%   \country{}
% }
% \email{{panagiotis.promponas, leandros.tassiulas}@yale.edu, \space tingjun.chen@duke.edu}

%%%%HERE
% \settopmatter{printacmref=false} % Removes citation information below abstract
% \renewcommand\footnotetextcopyrightpermission[1]{} % removes footnote with conference information in first column
% \pagestyle{plain} % removes running headers
%%%%%%%

% \renewcommand{\shortauthors}{P. Promponas, T. Chen, and L. Tassiulas}
% \renewcommand{\shorttitle}{Optimizing Sectorized Wireless Networks: Model, Analysis, and Algorithm}
% \acmConference[MobiHoc'23]{MobiHoc}{October 23--26, 2023}{Washington DC, USA}

\setlength{\abovedisplayskip}{6pt}
\setlength{\belowdisplayskip}{6pt}

\maketitle

\begin{abstract}
Future wireless networks need to support the increasing demands for high data rates and improved coverage. One promising solution is sectorization, where an infrastructure node is equipped with multiple sectors employing directional communication. Although the concept of sectorization is not new, it is critical to fully understand the potential of sectorized networks, such as the rate gain achieved when multiple sectors can be simultaneously activated. In this paper, we focus on sectorized wireless networks, where sectorized infrastructure nodes with beam-steering capabilities form a multi-hop mesh network. We present a sectorized node model and characterize the capacity region of these sectorized networks. We define the flow extension ratio and the corresponding sectorization gain, which quantitatively measure the performance gain introduced by node sectorization as a function of the network flow. Our objective is to find the sectorization of each node that achieves the maximum flow extension ratio, and thus the sectorization gain. Towards this goal, we formulate the corresponding optimization problem and develop an efficient distributed algorithm that obtains the node sectorization under a given network flow with an approximation ratio of 2/3.  Additionally, we emphasize the class of Even Homogeneous Sectorizations, which simultaneously enhances the efficiency of dynamic routing schemes with unknown arrival rates and increases network capacity. We further propose that if sectorization can be adapted dynamically over time, either a backpressure-driven or maximum weighted b-matching-based routing approach can be employed, thereby expanding the achievable capacity region while preserving stability under unknown traffic conditions. Through extensive simulations, we evaluate the sectorization gain and the performance of the proposed algorithms in various network scenarios. 
\end{abstract}


\newcommand{\littlesum}{\mathop{\textstyle\sum}}
\newcommand{\littleint}{\mathop{\textstyle\int}}

\newcommand{\myRe}[1]{\mathsf{Re}[#1]}
\newcommand{\myIm}[1]{\mathsf{Im}[#1]}

\newcommand{\nodeLoc}[1]{\mathbf{x}({#1})}

\newcommand{\sector}{\upsigma}
\newcommand{\sectorVec}{\upsigma}
\newcommand{\nodeRange}{R}
\newcommand{\nodeBeamwidth}{\theta}

% network and graph
\newcommand{\graph}{G}
\newcommand{\graphNet}{G^{\sector}}
\newcommand{\graphAux}{H^{\sector}}
\newcommand{\graphAuxPlain}{H}

\newcommand{\numNode}{N}
\newcommand{\setNode}{\mathcal{N}}
\newcommand{\link}{\ell}
\newcommand{\numLink}{L}
\newcommand{\setLink}{\mathcal{L}}

\newcommand{\setVertex}{V}
\newcommand{\edge}{e}
\newcommand{\setEdge}{E}

% sectorized node model
\newcommand{\numSector}{K}
\newcommand{\numSectorVec}{\mathbf{K}}
\newcommand{\setSector}{\Gamma}
\newcommand{\setHomSector}{\mathcal{H}}
\newcommand{\setSectorAxis}{\mathcal{L}}
\newcommand{\nodeFov}[2]{\omega_{#1}^{#2}}
\newcommand{\nodeSector}[2]{\sigma_{#1}^{#2}}
\newcommand{\nodeSectorAxis}[2]{\eta_{#1}^{#2}}

\newcommand{\directLink}[4]{({\nodeSector{#1}{#2}}, {\nodeSector{#3}{#4}})}
\newcommand{\angleTwoLink}[2]{\angle({#1},{#2})}

% traffic and flow
\newcommand{\arrival}{A}
\newcommand{\arrivalRate}{\upalpha}
\newcommand{\arrivalRateVec}{\bm{\upalpha}}
\newcommand{\arrivalRateMat}{\bm{\upalpha}}
\newcommand{\flow}{f}
\newcommand{\flowVec}{\mathbf{\flow}}
\newcommand{\flowSet}{\mathcal{F}}

\newcommand{\queue}{Q}
\newcommand{\queueVec}{\mathbf{Q}}

\newcommand{\capRegion}{\Lambda}
\newcommand{\convexHull}{\mathrm{Co}}

% schedule and matching
\newcommand{\bp}{D}
\newcommand{\bpVec}{\mathbf{\bp}}

\newcommand{\schedule}{X}
\newcommand{\scheduleVec}{\mathbf{\schedule}}
\newcommand{\scheduleSet}{\mathcal{\schedule}}
\newcommand{\scheduleLink}[1]{\schedule_{#1}}

\newcommand{\matching}{M}
\newcommand{\matchingVec}{\mathbf{\matching}}
\newcommand{\matchingSet}{\mathcal{\matching}}

\newcommand{\isomorphism}{\mathcal{O}}

% polytope
\newcommand{\polytope}{\mathcal{P}}
\newcommand{\polytopeP}{\mathcal{P}}
\newcommand{\polytopeQ}{\mathcal{Q}}

\newcommand{\flowExtension}{\uplambda}
\newcommand{\flowExtensionOpt}{\flowExtension}
\newcommand{\flowExtensionQ}{\upmu}
\newcommand{\flowExtensionQOpt}{\flowExtensionQ}

% optimization
\newcommand{\opt}{(\textbf{Opt})}
\newcommand{\optApprox}{(\textbf{Opt-Approx})}
\newcommand{\optApproxN}{(\textbf{Opt-Approx}-$n$)}

\newcommand{\sectorOpt}{\sector^{\star}}
\newcommand{\sectorApprox}{\widetilde{\sector}}
\newcommand{\sectorApproxDist}{\upnu}
\newcommand{\sectorApproxAlgo}{\uppi}

\newcommand{\algoName}{{\textsc{Sectorize}\textrm{-}\emph{n}}}

\newcommand{\decisionProbN}{(\textbf{Decision}-$n$)}
\newcommand{\decisionProb}{{\textsc{ExistSectorization}\textrm{-}\emph{n}}}
\newcommand{\decisionThreshold}{T}
\newcommand{\decisionThresholdCrit}{T_{n}^{\textrm{crit}}}

%% gain
\newcommand{\sectorizationGain}{g}
\newcommand{\sectorizationGainP}{g_{\flowExtension}}
\newcommand{\sectorizationGainQ}{g_{\flowExtensionQ}}

\section{Introduction}
\label{sec:intro}

Future wireless networks and systems including 5G/6G need to provide multi-Gbps data rates with guaranteed coverage, leveraging massive antenna systems~\cite{shepard2012argos}, the widely available spectrum at millimeter-wave (mmWave) frequency~\cite{rappaport2013millimeter, saad2020vision}, and network densification~\cite{giordani2020toward}. In addition to deploying more cell sites, the \emph{sectorization} of each cell -- dividing each cell into a number of non-overlapping sectors -- can significantly enhance the cell capacity and coverage by improving the spatial reuse and reducing interference~\cite{Agiwal_Roy_Saxena_2016}.

There are many applications of sectorized networks to wireless access and backhaul networks, in both {sub-6}\thinspace{GHz} and mmWave frequency bands. For example, in a mmWave backhaul network that can provide fiber-like data rates (e.g., the Terragraph {60}\thinspace{GHz} solution~\cite{terragraph_spectrum}), each mmWave node is usually composed of a number of sectors, each of which is equipped with a phased array with beamforming capability~\cite{sadhu201728, gao2024mambas, chen2023open}. In addition, integrated access and backhaul (IAB)~\cite{polese2020integrated} in the mmWave band supporting flexible and sectorized multi-hop backhauling started to be standarized since 3GPP Release 16. Recent efforts also focused on using increased number of sectors per infrastructure node to provide better coverage (e.g., SuperCell~\cite{supercell} supports 36 azimuth sectors per node). Therefore, it is important to study the performance of sectorized networks, especially when each node can simultaneously activate multiple sectors for signal transmission and/or reception.

In this paper, we focus on the \emph{modeling, analysis, and optimization of sectorized wireless networks}, where sectorized nodes form a multi-hop mesh network for data forwarding and routing. In particular, we consider the scenario where a sectorized infrastructure node can simultaneously activate many sectors supporting beam-steering capability, and focus on optimizing the sectorization of each node given the network conditions.
We present the model of a sectorized wireless backhaul network consisting of a number of (fixed) sectorized infrastructure nodes, and describe the link interference model and characterize the capacity region of the these networks.
For a sectorized network, we introduce a latent structure of its connectivity graph, called the auxiliary graph, which captures the underlying structural property of the network as a function of the sectorization of each node. We show that the capacity region of a sectorized network can be described by the matching polytope of its auxiliary graph.

Then, we present the definitions of \emph{flow extension ratio} and the corresponding \emph{sectorization gain} as a function of the network flow. These two metrics quantitatively measure how much the network flow can be extended in a sectorized wireless network, and thus quantifies the performance of the network sectorization.
We formulate an optimization problem with the objective to find the optimal sectorization of the network that maximizes the flow extension ratio (i.e., achieves the highest sectorization gain) under a given network flow. Due to the analytical intractability of the problem, we develop a novel distributed algorithm, {\algoName}, that approximates the optimal sectorization of each node in the network. We also prove that {\algoName} is a $2/3$-approximation algorithm.

\iffullpaper
Moreover, we motivate the sectorization of a network as a medium to embed a useful structural property in the network's ``effective'' topology. In particular, the sectorization choice changes the topology of the auxiliary graph, on which we can run the algorithms to study and operate the (sectorized) network. As a motivating example, we introduce a specific sectorization option, called Even Homogeneous Sectorization, which forces the Auxiliary Graph to be bipartite. Such a sectorization option both
(\emph{i}) simplifies the characterization of the capacity region, and
(\emph{ii}) speeds up the stabilization of the network under the backpressure dynamic policy even compared to the unsectorized network. Surprisingly, with sufficient distributed computing power, the more we sectorize the network, the more we increase its capacity region while we decrease the time needed for its stabilization.
\else
\fi

While most existing research has focused on a single, fixed sectorization, recent technology advances might offer the prospect of dynamic sectorization, where nodes can continuously adjust their sector alignments. We also propose two complementary solutions that capitalize on this flexibility. First, a backpressure-driven heuristic that obtains an instantaneous flow estimate and then re-sectorizes to optimize performance with respect to that flow. Second, a throughput-optimal policy that jointly optimizes routing, scheduling, and sectorization via a maximum weighted b-matching formulation.

Finally, we numerically evaluate the performance of the proposed algorithm through extensive simulations. We consider both an example 7-node network and a large number of random networks with varying numbers of sectors per node, node density, and network flows. The simulation results confirm our analysis and show that the approximate sectorization gain increases sublinearly with respect to the number of sectors per infrastructure node. Moreover, the results show that the dynamic sectorization routing schemes can significantly expand the network’s capacity region and unlock new performance gains not achievable with static sectorization.

To summarize, the main contributions of this paper include:
\begin{enumerate}[topsep=3pt, itemsep=3pt]
\item[(\emph{i})]
A general sectorized multi-hop wireless network model and a comprehensive characterization of its capacity region based on matching polytopes,
\item[(\emph{ii})]
A distributed approximation algorithm that optimizes the sectorization of each node under a given network flow with performance guarantee, 
\item[(\emph{iii})]
A specific sectorization, called Even Homogeneous Sectorization, which simplifies the characterization of the capacity region, speeds up the stabilization of the network under the backpressure dynamic policy and increases the capacity. 
\item[(iv)]
\emph{Dynamic Sectorization Approaches}. We further extend our framework beyond a single fixed sectorization to a dynamic scenario, proposing both an efficient \emph{backpressure-based heuristic} for real-time sector reconfiguration and a \emph{throughput-optimal} solution leveraging maximum weighted b-matching.
\item[(\emph{v})]
Extensive simulations for performance evaluation of the proposed sectorized network model and algorithm.

\end{enumerate}
We also note that the developed sectorized network model and analysis are general and applicable to other networks that share similar structures of the connectivity and auxiliary graphs.

The remainder of the paper is organized as follows: In Section~\ref{sec:related} we discuss related work. Section~\ref{sec:model} introduces the system model, interference assumptions and preliminaries. In Section~\ref{sec:auxiliary}, we define the auxiliary graph and demonstrate how feasible schedules of sectorized networks correspond to its matchings. Section~\ref{sec:flow-extension} presents our sectorization metrics, including flow extension ratio and sectorization gain. Section~\ref{sec:optimization} formulates the sectorization optimization and proposes a distributed approximation algorithm for the case of one-shot optimization decision. We then show in Section~\ref{sec:bipartite} how a special class of sectorization can induce bipartite auxiliary graphs, simplifying the analysis and accelerating stabilization. Section~\ref{sec:routing} discusses extensions to dynamic sectorization, including heuristic and throughput optimal routing approaches. Section~\ref{sec:evaluation} evaluates the proposed algorithms, Section~\ref{sec:future_directions} explores future directions and Section~\ref{sec:conclusion} concludes the paper.
\section{Related Work}
\label{sec:related}

There has been extensive work on characterizing the capacity region of {sub-6}\thinspace{GHz} wireless networks where each node is equipped with a single directional antenna (i.e., \emph{without sectorization}), as well as on developing medium access control (MAC), scheduling, and routing algorithms for these directional networks~\cite{korakis2003mac, Zhang_Xu_Wang_Guizani_2010, Yu_Teng_Bai_Xuan_Jia_2014}. Recently work also focused on mmWave networks where nodes apply beamforming techniques for directional communication, and considered multi-user MIMO and joint transmission~\cite{zhang2018mmchoir}, IAB~\cite{polese2020integrated}, joint scheduling and congestion control~\cite{garcia2015analysis}, and the corresponding scheduling/routing and resource allocation problems in these networks.
For networks \emph{with sectorization}, recent work has considered the design of routing protocols when only a single sector can be activated at any time for each node (e.g.,~\cite{Roy_Saha_Bandyopadhyay_Ueda_Tanaka}).

Most relevant to our work are~\cite{dai2006efficient, gupta2020learning, arribas2019optimizing}. In particular,~\cite{dai2006efficient} focuses on efficient message broadcasting in multi-hop sectorized wireless networks, where each node has a pre-fixed sectorization. In \cite{gupta2020learning} the authors consider the multi-hop link scheduling problem in self-backhauled mmWave cellular networks and applied deep reinforcement learning for minimizing the end-to-end delay. Moreover, \cite{arribas2019optimizing} considers the relay optimization problem between macro and micro base stations in mmWave backhaul networks with at most two-hop path lengths. 
% \cite{garcia2015analysis} considers joint scheduling and congestion control in multi-hop mmWave networks with different interference models using the framework of network utility maximization. However, it does not consider sectorized nodes.
%
In contrast, our work uniquely focuses on
(\emph{i}) characterizing the fundamental capacity of sectorized wireless networks when multiple sectors can be simultaneously activated at each node,
(\emph{ii}) optimizing the node sectorization in these networks under different network flow conditions, and
(\emph{iii}) analyzing the network-level gain introduced by optimizing the sectorization of each node, which has not been considered in prior work. To the best of our knowledge, this paper is \emph{the first thorough study of these topics}.

\section{Model and Preliminaries}
\label{sec:model}

% \mynote{See nomenclature.tex for defined macro for symbols, for example, $\numNode$ is the number of nodes, $\setNode$ is the set of nodes, etc.}

% \iffullpaper
% \input{tex/notation}
% \fi

In this section, we present the sectorized network model and the corresponding interference constraints and capacity region.

% \iffullpaper
% The notation is summarized in Table~\ref{tab:nomenclature}.
% \fi

% \tingjun{We do not considered the case where the nodes are with different numbers of sectors, right?}\panos{We actually do. For presentation purposes we assume that every node has the same number of sectors. However, its an obvious extension how to optimize over different number of sectors. Or even output what is the minimum number of sectors per node to include a flow vector. About the latter I will put a discussion part to say it.}

% The notation is summarized in Table~\ref{tab:nomenclature}.

%% figure begins
\begin{figure}[!t]
\centering
\subfloat[]{
\includegraphics[height=1.25in]{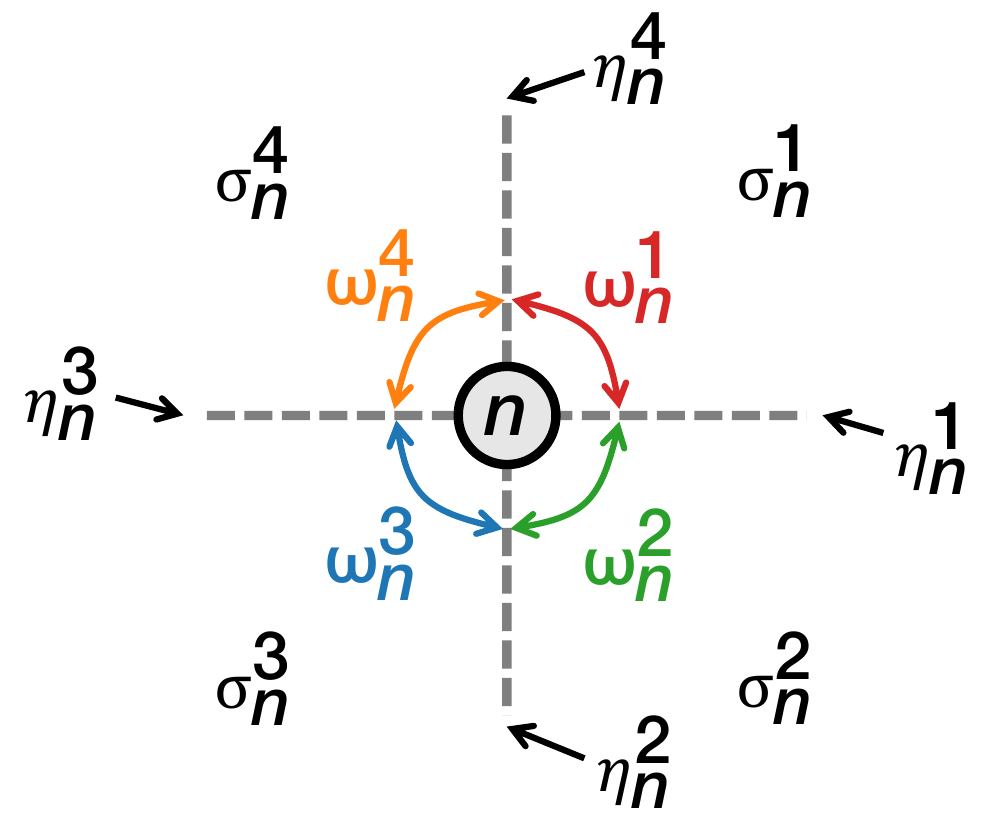}
\label{fig:model-node-a}}
\hspace{12pt}
\subfloat[]{
\includegraphics[height=1.25in]{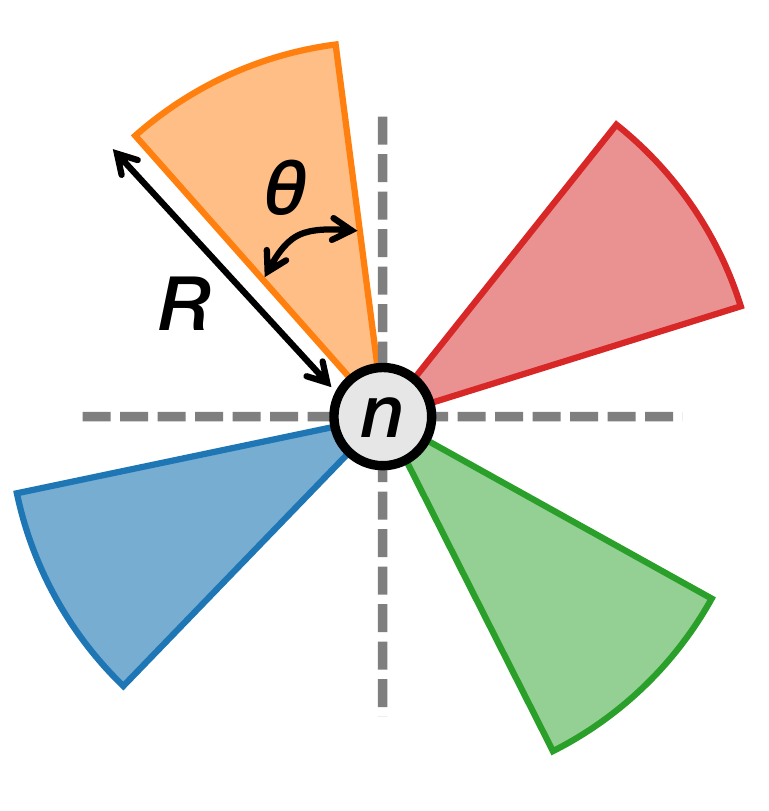}
\label{fig:model-node-b}}
% \vspace{-3mm}
\caption{Sectorized infrastructure node model:
(a) A sectorized node $n$ with $\numSector_{n} = 4$ sectors, $\{\nodeSector{n}{k}\}$, the corresponding field of view (FoV) of each sector, $\{\nodeFov{n}{k}\}$, and the sectoring axes, $\{\nodeSectorAxis{n}{k}\}$.
(b) Each node sector can perform transmit (TX) or receive (RX) beamforming with a range of $\nodeRange$ and main lobe beamwidth of $\nodeBeamwidth$.}
\label{fig:model-node}
% \vspace{-3mm}
\end{figure}
%% figure ends

%%%%%
%%%%%
\subsection{Network Model}
\label{ssec:model-network}

%%%%%
%%%%%
% \subsubsection{Sectorized mmWave Node Model}

% We consider \emph{sectorized} mmWave infrastructure nodes. In particular, a mmWave node $n \in \setNode$ is equipped with $\numSector_{n}$ sectors and we denote its set of sectors by $\setSector_{n} = \{1, 2, \cdots, \numSector_{n}\}$. Let $\setSector$ denote the set of all possible sectorization for a node. The $k$-th sector of node $n$, denoted by $\nodeSector{n}{k},\ \forall k \in \setSector_{n}$, has a field of view (FoV) of $\nodeFov{n}{k}$. The $\numSector_{n}$ sectors of node $n$ combine to cover an FoV of the entire 360$^\circ$ in the azimuth plane, i.e., $\sum_{k=1}^{\numSector} \nodeFov{n}{k} = 360^{\circ}$.

% \vspace{0.5ex}
% \noindent\textbf{Sectorized mmWave Node Model}:

We consider a network consisting of $\numNode$ \emph{sectorized} infrastructure nodes. We denote the set of nodes by $\setNode$ and index them by $[n] = \{1, 2, \cdots, \numNode\}$. In particular, let $\sector_{n}$ denote the sectorization of node $n \in \setNode$ equipped with $\numSector_{n}$ sectors. Let $\setSector_{n}(\numSector_{n})$ denote the set of all possible sectorizations for node $n$ with a fixed number of $\numSector_{n}$ sectors. As shown in Fig.~\ref{fig:model-node}\subref{fig:model-node-a}, the $k^{\textrm{th}}$ sector of node $n$, denoted by $\nodeSector{n}{k}$ ($k = 1, \cdots, \numSector_{n}$), has a field of view\footnote{We define the field of view of a sector as the angular coverage in the azimuth plane that this sector can transmit to or receive from.} (FoV) of $\nodeFov{n}{k}$, and the $\numSector_{n}$ sectors of node $n$ combine to cover an FoV of the entire azimuth plane, i.e., $\sum_{k=1}^{\numSector_{n}} \nodeFov{n}{k} = 360^{\circ}$.
% For two adjacent sectors $k$ and $(k+1)$ of node $n$, we call their boundary \emph{a sectoring axis} and denote it by $\nodeSectorAxis{n}{k}$.
For two adjacent sectors $k$ and $(k+1)$ of node $n$ ($k = 1,\cdots,K-1$), we call their boundary \emph{a sectoring axis} and denote it by $\nodeSectorAxis{n}{k}$. We define the sectoring axis between sector $K$ and sector $1$ with $\nodeSectorAxis{n}{K}$.
Let $\sectorVec = (\sector_{1}, \cdots, \sector_{\numNode}) = [\sector_{n}: \forall n \in \setNode]$ denote the \emph{network sectorization}, and $\numSectorVec = (\numSector_{1}, \cdots, \numSector_{\numNode}) = [\numSector_{n}: \forall n \in \setNode]$ be the vector of the number of sectors for all nodes. For a given $\numSectorVec \in \mathbb{Z}_{+}^{\numNode}$, let $\setSector(\numSectorVec)$ be the set of all possible network sectorizations, where node $n$ is equipped with $\numSector_{n}$ sectors.

Each sector of an infrastructure node is equipped with a half-duplex phased array antenna to perform transmit (TX) or receive (RX) beamforming. We adapt the sectored antenna model~\cite{bai2014coverage} to approximate the TX/RX beam pattern that can be formed by each sector, where $\nodeRange$ is the TX/RX range and $\nodeBeamwidth$ is the beamwidth of the main lobe, as depicted in Fig.~\ref{fig:model-node}\subref{fig:model-node-b}. We assume that the beamwidth $\nodeBeamwidth$ is smaller than the sector's FoV, and the TX/RX beam can be \emph{steered} to be pointed in different directions within each sector.

%%%%%
%%%%%
% \subsubsection{Directional mmWave Link Model}

% Before introducing our network model and interference constraints, we first describe how a \emph{directional} link can be established between two sectorized mmWave nodes $n$ and $n'$ located at $\nodeLoc{n}$ and $\nodeLoc{n'}$, respectively. We use $\directLink{n}{k}{n'}{k'}$ to denote a directional communication link from $\nodeSector{n}{k}$ (the \emph{TX sector}) to $\nodeSector{n'}{k'}$ (the \emph{RX sector}). In particular, $\directLink{n}{k}{n'}{k'}$ is a \emph{feasible link} if and only if the following two conditions are satisfied:
% %%
% \begin{itemize}[leftmargin=*]
% \item
% (\textbf{C1}) Node $n$ lies in the FoV $\nodeFov{n'}{k'}$ of node $n'$, and node $n'$ lies in the FoV $\nodeFov{n}{k}$ of node $n$; and
% \item
% (\textbf{C2}) The distance between the two nodes is less than the sum of their communication range, i.e., $|\nodeLoc{n} - \nodeLoc{n'}| \leq 2\nodeRange$, where $|\mathbf{x}|$ denotes the $L_{2}$-norm of vector $\mathbf{x}$.
% \end{itemize}
% %%
% Since both the TX and RX beams can be steered within $\nodeFov{n}{k}$ of node $n$ and $\nodeFov{n'}{k'}$ of node $n'$, respectively, the above two conditions are sufficient for establishing the directional link $\directLink{n}{k}{n'}{k'}$. It is easy to see that if $\directLink{n}{k}{n'}{k'}$ is a directional link, $\directLink{n'}{k'}{n}{k}$ is also a directional link due to symmetry.

%% figure begins
\begin{figure}[!t]
\centering
\subfloat[$\graphNet = (\setNode, \setLink^{\sector})$]{
\includegraphics[height=1.35in]{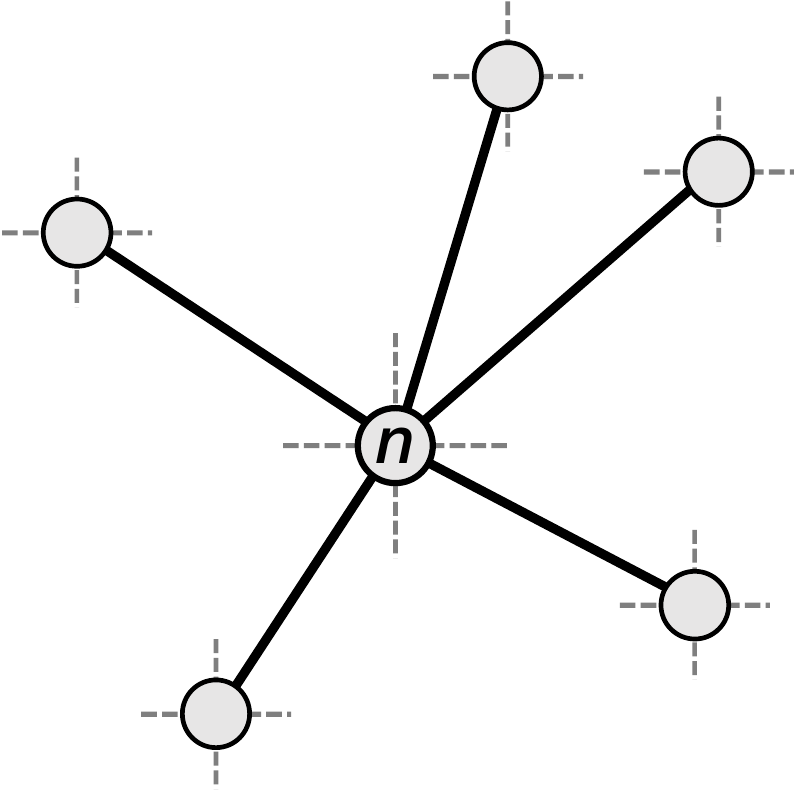}
\label{fig:model-network-a}}
\hspace{12pt}
\subfloat[$\graphAux = (\setVertex^{\sector}, \setEdge^{\sector})$]{
\includegraphics[height=1.35in]{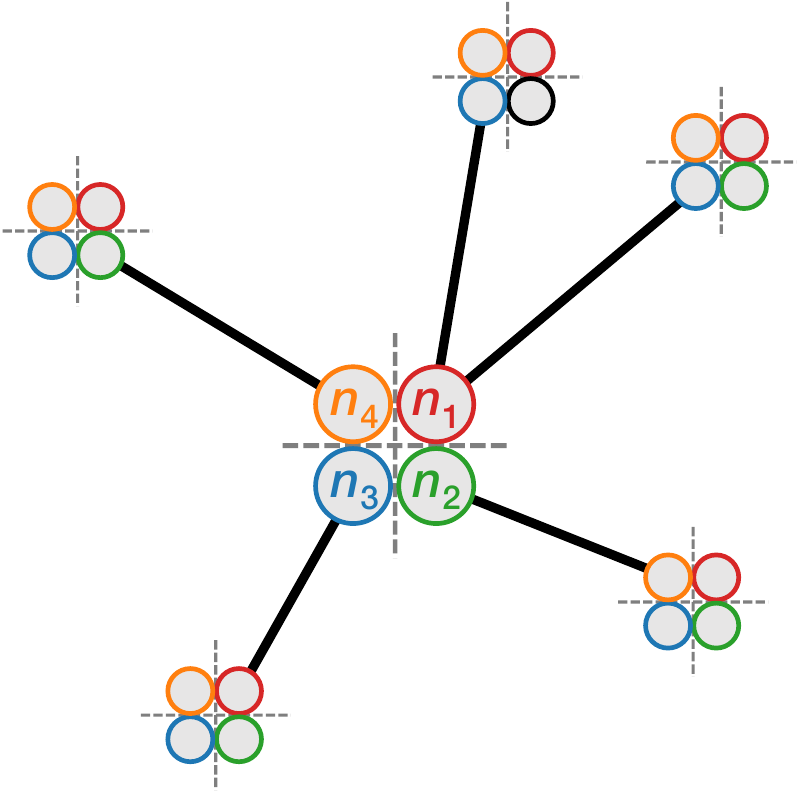}
\label{fig:model-network-b}}
\vspace{-1mm}
\caption{Graph representations of a sectorized network: (a) the connectivity graph of the physical network, $\graphNet = (\setNode, \setLink^{\sector})$, and (b) its corresponding auxiliary graph, $\graphAux = (\setVertex^{\sector}, \setEdge^{\sector})$.}
\label{fig:model-network}
\vspace{-2mm}
\end{figure}
%% figure ends

%%%%%
%%%%%
% \subsubsection{Graph Representation of Sectorized mmWave Networks}

% \vspace{0.5ex}
% \noindent\textbf{Graph Representation of Sectorized mmWave Networks}:
We use a \emph{directed} graph $\graphNet = (\setNode, \setLink)$ to denote the \emph{connectivity graph} of the network under sectorization $\sector \in \setSector (\numSectorVec)$, where $\setNode$ (with $|\setNode| = \numNode$) is the set of nodes and $\setLink$ (with $|\setLink| = \numLink$) is the set of directed links. A directed link $\link$ \emph{from} node $n$ \emph{to} node $n'$, denoted by $\link = (n, n')$, exists if the distance between the two nodes is less than the sum of their communication range, i.e., $|\nodeLoc{n} - \nodeLoc{n'}| \leq 2\nodeRange$, where $\nodeLoc{x}$ denotes the node's location vector (e.g., using the Cartesian coordinate system) and $|\cdot|$ denotes the $L_{2}$-norm. Without loss of generality, we use $\graph$ (without the superscript $\sector$) to denote the connectivity graph of an unsectorized network. 

We also present an equivalent representation of each directed link in $\graphNet$ based on the node sectors. In particular, each $\link \in \setLink$ can also be represented by $\link = \directLink{n}{k}{n'}{k'}$, where the $k^{\textrm{th}}$ sector of node $n$ (i.e., $\nodeSector{n}{k}$) is the \emph{TX sector}, and the $k'^{\textrm{th}}$ sector of node $n'$ (i.e., $\nodeSector{n'}{k'}$) is the \emph{RX sector}. For a link in the form of $\directLink{n}{k}{n'}{k'}$ to be a \emph{feasible} link, it needs to satisfy the following two conditions\footnote{A similar interference model was presented in~\cite{yuan2020optimal}, where beamforming can reduce the interference in mmWave networks. Our framework can also be generalized to other types of networks using their corresponding connectivity graphs.} (see Fig.~\ref{fig:model-link}\subref{fig:model-link-a}):

\begin{itemize}[leftmargin=*, topsep=0pt]
\item
(\textbf{C1}) The distance between the two nodes is less than the sum of their communication range, i.e., $|\nodeLoc{n} - \nodeLoc{n'}| \leq 2\nodeRange$; and
\item
(\textbf{C2}) Node $n$ lies in the FoV $\nodeFov{n'}{k'}$ of node $n'$, and node $n'$ lies in the FoV $\nodeFov{n}{k}$ of node $n$.
\end{itemize}

Since both the TX and RX beams can be steered within $\nodeFov{n}{k}$ of node $n$ and $\nodeFov{n'}{k'}$ of node $n'$, respectively, the above two conditions are sufficient for establishing $\link = \directLink{n}{k}{n'}{k'}$. Moreover, if $\directLink{n}{k}{n'}{k'}$ is a directional link, $\directLink{n'}{k'}{n}{k}$ is also a directional link due to symmetry.
Therefore, the set of feasible directed edges in $\graphNet$ is given by:
\begin{align*}
\setLink
& = \big\{ (n,n'): \forall n,n' \in \setNode, n \ne n', ~\textrm{s.t. (\textbf{C1}) is satisfied} \big\},\ \textrm{or} \\
\setLink^{\sector}
& = \big\{ \directLink{n}{k}{n'}{k'}: \forall n,n' \in \setNode, n \ne n', k \in [\numSector_{n}], k' \in [ \numSector_{n'}], \nonumber \\
& \qquad\qquad\qquad \textrm{s.t. (\textbf{C1}) and (\textbf{C2}) are satisfied} \big\}.
\end{align*}
Although $\setLink$ and $\setLink^{\sector}$ are identical with respect to $\graphNet$, for clarity, we use $\setLink^{\sector}$ with superscript $\sector$ to indicate the directed links represented by node sectors. Moreover, let $\setLink_{n}^{+}, \setLink_{n}^{-} \subseteq \setLink^{\sector}$ denote the set of (directed) outgoing and incoming links with end point of node $n$.

% Based on the sectorized node model and directional link model, we use a \emph{directed} graph $G = (V, E)$ to represent the considered sectorized mmWave network, where $V = \setNode$ (with $|V| = \numNode$) is the set of nodes, and $E$ is the set of feasible directional links, i.e.,
% %%
% \begin{align}
% E = \{\directLink{n}{k}{n'}{k'}\}\ \textrm{with}~\directLink{n}{k}{n'}{k'}~\textrm{satisfying (\textbf{C1}) and (\textbf{C2})} \nonumber \\
% \textrm{for}\ \forall n,n' \in \setNode, k \in \setSector_{n},\ \textrm{and}~ k' \in \setSector_{n'}.
% \end{align}
% %%

\vspace{0.5ex}
\noindent\textbf{Remark.}
Note that one of the main differences between the sectorized and traditional unsectorized networks is that \emph{each node $n \in \setNode$ can have multiple links being activated at the same time, at most one per sector}. As a result, a number of links in $\setLink^{\sector}$ can share the same end point (node) in $\setNode$ while being activated simultaneously. We use the terms ``node'' and ``link'' in reference to the \emph{connectivity graph}, $\graphNet$, while reserving the terms ``vertex'' and ``edge'' for the \emph{auxiliary graph} of $\graphNet$, which will be presented in Section~\ref{sec:auxiliary}.

%%%%%
%%%%%
\subsection{Interference Model}
\label{ssec:model-interference}

%% figure begins
\begin{figure}[!t]
\centering
\subfloat[]{
\includegraphics[height=1.1in]{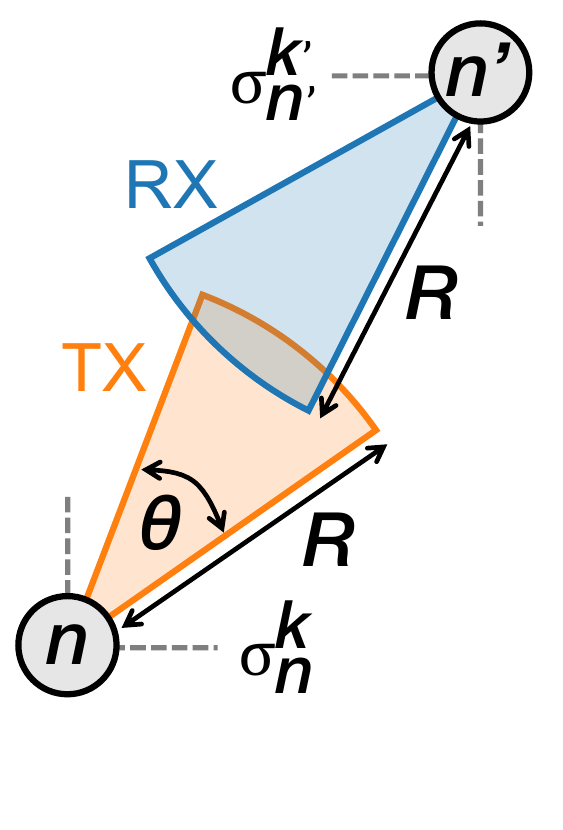}
\label{fig:model-link-a}}
\hfill
\subfloat[]{
\includegraphics[height=1.1in]{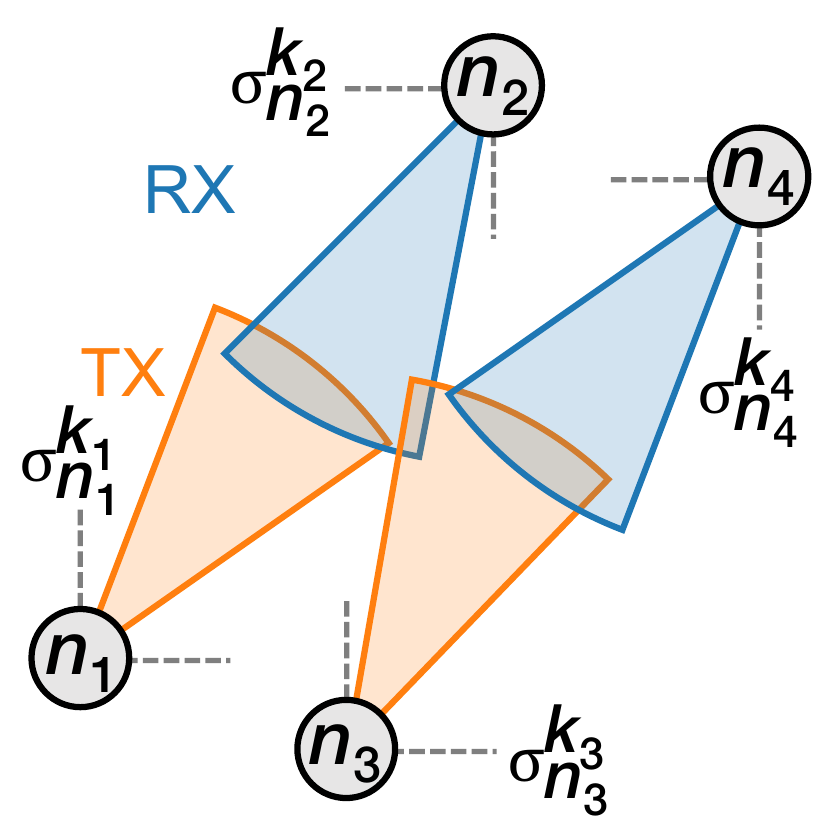}
\label{fig:model-link-b}}
\hfill
\subfloat[]{
\includegraphics[height=1.1in]{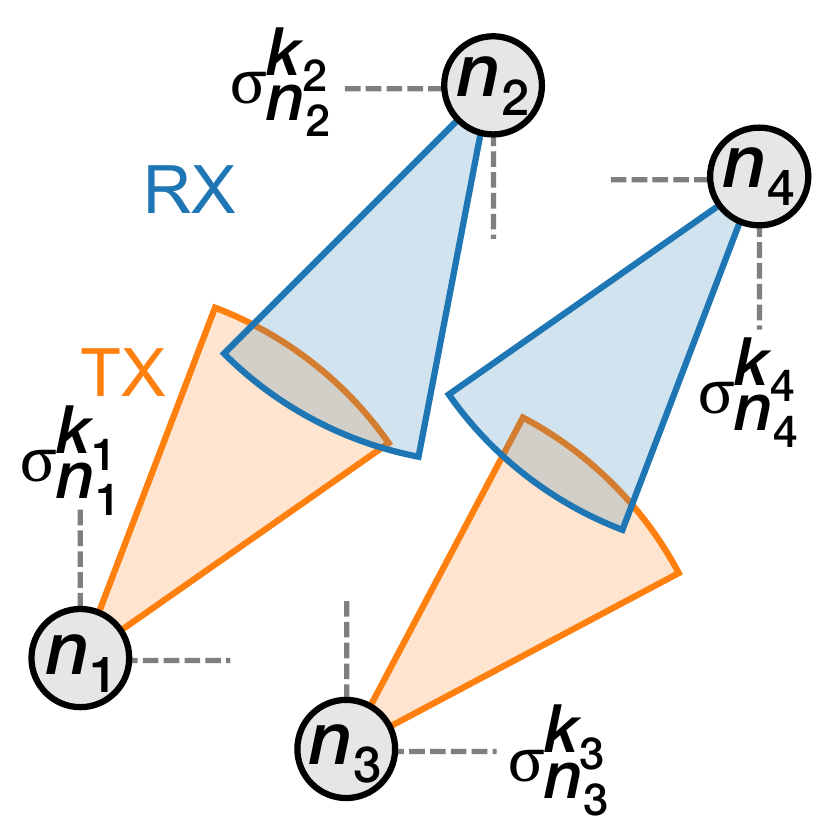}
\label{fig:model-link-c}}
\vspace{-1mm}
\caption{(a) A feasible link $\link = (\nodeSector{n}{k}, \nodeSector{n'}{k'})$,
(b) A secondary interfering TX sector $\nodeSector{n_3}{k_3}$ to RX sector $\nodeSector{n_2}{k_2}$, (c) The secondary TX interference in (b) can be avoided by steering node 3's TX beam in $\nodeSector{n_3}{k_3}$, so that both links $(\nodeSector{n_1}{k_1}, \nodeSector{n_2}{k_2})$ and $(\nodeSector{n_3}{k_3}, \nodeSector{n_4}{k_4})$ can be simultaneously activated.}
\label{fig:model-link}
\vspace{-2mm}
\end{figure}
%% figure ends

The link interference model is essential for determining the set of directional links that can be activated simultaneously, or the \emph{feasible schedules}. Below, we describe our link interference model based on the protocol model~\cite{gupta2000capacity} adapted to the considered sectorized networks.

%% definition begins
\begin{definition}[Primary Interference Constraints in Sectorized Networks]
\label{def:primary-interference}
The transmission on a feasible link $\directLink{n}{k}{n'}{k'} \in \setLink^{\sector}$ from the $k^{\textrm{th}}$ sector of node $n$ to the $k'^{\textrm{th}}$ sector of node $n'$ is successful if it does not overlap with any other feasible directional link $\directLink{n}{k}{m}{j}$ or $\directLink{m}{j}{n'}{k'}$ that share a TX or RX sector in common with $\directLink{n}{k}{n'}{k'}$. Essentially, at any time, at most one outgoing or incoming link is allowed in each node sector $\nodeSector{n}{k}, \forall n \in \setNode, k \in [ \numSector_{n}]$.
\end{definition}
%% definition ends

% %% definition begins
% \begin{definition}[Secondary Interference Constraints in Sectorized mmWave Networks]
% Consider two feasible directional links $\link_{12} = \directLink{n_1}{k_1}{n_2}{k_2}$ and $\link_{34} = \directLink{n_3}{k_3}{n_4}{k_4}$ between four distinct nodes $n_{i}$ ($i = 1, 2, 3, 4$), and let $\angleTwoLink{\link}{\link'}$ denote the angle between two links $\link$ and $\link'$.
% If the directional link $\link_{32} = \directLink{n_3}{k_3}{n_2}{k_2}$ is also a feasible link with $\angleTwoLink{\link_{12}}{\link_{32}} < \nodeBeamwidth$, then $\nodeSector{n_3}{k_3}$ is a secondary interfering TX sector to the RX sector $\nodeSector{n_2}{k_2}$.
% Symmetrically, if the directional link $\link_{14} = \directLink{n_1}{k_1}{n_4}{k_4}$ is also a feasible link with $\angleTwoLink{\link_{12}}{\link_{14}} < \nodeBeamwidth$, then $\nodeSector{n_1}{k_1}$ is a secondary interfering TX sector to the RX sector $\nodeSector{n_4}{k_4}$.
% \end{definition}
% %% definition ends

%% definition begins
\begin{definition}[Secondary Interference Constraints in Sectorized Networks]
Consider two feasible directional links $\link_{12} = \directLink{n_1}{k_1}{n_2}{k_2}$ and $\link_{34} = \directLink{n_3}{k_3}{n_4}{k_4}$ between four distinct nodes $n_{i}$ ($i = 1, 2, 3, 4$) with fixed beamforming directions in each TX/RX sector. If the directional link $\link_{32} = \directLink{n_3}{k_3}{n_2}{k_2}$ is also a feasible link and the TX beam in $\nodeSector{n_3}{k_3}$, which is intended to communicate with the RX beam in $\nodeSector{n_4}{k_4}$, overlaps with the RX beam in $\nodeSector{n_2}{k_2}$, then $\nodeSector{n_3}{k_3}$ is a secondary interfering TX sector to the RX sector $\nodeSector{n_2}{k_2}$ (see Fig.~\ref{fig:model-link}\subref{fig:model-link-b}).
\end{definition}
%% definition ends

% \panos{it is obvious but should we maybe underline that we assume nodes that are points?}.

% \tingjun{Mention that the secondary interference constraint is particularly interesting when we have the beam steering capability.}

% \tingjun{Mention that we will start ``ignoring" the secondary interference constraint.} \panos{Done}
% \tingjun{Talk about $\theta_{\min}$.} \panos{Done}

%% figure begins
% \begin{figure}[!t]
% \centering
% \subfloat[]{
% \includegraphics[width=0.48\columnwidth]{./figs/thetamincdfR003.eps}
% \label{fig:thetamin3}}
% \hspace{-6pt}
% \subfloat[]{
% \includegraphics[width=0.48\columnwidth]{./figs/thetamincdfR005.eps}
% \label{fig:thetamin5}}
% \vspace{-\baselineskip}
% \caption{CDF of $\theta_{min}$ calculated for $10000$ random networks with $100$ nodes and when the nodes have beamwidth range that reaches (a) $3 \%$ of the space ($R = 0.03$) and (b) $5 \%$ of the space ($R = 0.05$).}
% \label{fig:thetamin}
% \end{figure}
%% figure ends

In this paper, we consider only primary interference constraints in sectorized networks, which is a realistic assumption because of two main reasons.
First, note that secondary interference constraints can be avoided using the beam steering capacity of an infrastructure node. For the example depicted in Fig.~\ref{fig:model-link}\subref{fig:model-link-b}, the interference from TX sector $\nodeSector{n_3}{k_3}$ to RX sector $\nodeSector{n_2}{k_2}$ can be avoided by steering the TX beam of node 3 in $\nodeSector{n_3}{k_3}$, so that both links $(\nodeSector{n_1}{k_1}, \nodeSector{n_2}{k_2})$ and $(\nodeSector{n_3}{k_3}, \nodeSector{n_4}{k_4})$ can be activated simultaneously, as shown in Fig.~\ref{fig:model-link}\subref{fig:model-link-c}.
Second, it is intuitive to note that for sufficiently small values of the beamwidth, $\theta$, the network becomes highly directional with ``pencil'' beams. Essentially, there exists a threshold for $\theta$, denoted by $\theta_{\textrm{th}}$, below which the secondary interference constraints can be completely eliminated regardless of the TX/RX beamforming directions at each node.
% A straightforward upper bound for that beamwidth is $\theta^{\textrm{min}}$, which is defined as $\theta^{\textrm{min}} = \min_{n \in \setNode} \theta_{n}^{\textrm{min}}$, where $\theta_{n}^{\textrm{min}}, \forall n$ denotes the minimum angle between adjacent links with an end point of $n$.
For each node $n$, let $\theta_{n}^{\min}$ denote the minimum angle between adjacent links with $n$ being an endpoint. Then, we have $\theta_{\textrm{th}} = \min_{n \in \setNode} \{ \theta_{n}^{\min} \}$.
In the example network shown in Fig.~\ref{fig:example-network} (see Section~\ref{ssec:evaluation-example-network} for the detailed setup), $\theta_{\textrm{th}} = 15.8^{\circ}$, which can be achieved by state-of-the-art phased arrays~\cite{sadhu201728}.
Fig.~\ref{fig:thetamin} illustrates the cumulative distribution function (CDF) of $\theta_{\textrm{th}}$ calculated over 1,000 random networks (see Section~\ref{ssec:evaluation-random-network} for the detailed setup) with $\numNode \in \{20,40,60\}$ nodes deployed in a unit square area and when a communication range of $2R = 0.1$. The median value for $\theta_{\textrm{th}}$ is $107.0^{\circ}/6.7^{\circ}/2^{\circ}$ for $N = 20/40/60$, respectively. This illustrates that in certain scenarios, it is reasonable to remove the secondary interference constraint under realistic node density and beamwidth.
Note that a single sector can include multiple links, but at most one link can be activated in any time slot. Under our no-secondary-interference assumption, we assume the beam can be aligned with whichever link is scheduled within that sector.
\subsection{Traffic Model, Schedule, and Queues}
\label{ssec:model-schedule}

\begin{figure}[!t]
\centering
\includegraphics[height=1.5in]{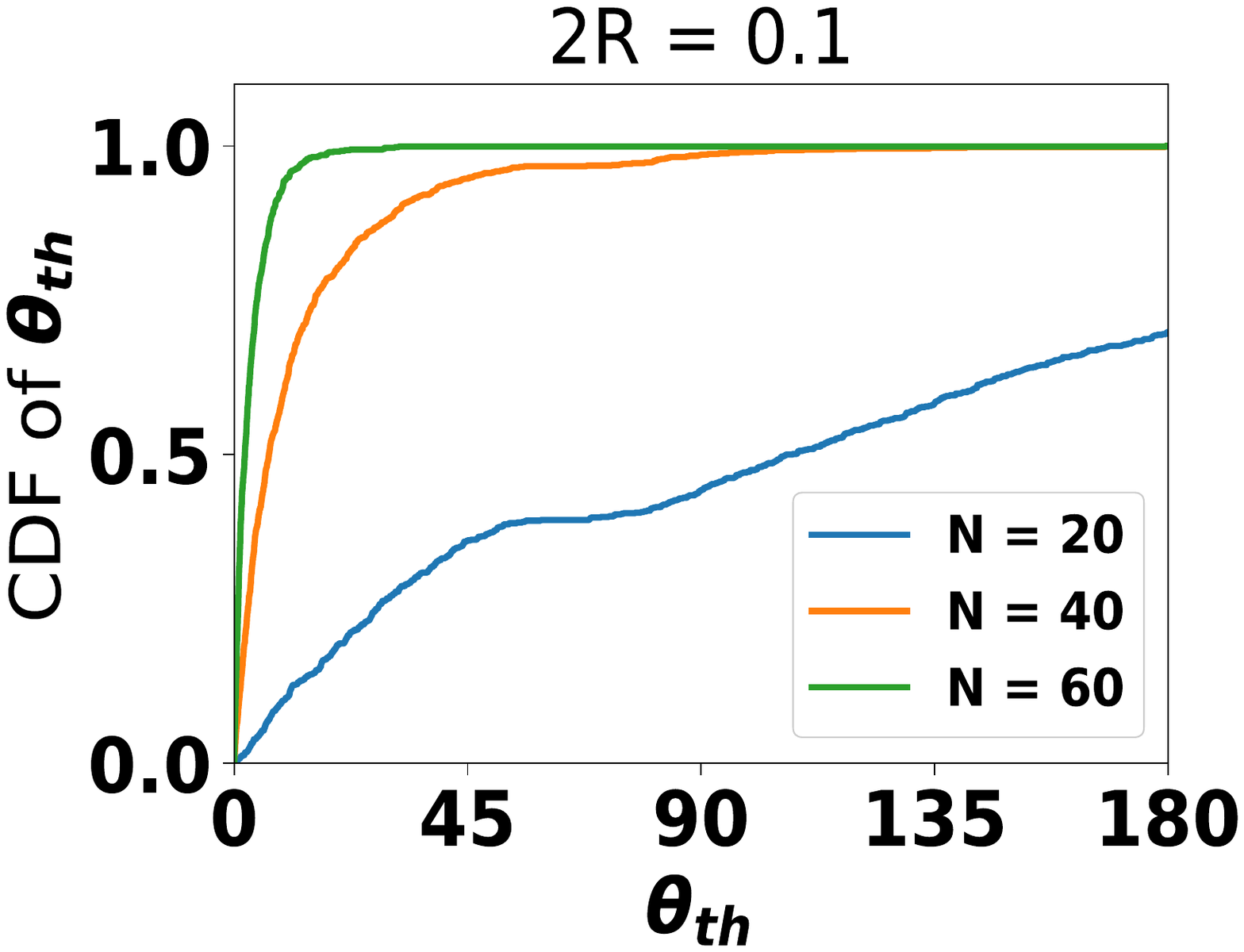}
\vspace{-1mm}
\caption{The CDF of $\theta_{\textrm{th}}$ calculated for 1,000 random networks with $N \in \{20,40,60\}$ and a communication range of $2R = 0.1$.}
\label{fig:thetamin}
\vspace{-2mm}
\end{figure}
%% figure ends

We assume that time is slotted and packets arrive at each node according to some stochastic process. For convenience, we classify all packets passing through the network to belong to a particular \emph{commodity}, $c \in \setNode$, which represents the destination node of each packet. Let $\arrival_{n}^{(c)}(t) \leq \arrival_{\textrm{max}} < +\infty$ be the number of commodity-$c$ packets entering the network at node $n$ and destined for node $c$ in slot $t$.
% A packet may enter the network at any node and although its destination can be a subset of the nodes, without loss of generality, we assume that each packet has a single target destination node.
The packet arrival process $\arrival_{n}^{(c)}(t)$ is assumed to have a well-defined long-term rate of $\arrivalRate_{n}^{(c)} = \lim_{T \to +\infty} \frac{1}{T} \sum_{t=1}^{T} \arrival_{n}^{(c)}(t)$. Let $\arrivalRateMat = [\arrivalRate_{n}^{(c)}]$ be the $\numNode \times \numNode$ multi-commodity arrival rate matrix.

All the (directional) links have a capacity of one packet per time slot\footnote{For simplicity, our analysis uses unit-capacity links, but we expect the same framework to extend to heterogeneous capacities by applying each link’s capacity as a multiplicative factor in both the feasible flow region and the max-weight scheduling (cf. \cite[Section 4.3]{georgiadis2006resource}).}. A \emph{schedule} in any time slot $t$ is represented by a vector $\scheduleVec(t) = [\scheduleLink{\link}(t)] \in \{0,1\}^{\numLink}$, in which $\scheduleLink{\link}(t) = 1$ if link $\link$ is scheduled to transmit a packet in time slot $t$ and $\scheduleLink{\link}(t) = 0$ otherwise. We denote the set of feasible schedules in $\graphNet$ by $\scheduleSet_{\graphNet}$. In addition, let $\mu_{\link}^{(c)}(t) = 1$ if link $\link$ serves a commodity-$c$ packet in time $t$ (determined by the scheduling and routing algorithm), and $\mu_{\link}^{(c)}(t) = 0$ otherwise.
We define $\queue_{n}^{(c)}(t)$ as the number of commodity-$c$ packets at node $n$ in time slot $t$, or the \emph{queue backlog}. Choosing a schedule $\scheduleVec(t) \in \scheduleSet_{\graphNet}$ and let $[x]^{+} = \max\{0, x\}$, the queue dynamics are described by:
\begin{align*}
\queue_{n}^{(c)}(t) = \big[ \queue_{n}^{(c)}(t-1) - \littlesum\nolimits_{\link \in \setLink_{n}^{+}} \schedule_{\link}(t) \cdot \mu_{\link}^{(c)}(t) \big]^{+} \\
+ \littlesum\nolimits_{\link \in \setLink_{n}^{-}} \schedule_{\link}(t) \cdot \mu_{\link}^{(c)}(t) + \arrival_{n}^{(c)}(t), \forall t.
\end{align*}
Let $\queueVec(t) = [\queue_{n}^{(c)}(t): n,c \in \setNode]$ denote the queue vector.

% We call the vector $\flowVec_{\arrivalRateVec} = (\flow_{ej}: \forall e \in E, j \in V) \in \mathbb{R}^{|V|\cdot|E|}$ an $\arrivalRateVec$-\emph{admissible multicommodity flow vector} if it satisfies the flow conservation equations under arrival rate vector of $\arrivalRateVec$. Let $\flowSet_{\arrivalRateVec}$ be the set of all $\arrivalRateVec$-admissible flow vectors.

%%%%%
\subsection{Capacity Region \& Throughput Optimality}
\label{ssec:model-capacity}

% \tingjun{Flows are determined by the scheduling/routing policy so they need to be defined beforehand.}

A dynamic scheduling and routing algorithm will determine the schedule $\scheduleVec(t)$ in each time $t$. Let $\flow_{\link}^{(c)}$ denote the long-term rate at which commodity-$c$ packets are served by link $\link$, or the commodity-$c$ \emph{flow} on link $\link$. We define $\flowVec = [\flow_{\link}: \link \in \setLink^{\sector}]$ as the \emph{network flow vector}, where $\flow_{\ell} = \sum_{c \in \setNode} f_{\link}^{(c)}$ is the total flow served by link $\link$.
The capacity region of the considered sectorized network $\graphNet$, denoted by $\capRegion(\graphNet)$, is defined as the set of all arrival rate matrices, $\arrivalRateMat$, for which there exists a multi-commodity network flow vector, $\flowVec$, satisfying the flow conservation equations given by:
\begin{align}
\arrivalRate_{nc} = \littlesum_{\link \in \setLink_{n}^{+}} \flow_{\link}^{(c)} - \littlesum_{\link \in \setLink_{n}^{-}} \flow_{\link}^{(c)},\ &~ \forall n \in \setNode ~\textrm{and}~ n \ne c, \label{eqn:flow-conservation-1} \\
\littlesum_{n \in \setNode} \arrivalRate_{nc} = \littlesum_{\link \in \setLink_{c}^{-}} \flow_{\link}^{(c)},\ &~ \forall c \in \setNode, \label{eqn:flow-conservation-2} \\
\flow_{\link}^{(c)} \geq 0,\ &~ \forall \link \in \setLink^{\sector}, c \in \setNode, \label{eqn:flow-conservation-3} \\
\flow_{\link} = \littlesum_{c \in \setNode} \flow_{\link}^{(c)} \leq 1,\ &~ \forall \link \in \setLink^{\sector}. \label{eqn:flow-conservation-4}
\end{align}
In particular, {\eqref{eqn:flow-conservation-1}}--{\eqref{eqn:flow-conservation-3}} define a \emph{feasible} routing for commodity-$c$ packets, and {\eqref{eqn:flow-conservation-4}} indicates that the total flow on each edge should not exceed its capacity.
Therefore, for $\graphNet = (\setNode, \setLink^{\sector})$, an arrival rate matrix $\arrivalRateMat$ is in the capacity region $\capRegion(\graphNet)$ if there exists a \emph{feasible} multi-commodity flow vector supporting $\arrivalRateMat$ with respect to the network defined by $\graphNet$. As a result,
\begin{align}
\capRegion(\graphNet) = \big\{ ~\arrivalRateMat: \exists \flowVec \in \convexHull(\scheduleSet_{\graphNet}) ~\textrm{s.t. {\eqref{eqn:flow-conservation-1}}--{\eqref{eqn:flow-conservation-4}} are satisfied}~ \big\},
\label{eqn:cap-region}
\end{align}
where $\convexHull(\cdot)$ is the convex hull operator.

A scheduling and routing algorithm is called \emph{throughput-optimal} if it can keep the network queues stable for all arrival rate matrices $\arrivalRateMat \in \textrm{int}(\capRegion(\graphNet))$, where $\textrm{int}(\capRegion(\graphNet))$ denotes the interior of $\capRegion(\graphNet)$.
A well-known throughput-optimal algorithm is the dynamic \emph{backpressure routing algorithm}~ \cite{tassiulas1990stability}, which works as follows. For each link $\link = (n,m) \in \setLink^{\sector}$, we define its \emph{backpressure} in time slot $t$ as $\bp_{\link}(t) = \max_{c \in \setNode} \{ \queue_{n}^{(c)}(t) - \queue_{m}^{(c)}(t) \}$. Let $\bpVec(t) = [\bp_{\link}(t): \forall \link \in \setLink^{\sector}]$. In every time slot $t$, the backpressure algorithm selects $\scheduleVec^{\textrm{BP}}(t)$ as follows:
\begin{align}
\scheduleVec^{\textrm{BP}}(t) \in \arg \max\nolimits_{\scheduleVec \in \scheduleSet_{\graphNet}}\ \{ \bpVec^{\top}(t) \cdot \scheduleVec \},
\label{eqn:backpressure}
\end{align}
together with the commodity to be served on each link $\link$. Given $\arrivalRateMat \in \textrm{int}(\capRegion(\graphNet))$, the backpressure algorithm returns a feasible network flow $\flowVec$ that supports $\arrivalRateMat$ in $\graphNet$. In the rest of the paper, although the target optimization problems may admit multiple optimal solutions, without loss of generality and for notational brevity, we treat them as singletons.

% \section{The Auxiliary Graph and Matching Polytopes}
\section{The Auxiliary Graph and Matchings}
\label{sec:auxiliary}

% \tingjun{Do the results in this part assume that all nodes have the same number of sectors? Also, do we need this assumption for the rest of the analysis/optimization?}

In this section, we introduce the auxiliary graph for a sectorized network, followed by an overview of matching polytopes and the definition of equivalent sectorizations. All of these serve as the foundations for the results presented in the remaining of this paper.

%%%%%
%%%%%
\subsection{The Auxiliary Graph, $\graphAux$}
\label{ssec:auxiliary-graph-construction}

% \vspace{0.5ex}
% \noindent\textbf{The Auxiliary Graph, $\graphAux$}:
To allow for analytical tractability and support the analysis, we introduce the \emph{auxiliary graph}, $\graphAux = (\setVertex^{\sector}, \setEdge^{\sector})$ of the considered sectorized network, under a given sectorization rule, $\sector \in \setSector(\numSectorVec)$. In particular, $\graphAux$ is generated based on the connectivity graph $\graphNet = (\setNode, \setLink^{\sector})$ as follows:
\begin{itemize}[leftmargin=*, topsep=3pt, itemsep=3pt]
\item
Each node $n \in V$ is duplicated $\numSector_{n}$ times into a set of \emph{vertices}, $\{n_{1}, \cdots, n_{\numSector_{n}}\} \in V^{\sector}$, one for each node sector $\nodeSector{n}{k},\ k \in [\numSector_n]$.
\item
Each directed link $\link = \directLink{n}{k}{n'}{k'} \in \setLink^{\sector}$ between the TX sector $\nodeSector{n}{k}$ and RX sector $\nodeSector{n'}{k'}$ is ``inherited" as a directed \emph{edge} $\edge$ from vertex $n_{k}$ to vertex $n'_{k'}$ in $\setVertex^{\sector}$.
\end{itemize}
Note that $|\setVertex^{\sector}| = \sum_{n \in \setNode} \numSector_{n}$ and $|\setEdge^{\sector}| = |\setLink| = \numLink$, and as we show in the rest of the paper, $\graphAux$ can facilitate the analysis and optimization in sectorized networks.
An illustrative example is shown in Fig.~\ref{fig:model-network}, where each node in $\graphNet$ has an equal number of 4 sectors based on the Cartesian coordinate system. Ntode $n$ in $\graphNet$ is duplicated 4 times to become $n_{i}$ ($i=1,2,3,4$) in $\graphAux$, while the $5$ feasible links in $\setLink^{\sector}$ are ``inherited" from $\graphNet$ to $\graphAux$, whose end points are the duplicated vertices representing the corresponding TX/RX sectors.
% Throughout the rest of the paper, we refer to $\graphNet = (\setNode, \setLink^{\sector})$ as the \emph{connectivity graph} of the sectorized communication network and to $\graphAux = (\setVertex^{\sector}, \setEdge^{\sector})$ as its \emph{auxiliary graph}.

%%%%%
%%%%%
\subsection{Feasible Schedules in $\graphNet$ as Matchings in $\graphAux$}
\label{ssec:auxiliary-schedule-matching}

% \vspace{0.5ex}
% \noindent\textbf{Feasible Schedules in $\graphNet$ as Matchings in $\graphAux$}:
For an unsectorized network $\graph$ (i.e., $\sector = \varnothing$) with $\numSector_{n} = 1, \forall n$, its auxiliary graph $H = (\setVertex, \setEdge)$ is identical to the connectivity graph $\graph = (\setNode, \setLink)$. Therefore, we use $\graph$ and $H$ interchangeably when referring to an unsectorized network. Each feasible schedule $\scheduleVec \in \scheduleSet_{\graph}$ is a \emph{matching} -- a set of links among which no two links share a common node -- in $\graph$ (and thus in $H$).
% Let $\matchingSet_{\graph}$ be the set of matchings in $\graph$.
% \tingjun{check how does this depend on the primary/secondary interference constraints.}

In a sectorized network $\graphNet = (\setNode, \setLink^{\sector})$, a feasible schedule in $\graphNet$ may not be a matching in $\graphNet$, since up to $\numSector_{n}$ links can share a common node $n$.
However, under the primary interference constraints and with the use of the auxiliary graph, \emph{any} feasible schedule in $\graphNet$ corresponds to a matching in $\graphAux$. Let $\matchingVec \in \{0,1\}^{\numLink}$ denote a matching vector in $\graphAux$, in which every element (edge) is ordered according to the position of the element (link) that it corresponds to in $\scheduleVec$. Let $\matchingSet_{\graphAux}$ be the set of matchings in $\graphAux$.
% As we will show in \mynote{REF SEC}, this property harnessed from the construction of $\graphAux$ can largely facilitate the development of efficient scheduling and routing policies for sectorized networks.
% \tingjun{say something that the schedule and matching are with the same ordered list of edges.}
To rigorously connect a feasible schedule $\scheduleVec \in \scheduleSet_{\graphNet}$ in $\graphNet$ to a matching $\matchingVec \in \matchingSet_{\graphAux}$ in $\graphAux$, we need to take a deeper look into the sets $\setLink^{\sector}$ and $\setEdge^{\sector}$: they not only have the same cardinality of $\numLink$, but are also in fact two \emph{isomorphic} sets. Due to the way $\graphAux$ is constructed, there exists an isomorphism $\isomorphism: \setEdge^{\sector} \rightarrow \setLink^{\sector}$ such that for $\edge = (n_{k}, n'_{k'}) \in \setEdge^{\sector}$ and $\link = \directLink{n}{k}{n'}{k'} \in \setLink^{\sector}$, $\isomorphism(\edge) = \link$.
% For the rest of the paper, using this isomorphism if needed, we assume that every given component of two vectors in $\mathbb{R}^{|E|}$, correspond to the same physical link.
% \panos{Tingjun, is the previous sentence clear?}
% Let $\isomorphism(E')$ denote the component-wise isomorphism applied to a set of edges, $E'$, i.e., $\isomorphism(E') = \{ \isomorphism(e): \forall e \in E' \}$.

%%%%%
%%%%%
\subsection{Background on Matching Polytopes}
\label{ssec:auxiliary-background-matching-polytopes}

% \vspace{0.5ex}
% \noindent\textbf{Matching Polytopes}:
The \emph{matching polytope} of a (general) graph $\graph = (\setVertex, \setEdge)$, denoted by $\polytope_{\graph}$, is a convex polytope whose corners correspond to the matchings in $\graph$, and can be described using Edmonds' matching polytope theorem~\cite{edmonds1965maximum}.
% Recall that for an unsectorized network, $\graph$, its auxiliary graph $H \equiv \graph$. Therefore, its matching polytope is the convex hull of the set of matchings, i.e., $\polytope_{\graph} = \convexHull(\matchingSet_{\graph})$. 
Specifically, a vector $\mathbf{x} = [x_{\edge}: \edge \in \setEdge] \in \mathbb{R}^{|\setEdge|}$ belongs to $\polytope_{\graph}$ if and only if it satisfies the following conditions~\cite{edmonds1965maximum, kahn1996asymptotics}:
\begin{align}
\textrm{(\textbf{P})} \quad
(i) & \quad x_{\edge} \geq 0,\ \forall \edge \in \setEdge,\ ~\textrm{and}~
\littlesum\nolimits_{\edge \in \delta(v)} x_{\edge} \leq 1,\ \forall v \in \setVertex, \nonumber \\
(ii) & \quad \littlesum\nolimits_{\edge \in E(U)} x_{\edge} \leq \lfloor |U|/2 \rfloor,\ \forall U \subseteq \setVertex ~\textrm{with}~ |U| ~\textrm{odd},
\label{eqn:polytope-P}
\end{align}
where $\delta(v)$ is the set of edges incident to $v$, and $E(U)$ is the set of edges in the subgraph induced by $U \subseteq \graph$. Note that in general, it is challenging to compute $\polytope_{\graph}$, since (\emph{ii}) of {\eqref{eqn:polytope-P}} includes an exponential number of constraints since all sets of vertices $U \subseteq \setVertex$ with an odd cardinality need to be enumerated through.
The \emph{fractional matching polytope} of $\graph$, denoted by $\polytopeQ_{\graph}$, is given by
\begin{align}
\textrm{(\textbf{Q})} \quad
x_{\edge} \geq 0,\ \forall \edge \in \setEdge,\ ~\textrm{and}~
\littlesum\nolimits_{\edge \in \delta(v)} x_{\edge} \leq 1,\ \forall v \in \setVertex.
\label{eqn:polytope-Q}
\end{align}
For a general graph $\graph$, $\polytope_{\graph} \subseteq \polytopeQ_{\graph}$ since (\emph{ii}) in (\textbf{P}) is excluded in (\textbf{Q}). For a bipartite graph $\graph_{bi}$\footnote{A bipartite graph is a graph whose vertices can be divided into two \emph{disjoint} and \emph{independent} sets $U_{1}$ and $U_{2}$ such that every edge connects a vertex in $U_{1}$ to one in $U_{2}$.}, its matching polytope and fractional matching polytope are equivalent, i.e., $\polytope_{\graph_{bi}} = \polytopeQ_{\graph_{bi}}$~\cite{schrijver2003combinatorial}.

Recall that for an unsectorized network $\graph$, its auxiliary graph $H \equiv \graph$ and its matching polytope is the convex hull of the set of matchings, i.e., $\polytope_{H} = \convexHull(\matchingSet_{H})$. Therefore, its capacity region $\capRegion(\graph)$ is determined by $\polytope_{H}$, since $\scheduleSet_{\graph} = \matchingSet_{H}$. For a sectorized network $\graphNet$, since a feasible schedule in $\graphNet$, $\scheduleVec \in \scheduleSet_{\graphNet}$, corresponds to a matching in the auxiliary graph $\graphAux$, $\matchingVec \in \matchingSet_{\graphAux}$ (see Section~\ref{ssec:auxiliary-schedule-matching}), we show in Section~\ref{sec:flow-extension} that its capacity region, $\capRegion(\graphNet)$, can be determined by the matching polytope of $\graphAux$, $\polytope_{\graphAux}$ (Lemma~\ref{lem:schedule-matching}).

%%%%%
%%%%%
\subsection{Equivalent Sectorizations}
\label{ssec:auxiliary-equivalent-sectorizations}

Since our objective is to obtain the optimal sectorization for each node in the network, it is important to understand the structural property of a sectorization, $\sector$. In particular, for a node $n \in \setNode$ with $\numSector_n$ number of sectors, there is only \emph{a finite number of distinct sectorizations} in $\setSector_{n}(\numSector_n)$, due to the equivalent classes of sectorizations.

Consider a node $n \in \setNode$ and two of its undirected adjacent incident links $\link_{1} = (n, u)$ and $\link_{2} = (n, v)$ for some $u, v$ incident to $n$. To define the undirected links from the set $\setLink$, we assume, for example, that the undirected link $\link_{1}$ includes the two directed links $(n,u) \in \setLink_{n}^{+}$ and $(u,n) \in \setLink_{n}^{-}$. Note that, under the primary interference constraints, it does not make a difference where exactly a sectorization axis is put between $\link_{1}$ and $\link_{2}$ of node $n$. What matters is whether a sectorization axis will be placed between them or not. If two sectorizations $\sector_{1} \in \setSector_{n}(\numSector_{n})$ and $\sector_{2} \in \setSector_{n}(\numSector_{n})$ differ only on the exact position of a sectorization axis while both sectorization have a sectoring axis placed between $\link_{1}$ and $\link_{2}$, it holds that $\graphAuxPlain^{\sector_{1}} \equiv \graphAuxPlain^{\sector_{2}}$, and thus $\matchingSet_{\graphAuxPlain^{\sector_{1}}} = \matchingSet_{\graphAuxPlain^{\sector_{2}}}$. Hence, the set $\setSector_{n}(\numSector_{n})$ consists of a finite number of equivalent sectorizations.

%% figure begins
\begin{figure}[!t]
\centering
\includegraphics[width=0.9\columnwidth]{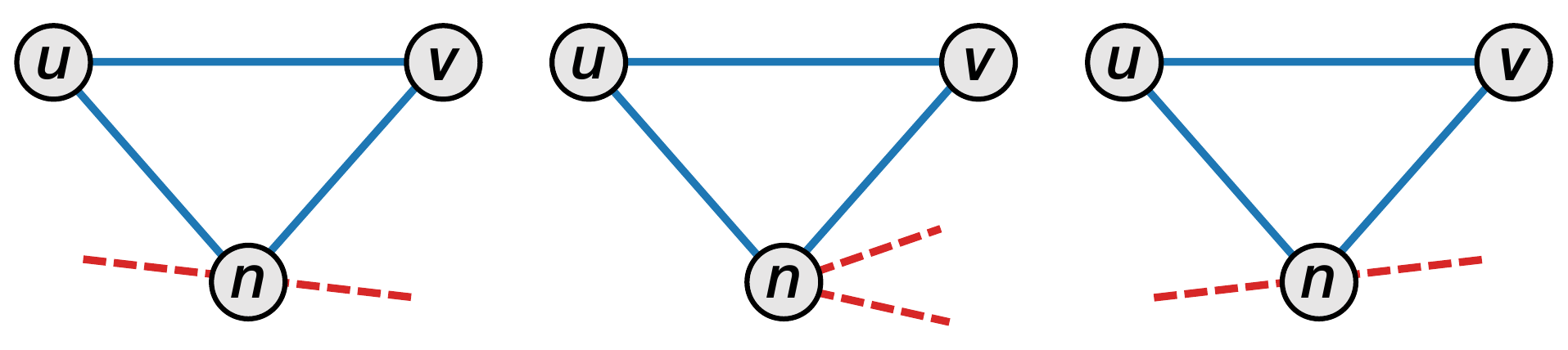}
\vspace{-1mm}
\caption{Three equivalent sectorizations for a node $n \in \setNode$.}
\label{fig:equivalent-sectorization}
\vspace{-2mm}
\end{figure}
%% figure ends

Fig.~\ref{fig:equivalent-sectorization} depicts an example of three equivalent sectorizations for node $n$, where both links $\link_{1} = (n,u)$ and $\link_{2} = (n,v)$ are placed within the same sector. 
% Furthermore, note that sectorization 1 in Fig. \ref{fig:equivalent-sectorization}, has two antipodal links whereas sectorization 2 does not.
% This, highlights the fact that sectorization 1, which has 2 antipodal sectorization axes and thus makes the auxiliary graph bipartite, is equivalent with (infinitely) many sectorizations that do not have literally antipodal axes.
% Since it is practically challenging to perfectly align the sectorization axes, 
Note that this notion of equivalent sectorizations makes the optimization over the set $\setSector(\numSectorVec)$, for every fixed $\numSectorVec \in \mathbb{Z}_{+}^{\numNode}$, more approachable. Motivated by the comparison of sectorizations of Fig.~\ref{fig:equivalent-sectorization}, when rotating a sectorization axis of node $n$, an infinite number of equivalent sectorizations can be generated until the axis meets the next incident link. Hence, up to an equivalence relation, the sectorization of a network, $\sector \in \setSector(\numSectorVec)$ is determined by whether a sectoring axis will be placed (instead of where exactly it will be placed) between adjacent undirected links of each node.
% \flowissue{I define delta here. In the opt section I say more about the undirected stuff.}

For node $n$, since a sectoring axis cannot be placed between the outgoing and incoming links corresponding to the same neighboring node of $n$, we consider undirected links regarding node $n$'s sectorization, $\sector_{n}$. In the rest of the paper, let $\delta(n)$ denote the set of the undirected incident links of $n$, i.e., $\delta(n) = \{ \link_{1}, \cdots, \link_{|\delta(n)|} \}$, ordered in a way such that geometrically neighboring (consecutive) undirected links are next to each other. Note that $\link_{|\delta(w)|}$ and $\link_{1}$ are also adjacent links due to the cyclic nature of $\sector_{n}$. Using this notation, the sectorization of node $n$, $\sector_{n} \in \setSector_{n}(\numSector_n)$, essentially partitions $\delta(n)$ into $\numSector_{n}$ non-overlapping subsets of links. We can also denote $\sector_{n}$ by placing the symbol ``$\vert$" between two adjacent links where a sectoring axis is placed (e.g., $\delta_{n} = \{ \vert (n,u), (n,v)\}$ in Fig.~\ref{fig:equivalent-sectorization}).
\section{Sectorization Gain based on Network Flow Extension Ratio}
\label{sec:flow-extension}

In this section, we characterize the capacity region of sectorized networks and present the definitions of flow extension ratio and the corresponding sectorization gain.
We consider a \emph{general} sectorized network $\graphNet$ under sectorization $\sector \in \setSector(\numSectorVec)$ with given $\numSectorVec \in \mathbb{Z}_{+}^{\numNode}$, and its auxiliary graph $\graphAux$. Let $\polytope_{\graphAux}$ be the matching polytope of $\graphAux$.

\iffullpaper
We first present the following lemma that relates the convex hull of all feasible schedules in $\graphNet$ to $\polytope_{\graphAux}$.
\else
We first present the following lemma that relates the convex hull of all feasible schedules in $\graphNet$ to $\polytope_{\graphAux}$ (for detailed proof see~\cite{TechReport}).
\fi

\begin{lemma}
\label{lem:schedule-matching}
For a sectorized network $\graphNet$ and its auxiliary graph $\graphAux$, let $\scheduleSet_{\graphNet}$ be the set of feasible schedules in $\graphNet$ and $\matchingSet_{\graphAux}$ be the set of matchings in $\graphAux$. Under the primary interference constraints, it holds that $\scheduleSet_{\graphNet} = \matchingSet_{\graphAux}$ and $\convexHull(\scheduleSet_{\graphNet}) = \convexHull(\matchingSet_{\graphAux}) = \polytope_{\graphAux}$.
\end{lemma}
%%
% where $\isomorphism(\matchingSet_{\graphAux}) = \{ \isomorphism(\matching): \forall \matching \in \matchingSet_{\graphAux} \}$.
% \panos{We should remove the isomorphism when the input is schedule or a matching since they are 0,1 vectors already. So it should state actually $\scheduleSet_{\graphNet} = \matchingSet_{\graphAux}$. We can even omit the proof.}

\iffullpaper
%% proof begins
\begin{proof}
First, we prove that $\scheduleSet_{\graphNet} \subseteq \matchingSet_{\graphAux}$. For every feasible schedule $\scheduleVec$ in  $\scheduleSet_{\graphNet}$, due to the way $\graphAux$ is constructed based on $\graphNet$ (see Section~\ref{ssec:auxiliary-graph-construction}), $\exists \matchingVec \in \matchingSet_{\graphAux}$ such that $\scheduleVec = \matchingVec$.
Next, we prove that $\scheduleSet_{\graphNet} \supseteq \matchingSet_{\graphAux}$. Consider a matching $\matchingVec \in \matchingSet_{\graphAux}$, it allows each node sector $\nodeSector{n}{k}, \forall n \in \setNode, \forall k \in [\numSector_{n}]$ to be connected to at most one direct edge (see Section~\ref{ssec:model-network}). As a result, due to the primary interference constraints, $\matchingVec$ is also a feasible schedule of $\graphNet$.
\end{proof}
%% proof ends
\fi

Based on Lemma~\ref{lem:schedule-matching}, the backpressure algorithm that obtains \emph{maximum weight schedule (MWS)} $\scheduleVec^{\textrm{BP}}(t)$ in $\graphNet$ (see Section~\ref{ssec:model-capacity}), which in general is NP-hard, is equivalent to finding the corresponding \emph{maximum weight matching (MWM)} in $\graphAux$, with the backpressure $\mathbf{D}(t)$ being the weights on the edges, i.e.,
\begin{align*}
\scheduleVec^{\textrm{BP}}(t) 
& :=
\arg \max\limits_{\scheduleVec \in \scheduleSet_{\graphNet}}\ \big\{ \bpVec^{\top}(t) \cdot \scheduleVec \big\} \\
& =
\arg \max\limits_{\matchingVec \in \matchingSet_{\graphAux}}\ \big\{ \bpVec^{\top}(t) \cdot \matchingVec \big\}.
\end{align*}
From Edmond's theory~\cite{schrijver2003combinatorial, edmonds1965paths}, the MWM of $\graphAux$ can be obtained via a polynomial-time algorithm with complexity of $O(|\setVertex^\sigma|^2 \cdot |E|)$. 

Intuitively, it is expected that $\graphNet$ under sectorization $\sector$ (with $\graphAux = (\setVertex^{\sector}, \setEdge^{\sector})$) has a capacity region that is at least the same as the unsectorized network $\graph$ (with $H = (\setVertex, \setEdge)$), i.e., any network flow $\flowVec$ that can be maintained in the unsectorized network $\graph$ can always be maintained in $\graphNet$. 
\iffullpaper
This is proved by the following theorem.
\else
This is by the following theorem, whose proof can be found in~\cite{TechReport}.
\fi

\begin{theorem}
\label{thm:capincrease}
$\polytope_{\graph} = \polytope_{H} \subseteq \polytope_{\graphAux},\ \forall \numSectorVec \in \mathbb{Z}_{+}^{\numNode} ~\textrm{and}~ \forall \sector \in \setSector(\numSectorVec)$.
\end{theorem}
\iffullpaper
\begin{proof}
Since $\graph \equiv H$, it holds that $\polytope_{\graph} = \polytope_{H}$. Consider a network flow $\flowVec \in \polytope_{\graph}$, we will show that $\flowVec$ satisfies {\eqref{eqn:polytope-P}} for $\polytope_{\graphAux}$.
First, since $\flow_{\edge} \geq 0, \forall \edge \in \setEdge$, it holds that in $\graphAux$, $\flow_{\edge} \geq 0, \forall \edge \in \setEdge^{\sector}$. For $v \in \setVertex^\sector$ corresponding to node $n \in \setNode$ with $\sector_{n}$, note that $ \sum_{e \in \delta(v)} f_e \leq   \sum_{e \in \delta(n)} f_e \leq 1$. 
% The proof would have been completed in case $\sector$ makes the auxiliary graph, $\graphAux$, bipartite. 
Second, we prove that $\sum_{\edge \in E(U')} \flow_{\edge} \leq \lfloor |U'|/2 \rfloor, \forall U' \subseteq \setVertex^{\sector}$ with an odd cardinality $|U'|$. Let $U'$ be an odd subset of vertices of $\setVertex^{\sector}$ and $U \subseteq \setNode$ denote the set of the corresponding nodes in $\setNode$ from which the vertices in $U'$ were deduced from. If $|U|$ is odd, the result easily follows since $\sum_{\edge \in E(U')} \flow_{\edge} \leq \sum_{\edge \in E(U)} \flow_{\edge} \leq \lfloor |U|/2 \rfloor \leq \lfloor |U'|/2 \rfloor$. 
If $|U|$ is even, let a node $n \in \setNode$ such that $n_k \in U'$ for a $k \in [\numSector_n]$. Let $U'' = U \setminus \{n\}$ with an odd cardinality $|U''|$. Note that $\sum_{\edge \in E(U')} \flow_{\edge} \leq \sum_{\edge \in E(U'')} \flow_{\edge} + \sum_{\edge \in \delta(n)} \flow_{\edge} \leq \sum_{\edge \in E(U'')} \flow_{\edge} + 1 \leq \frac{|U''|-1}{2} + 1 \leq \frac{|U'| - 1}{2} $. The last inequality holds since $|U'| \geq |U''| + 1$ and both $|U'|$ and $|U''|$ are odd.
\end{proof}
\fi
Based on {\eqref{eqn:cap-region}} and Lemma~\ref{lem:schedule-matching}, the capacity region of a sectorized network $\graphNet$ can be characterized by the matching polytope of its auxiliary graph, $\polytope_{\graphAux}$. With a given (unknown) $\arrivalRateMat \in \textrm{int}(\capRegion(\graphNet))$ and a sectorization $\sector$, the dynamic backpressure algorithm will converge to and return a network flow vector $\flowVec \in \polytope_{\graphAux}$.
% Our objective is to find the maximum \emph{flow extension ratio}, denoted by $\flowExtensionOpt^{\sector}(\flowVec) \in \mathbb{R}^{+}$, such that $\flowExtensionOpt^{\sector}(\flowVec) \cdot \flowVec \in \polytope_{\graphAux}$ still holds, i.e., \emph{how much the network flow $\flowVec$ can be extended in its direction until it intersects the boundary of $\polytope_{\graphAux}$}?
We define the flow extension ratio, denoted by $\flowExtensionOpt^{\sector}(\flowVec)$, which quantitatively measures \emph{how much the network flow $\flowVec$ can be extended in its direction until it intersects the boundary of $\polytope_{\graphAux}$}. 

% Without loss of generality, we assume that $\flowVec$ is normalized to be a unit vector. 
% \tingjun{do we really need this normalization? I feel like this is redundant.}

%%
\begin{definition}[Flow Extension Ratio]
For a sectorized network $\graphNet$ and network flow $\flowVec \in \polytope_{\graphAux}$, the flow extension ratio, $\flowExtensionOpt^{\sector}(\flowVec)$, is the maximum scalar such that $\flowExtensionOpt^{\sector}(\flowVec) \cdot \flowVec \in \polytope_{\graphAux}$ still holds.
\end{definition}

Let $\zeta^{\sector}(\flowVec) := \min_{U \subseteq \setVertex^{\sector}, |U| ~\textrm{odd}} \left( \frac{\lfloor |U|/2 \rfloor}{\sum_{\edge \in E(U)} \flow_{\edge}} \right)$ (see {\eqref{eqn:polytope-P}}). 
%\flowissue{I add here the comment about delta for now. Is it ok? Revisit.}
With a little abuse of notation, for node $n$, network flow $\flowVec$, and an undirected edge $\edge \in \delta(n)$ with directed components $\edge_{+} \in \setLink_{n}^{+}$ and $\edge_{-} \in \setLink_{n}^{-}$, we let $\flow_{\edge} = \flow_{\edge_{+}} + \flow_{\edge_{-}}$. The flow extension ratio, $\flowExtensionOpt^{\sector}(\flowVec)$, can be obtained by the following optimization problem:
\begin{align}
\flowExtensionOpt^{\sector}(\flowVec)
& := \max_{\flowExtension \in \mathbb{R}^{+}}~ \flowExtension,\ ~\textrm{s.t.:}~ \flowExtension \cdot \flowVec \in \polytope_{\graphAux} \nonumber \\
& \stackrel{\eqref{eqn:polytope-P}}{=} \min \big\{
\min_{v \in \setVertex^{\sector}} \frac{1}{\sum_{\edge \in \delta(v)} \flow_{\edge}},\
% \min_{\substack{U \subseteq \setVertex^{\sector}, \\ |U| ~\textrm{odd}}} \frac{\lfloor |U|/2 \rfloor}{\sum_{\edge \in E(U)} \flow_{\edge}}
\zeta^{\sector}(\flowVec)
\big\}.
\label{eqn:flow-extension-P}
\end{align}
Similarly, we define the flow extension ratio for the corresponding unsectorized network $\graph$, denoted by $\flowExtensionOpt^{\varnothing}(\flowVec) := \max_{\flowExtension \in \mathbb{R}^{+}}~ \flowExtension,\ ~\textrm{s.t.:}~ \flowExtension \cdot \flowVec \in \polytope_{\graph}$. Note that $\flowExtensionOpt^{\varnothing}(\flowVec)$ depends solely on the topology of $\graph$ and is independent of the sectorization $\sector$.
% Therefore, we can quantitatively evaluate the efficiency of a sectorization $\sector \in \setSector$ with respect to the network flow $\flowVec$ by the sectorization gain, defined as $\sectorizationGainP^{\sector}(\flowVec) := \frac{\flowExtensionOpt^{\sector}(\flowVec)}{\flowExtensionOpt^{\varnothing}(\flowVec)}$.
Next, we define the \emph{approximate flow extension ratio} with respect to the polytope $\polytopeQ_{\graphAux}$.

\begin{definition}[Approximate Flow Extension Ratio]
For a sectorized network $\graphNet$ and a network flow $\flowVec \in \polytope_{\graphAux}$, the approximate flow extension ratio, $\flowExtensionQOpt^{\sector}(\flowVec)$, is the maximum scalar such that $\flowExtensionQOpt^{\sector}(\flowVec) \cdot \flowVec \in \polytopeQ_{\graphAux}$ still holds.
\end{definition}

The following optimization problem determines $\flowExtensionQ^{\sector}(\flowVec)$:
\begin{align}
\flowExtensionQOpt^{\sector}(\flowVec)
& := \max_{\flowExtensionQ \in \mathbb{R}^{+}}~ \flowExtensionQ,\ ~\textrm{s.t.:}~ \flowExtensionQ \cdot \flowVec \in \polytopeQ_{\graphAux} \nonumber \\
& \stackrel{\eqref{eqn:polytope-Q}}{=} \min_{v \in \setVertex^{\sector}} \frac{1}{\sum_{\edge \in \delta(v)} \flow_{\edge}}
= \frac{1}{\max_{v \in \setVertex^{\sector}} \sum_{\edge \in \delta(v)} \flow_{\edge}},
\label{eqn:flow-extension-Q}
\end{align}
%%
% where $\setVertex_{+}^{\sector}$ is the set of nodes that have at least one incident edge with positive component in $\flow_{\edge}$, i.e., the nodes $n \in \setVertex^{\sector}$ such that $\sum_{\edge \in \delta(n)} \flow_{\edge} > 0$ \tingjun{this definition is not very clear to me.}
% Clearly, $\flowExtensionQOpt^{\sector}(\flowVec) \geq \flowExtensionOpt^{\sector}(\flowVec)$.
We also define the approximate flow extension ratio for the unsectorized network $\graph$, denoted by $\flowExtensionQOpt^{\varnothing}(\flowVec) := \max_{\flowExtension \in \mathbb{R}^{+}}~ \flowExtension,\ ~\textrm{s.t.:}~ \flowExtension \cdot \flowVec \in \polytopeQ_{\graph}$.

To quantitatively evaluate the performance of a sectorization $\sector$, we present the following definition of sectorization gains.
\begin{definition}[Sectorization Gains]
The (explicit) sectorization gain of $\sector$ with a network flow $\flowVec$ is defined as $\sectorizationGainP^{\sector}(\flowVec) := \frac{\flowExtensionOpt^{\sector}(\flowVec)}{\flowExtensionOpt^{\varnothing}(\flowVec)}$ (with respect to $\polytope_{\graphAux}$).
The approximate sectorization gain of $\sector$ with a network flow $\flowVec$ is defined as $\sectorizationGainQ^{\sector}(\flowVec) := \frac{\flowExtensionQOpt^{\sector}(\flowVec)}{\flowExtensionQOpt^{\varnothing}(\flowVec)}$ (with respect to $\polytopeQ_{\graphAux}$).
\end{definition}

The following lemma states the relationships between $\flowExtensionOpt^{\sector}(\flowVec)$ and $\flowExtensionQOpt^{\sector}(\flowVec)$, and between $\sectorizationGainP^{\sector}(\flowVec)$ and $\sectorizationGainQ^{\sector}(\flowVec)$. Parts (i) and (ii) of the Lemma can be found in \cite{hajek1988link} (Corollary of Section III).

\begin{lemma}
\label{lemma:sasaki_approx}
For any sectorization $\sector$ and network flow $\flowVec$, it holds that:
(i) $\frac{2}{3} \cdot \flowExtensionQOpt^{\sector}(\flowVec) \leq \flowExtensionOpt^{\sector}(\flowVec) \leq \flowExtensionQOpt^{\sector}(\flowVec)$,
(ii) $\frac{1}{\flowExtensionQOpt^{\sector}(\flowVec)} \geq  \frac{1}{\flowExtensionOpt^{\sector}(\flowVec)} - \max_{\edge \in \setEdge^{\sector}} \flow_{\edge}$, and
(iii) $\frac{2}{3} \cdot \sectorizationGainQ^{\sector}(\flowVec) \leq \sectorizationGainP^{\sector}(\flowVec) \leq \frac{3}{2} \cdot \sectorizationGainQ^{\sector}(\flowVec)$.
\end{lemma}

\iffullpaper
\begin{proof}
(\emph{i}) can be proved with the observation that in any (odd) subset $U \subseteq \setVertex^{\sector}$, the cumulative weighted degree of any node cannot exceed $\frac{1}{\flowExtensionQOpt^{\sector}(\flowVec)}$ by definition. It can also be proven from Shannon's result for chromatic indices~\cite{shannon1949theorem, hajek1988link}. (\emph{ii}) can be deduced from Vizing's theorem~\cite{berge1973graphs, hajek1988link}, and (\emph{iii}) is an immediate result of (\emph{i}).
\end{proof}
\fi

\iffullpaper
%% remark begins
\begin{remark}
\label{remark:bipartiteequality}
If $\graphAux$ is a bipartite graph under sectorization $\sector$, since the two polytopes $\polytopeQ_{\graphAux}$ and $\polytope_{\graphAux}$ are identical (see Section~\ref{ssec:auxiliary-background-matching-polytopes}), it holds that $\flowExtensionOpt^{\sector}(\flowVec) = \flowExtensionQOpt^{\sector}(\flowVec), \forall \flowVec \in \mathbb{R}^{|\setEdge|}$. 
See the discussions in Section~\ref{sec:bipartite} for more details.
\end{remark}
\else
% See the discussions in Section~\ref{ssec:optimization-algorithm} for more details.
\fi
% \else
% See the discussions in section \ref{ssec:optimization-algorithm} for more details.
% \fi

%% remark ends

\section{Optimization and A Distributed Approximation Algorithm}
\label{sec:optimization}

Given a vector $\numSectorVec \in \mathbb{Z}_{+}^{\numNode}$, our objective is to find the sectorization $\sector \in \setSector(\numSectorVec)$ that maximizes the flow extension ratio, $\flowExtensionOpt^{\sector}(\flowVec)$, under a network flow, $\flowVec$. 
While the algorithm in this section focuses on a single, fixed sectorization, in Section~\ref{ssec:joint_throughput_optimality} we present a dynamic routing approach that allows the network to re-sectorize in real time, adapting to the current network state. The aforementioned optimization problem, {\opt}, is given by:
\begin{align}
\hspace{-8pt}
& \textrm{(\textbf{Opt})}~
\sectorOpt(\flowVec) := \arg \max_{\sector \in \setSector(\numSectorVec)}~ \flowExtensionOpt^{\sector}(\flowVec) \nonumber \\
& \stackrel{\eqref{eqn:flow-extension-P}}{=} \arg \max_{\sector \in \setSector(\numSectorVec)} \Big\{ \min \Big\{
\min_{v \in \setVertex^{\sector}} \frac{1}{\sum_{\edge \in \delta(n)} \flow_{\edge}},\
% \min_{\substack{U \subseteq \setVertex^{\sector}, \\ |U| ~\textrm{odd}}} \frac{\lfloor |U|/2 \rfloor}{\sum_{\edge \in E[U]} \flow_{\edge}}
\zeta^{\sector}(\flowVec)
\Big\} \Big\}.
\label{eqn:opt}
\end{align}
We are optimizing towards the boundary of the matching polytope, $\polytope_{\graphAux}$, over the set of sectorizations, $\setSector(\numSectorVec)$, for a given network flow, $\flowVec$. This is because there does not exist a single sectorization $\sector \in \setSector(\numSectorVec)$ that can achieve close-to-optimal performance of $\flowExtensionOpt^{\sector}(\flowVec)$ for every network flow $\flowVec \in \mathbb{R}^{|\setEdge|}$.
\begin{algorithm}[!t]
\caption{{\algoName} for node $n$.}
\label{algo:find-sector}
\small
\begin{algorithmic}[1]
\Statex\hspace{-\algorithmicindent}\textbf{Input}: $\numSector_{n}$, $\delta(n)$, $\flowVec_{n}$, and $\epsilon$
% and the decision problem for node $n$, $\decisionProb_{n}(\decisionThreshold, \delta(n))$ {\eqref{eqn:decisition-problem-node-n}}.
\vspace*{0.5ex}

\State
$\decisionThreshold_{\textrm{min}} \gets \max_{\edge \in \delta(n)} \flow_{\edge}$,
$\decisionThreshold_{\textrm{max}} \gets \sum_{\edge \in \delta(n)} \flow_{\edge}$,
$\decisionThreshold_{\textrm{temp}} \gets 0$

\State
$\decisionThresholdCrit \gets 0$, 
$\texttt{decision} \gets \texttt{No}$

\While{$(\decisionThreshold_{\textrm{max}} - \decisionThreshold_{\textrm{min}}) > \epsilon$}

\State
$\decisionThreshold_{\textrm{temp}} \gets (\decisionThreshold_{\textrm{min}} + \decisionThreshold_{\textrm{max}})/2$

\State
$(\texttt{decision}, \sectorApproxAlgo_{0}) \gets \decisionProb(\decisionThreshold_{\textrm{temp}})$ (Algorithm~\ref{algo:decision-problem})

\If{$\texttt{decision} = \texttt{Yes}$}
\State $\decisionThreshold_{\textrm{max}} \gets \decisionThreshold_{\textrm{temp}}$
\Else
\State $\decisionThreshold_{\textrm{min}} \gets \decisionThreshold_{\textrm{temp}}$
\EndIf

\EndWhile

\State
$\decisionThresholdCrit \gets \decisionThreshold_{\textrm{max}}$

\State
$(\texttt{decision}, \sectorApproxAlgo_{n}(\flowVec_{n})) \gets \decisionProb(\decisionThresholdCrit)$

\State \textbf{return} The sectorization for node $n$, $\sectorApproxAlgo_{n}(\flowVec_{n})$

\end{algorithmic}
\end{algorithm}
%%%%%
%%%%%

%%%%%
%%%%%
\begin{algorithm}[!t]
\caption{$\decisionProb(\decisionThreshold)$ for node $n$.}
\label{algo:decision-problem}
\small
\begin{algorithmic}[1]

\Statex\hspace{-\algorithmicindent}\textbf{Input}: $\numSector_{n}$, $\delta(n)$, $\flowVec_{n}$, and $\decisionThreshold$

\vspace*{0.5ex}

% \State
% $\texttt{min\_sectors\_needed} \gets +\infty$

% \tingjun{don't you need to reset these two parameters for every for iteration?}

\For{$\edge \in \delta(n)$}

\State
Reset node $n$'s sectorization $\sectorApproxAlgo_{n} \gets \varnothing$

\State
Put a sectorizing axis in $\sectorApproxAlgo_{n}$ right after $\edge$ (clockwise)

\State
$\texttt{sectors\_needed} \gets 1$, $\texttt{total\_weight} \gets 0$

\State
$\edge' \gets$ the edge next to $\edge$ in $\delta(n)$ (clockwise)

\While{$\edge' \neq \edge$}
\If{$\flow_{\edge'} > \decisionThreshold$}
\textbf{return} $(\texttt{NO}, \varnothing)$

\Else
\State
$\texttt{total\_weight} \gets \texttt{total\_weight} + \flow_{\edge'}$ 

\If{$\texttt{total\_weight} > \decisionThreshold$}
\State
$\texttt{sectors\_needed} \gets \texttt{sectors\_needed} + 1$

\State
Put a sectorizing axis in $\sectorApproxAlgo_{n}$ before $\edge'$ (counter-clockwise)
\State
$\texttt{total\_weight} \gets \flow_{\edge'}$ 

\EndIf

\State $\edge' \gets$ the edge next to $\edge$ in $\delta(n)$ (clockwise)

\EndIf

\EndWhile

\If{ $\texttt{sectors\_needed} \leq \numSector_{n}$ }

\State
\textbf{return} $(\texttt{Yes}, \sectorApproxAlgo_{n})$

\EndIf

\EndFor

%\State
% $\texttt{min\_sectors\_needed} \gets \min \{ \texttt{min\_sectors\_needed}, \texttt{sectors\_needed}\}$

% \State
% $\texttt{decision} \gets \texttt{sectors\_needed} <= \numSector$

\State
\textbf{return} $($\texttt{No}$, \varnothing)$

\end{algorithmic}
\end{algorithm}
%%%%%
%%%%%

%%%%%
%%%%%
\subsection{Key Insight and Intuition}
\label{ssec:optimization-insight}

In general, {\opt} is analytically intractable due to its combinatorial nature and the fact that $\zeta^{\sector}(\flowVec)$ needs exponentially many calculations even for a fixed sectorization $\sector$. Instead, we consider an alternative optimization problem:
\begin{align}
\textrm{(\textbf{Opt-Approx})}~
& \sectorApprox(\flowVec) 
:= \arg \max_{\sector \in \setSector(\numSectorVec)}~ \flowExtensionQOpt^{\sector}(\flowVec) \nonumber \\
& \stackrel{\eqref{eqn:flow-extension-Q}}{=} \arg \max_{\sector \in \setSector(\numSectorVec)}~
\Big\{ \min_{v \in \setVertex^{\sector}} \frac{1}{\sum_{\edge \in \delta(v)} \flow_{\edge}} \Big\}.
\label{eqn:opt-approx}
\end{align}
Although {\optApprox} excludes the constraint $\zeta^{\sector}(\flowVec)$ in {\opt}, solving it by a brute-force algorithm is still analytically intractable due to the coupling between the (possibly different) sectorization of \emph{all} nodes. For example, if each of the $\numNode$ nodes has $\widetilde{\numSector}$ possible sectorizations, a total number of $\numNode^{\widetilde{\numSector}}$ sectorizations need to be evaluated to solve {\optApprox}.
% \panos{In the previous sentence I changed K to mathcal(K) since we dont need the number of sectors but the number of different sectorizations.}

% Note that {\optApprox} is a max-min problem.  Brute-force methods cannot efficiently solve it due to its intractable combinatorial search space. Specifically, even with the observation that the set $\setSector$ is finite, it includes many sectorings that requires investigation, something that makes it intractable to work with. 
% From Remark 2, $\optApprox$ is equivalent to $\opt$ when $\graphAux$ is a bipartite graph for a given $\sector \in \setSector$. Therefore, in that case, solving $\optApprox$ gives us the optimal solution. \panos{again mention that we will present such sectorings.} \tingjun{this is a bit convoluted to me: if the objective is to find the best $\sector$ how can you make sure that the best $\sector$ will make the auxiliary graph bipartite?} \panos{We will optimize over the set of homogeneous sectorings in the next sections. So the optimization set will change.}

Our key insight behind developing a distributed approximation algorithm (described in Section~\ref{ssec:optimization-algorithm}) is that {\optApprox} can be decomposed into $\numNode$ individual optimization problems, {\optApproxN}, for each node $n \in \setNode$. With a little abuse of notation, let $\flowVec_{n} = (\flow_{\edge}: \edge \in \delta(n))$ be the ``local" flows on the undirected edges incident to node $n$. This decomposed optimization problem is given by:
\begin{align}
\hspace*{-9pt}
& \textrm{(\textbf{Opt-Approx}-$n$)}~
\sectorApproxDist_{n}(\flowVec_{n}) \nonumber \\
& \qquad := \arg \min_{ \sector_{n} \in \setSector_{n}(\numSector_n) } 
\max_{v \in \{n_{1}^{\sector_{n}}, \cdots, n_{\numSector_{n}}^{\sector_{n}}\}} \littlesum_{\edge \in \delta(v)} \flow_{\edge}. 
\label{eqn:opt-approx-n}
\end{align}
{\optApproxN} determines node $n$'s sectorization, $\sectorApproxDist_{n}(\flowVec_{n})$, based on its local flows, $\flowVec_{n}$, and is \emph{independent} of the other nodes. 
\iffullpaper
The following lemma shows that the sectorization obtained by solving the $\numNode$ decomposed problems, {\optApproxN}, is equivalent to that obtained by solving {\optApprox}.
\else
The following lemma shows that the sectorization obtained by solving the $\numNode$ decomposed problems, {\optApproxN}, is equivalent to that obtained by solving {\optApprox}, whose proof is in~\cite{TechReport}.
\fi

\begin{lemma}
\label{lem:opt-approx-n-equivalence}
For a given network flow $\flowVec$, the sectorization $\sectorApproxDist(\flowVec)$, which consists of the solutions of {\optApproxN} for all nodes $n \in \setNode$, is equivalent to the solution $\sectorApprox(\flowVec)$ to {\optApprox}, i.e.,
\begin{align}
\sectorApprox(\flowVec) = \sectorApproxDist(\flowVec) = (\sectorApproxDist_{1}(\flowVec_{1}), \cdots, \sectorApproxDist_{\numNode}(\flowVec_{\numNode})) ~\textrm{and}~
\flowExtensionQOpt^{\sectorApproxDist}(\flowVec) = \flowExtensionQOpt^{\sectorApprox}(\flowVec).
\end{align}
\end{lemma}
%%

%% proof begins
\iffullpaper
\begin{proof}
The key observation is that the sectorization of node $n$, $\sector_{n}$, affects only the values of $\{ \sum_{\edge \in \delta(n_1^\sector)} \flow_{\edge}, \cdots, \sum_{\edge \in \delta(n^{\sector}_{\numSector_{n}})} \flow_{\edge} \}$, for nodes $\{ n_1^{\sector}, \cdots, n_{\numSector_{n}}^{\sector} \} \subseteq \setVertex^{\sector}$ in the auxiliary graph $\graphAux$, respectively. Therefore, it holds that
\begin{align}
\sectorApprox(\flowVec)
& \stackrel{\eqref{eqn:opt-approx}}{=} \arg \max_{\sector \in \setSector(\numSectorVec)}~ \flowExtensionQOpt^{\sector}(\flowVec) \nonumber \\
& \stackrel{\eqref{eqn:flow-extension-Q}}{=} \arg \max_{\sector \in \setSector(\numSectorVec)}~
\Big\{ \min_{v \in \setVertex^{\sector}} \frac{1}{\sum_{\edge \in \delta(v)} \flow_{\edge}} \Big\} \nonumber \\
& = \arg \min_{(\sector_1, \dots, \sector_{\numNode}) \in \setSector(\numSectorVec)}~ \Big\{ \max_{v \in \setVertex^{\sector}} \littlesum_{\edge \in \delta(v)} \flow_{\edge} \Big\} \nonumber \\
& = \arg \min_{(\sector_1, \dots, \sector_{\numNode}) \in \setSector(\numSectorVec)}~ \Big\{ \max_{n \in \setNode}~ \Big\{ \max_{k \in [\numSector_n]} \littlesum_{\edge \in \delta(n_{k}^\sector)} \flow_{\edge} \Big\} \Big\} \nonumber \\
& = \Big[ \arg \min_{ \sector_{n} \in \setSector_{n}(\numSector_n) } 
\Big\{ \max_{v \in \{n_{1}^{\sector}, \cdots, n_{\numSector_{n}}^{\sector}\}} \littlesum_{\edge \in \delta(v)} \flow_{\edge} \Big\}: \forall n \in \setNode \Big] \nonumber \\
& \stackrel{\eqref{eqn:opt-approx-n}}{=} ( \sectorApproxDist_{1}(\flowVec_{1}), \cdots, \sectorApproxDist_{\numNode}(\flowVec_{\numNode}) ) = \sectorApproxDist(\flowVec).
\label{eqn:opt_separable}
\end{align}
Since $\sectorApprox(\flowVec) = \sectorApproxDist(\flowVec)$, it also holds that $\flowExtensionQOpt^{\sectorApproxDist}(\flowVec) = \flowExtensionQOpt^{\sectorApprox}(\flowVec)$ (see {\eqref{eqn:flow-extension-Q}}). Note that the second-to-last equation in \eqref{eqn:opt_separable} holds because the corresponding problem is separable.
\end{proof}
\else
\fi

Although {\optApproxN} is a distributed optimization problem for each node $n$, to solve it using a brute-force method still remains high complexity. Consider the sectorization of node $n$ with $|\delta(n)|$ incident edges into $\numSector_{n}$ sectors ($|\delta(n)| \leq (\numNode-1)$). To solve {\optApproxN} in a brute-force way, a total number of ${|\delta(n)| \choose \numSector_{n}} = O{\numNode \choose \numSector_{n}} = O(\numNode^{\numSector_{n}})$ possible sectorizations need to be enumerated. Next, we present a polynomial-time distributed algorithm that solves {\optApproxN} with a complexity independent of $\numSector_{n}$.

\subsection{A Distributed Approximation Algorithm}
\label{ssec:optimization-algorithm}
We now present \emph{a distributed approximation algorithm}, {\algoName}, that efficiently solves {\opt} with guaranteed performance. In essence, {\algoName} solves {\optApproxN} for each individual node $n \in \setNode$. Algorithm~\ref{algo:find-sector} presents the pseudocode for {\algoName}, which includes two main components:

\vspace{0.5ex}
\noindent
(1) A \underline{\textbf{decision problem}}, $\decisionProb(\decisionThreshold)$, that determines the \emph{existence} of a sectorization for node $n$ (with $\numSector_{n}$, $\delta(n)$, and $\flowVec_{n}$) under a given threshold value, $\decisionThreshold \in \mathbb{R}_{+}$, given by
\begin{align}
& \hspace*{-9pt} \decisionProb(\decisionThreshold) = \nonumber \\
& \hspace*{-6pt} \begin{cases}
(\texttt{Yes}, \sectorApproxDist_{0}), & \text{if}~ \exists \sectorApproxDist_{0} \in \setSector_{n}(\numSector_n): \max\limits_{v \in \{n_{1}^{\sectorApproxDist_{0}}, \cdots, n_{\numSector_{n}}^{\sectorApproxDist_{0}}\}} \littlesum\limits_{\edge \in \delta(v)} \flow_{\edge} \leq \decisionThreshold, \\
(\texttt{No}, \varnothing), & \text{if}~ \forall \sectorApproxDist_{0} \in \setSector_{n}(\numSector_n): \max\limits_{v \in \{n_{1}^{\sectorApproxDist_{0}}, \cdots, n_{\numSector_{n}}^{\sectorApproxDist_{0}}\}} \littlesum\limits_{\edge \in \delta(v)} \flow_{\edge} > \decisionThreshold.
\end{cases}
\label{eqn:decisition-problem-node-n}
\end{align}
$\decisionProb(\decisionThreshold)$ can be solved as follows. Consider the set of edges incident to node $n$, $\delta(n) = (\edge_{1}', \cdots, \edge_{|\delta(n)|}')$. We begin by assuming that the first sectoring axis of node $n$ is placed between its first and last incident edge, i.e., $(\vert \edge_{1}', \cdots, \edge_{|\delta(n)|}')$.
% Note that in case we can solve the decision problem assuming a single sectoring axis fixed, we can run the same procedure $|\delta(n)|$ times to deduce the final answer.
Under this assumption, the second sectoring axis will be placed between the $k^{\textrm{th}}$ and $(k+1)^{\textrm{th}}$ incident edges if $\littlesum_{i=1}^{k} \flow_{\edge_{i}'} \leq \decisionThreshold$ and $\littlesum_{i=1}^{k+1} \flow_{\edge_{i}'} > \decisionThreshold$. Then, the process is repeated for the remaining edges $\delta(n) \setminus \{\edge_{1}, \cdots, \edge_{k}\}$ to find the third sectoring axis, so on and so forth until all edges in $\delta(n)$ are enumerated for the threshold, $\decisionThreshold$. Although the above procedure started by placing the first sectoring axis arbitrarily, correctness is ensured by repeating this process for all $|\delta(n)|$ possible initial positions between incident edges.

\vspace{0.5ex}
\noindent
(2) A \underline{\textbf{binary search process}} that finds the critical threshold, $\decisionThresholdCrit$, based on which the optimized sectorization for node $n$, denoted by $\sectorApproxAlgo_{n}$, is determined by $\decisionProb(\decisionThresholdCrit)$.
The following theorem states our main results regarding the correctness of {\algoName} and its guaranteed performances as a distributed approximation algorithm.

% \tingjun{One logistical question here: should we start from a general network? However we won't have the flow vector then. It seems like we should start from a given sectorization so that we can get the flow vector first??}

%%
\begin{theorem}
\label{thm:main}
For a given network flow $\flowVec$, it holds that:
\begin{enumerate}[leftmargin=*,topsep=1pt]
\item[(i)]
[\textbf{Correctness}] The sectorization of node $n$ returned by {\algoName}, $\sectorApproxAlgo_{n}(\flowVec_{n})$, is equivalent to $\sectorApproxDist_{n}(\flowVec_{n})$ in {\optApproxN}, i.e.,
\begin{align}
\sectorApproxAlgo_{n}(\flowVec_{n}) = \sectorApproxDist_{n}(\flowVec_{n}),\ \forall n \in \setNode,
\label{eqn:thm-correctness}
\end{align}
\item[(ii)]
[\textbf{Approximation Ratio}] The distributed {\algoName} algorithm is a $2/3$-approximation algorithm, i.e., 
\begin{align}
\textstyle
\frac{2}{3} \leq \frac{\flowExtensionOpt^{\sectorApproxAlgo}(\flowVec)}{\flowExtensionOpt^{\sectorOpt}(\flowVec)} \leq 1,
\label{eqn:thm-approx-ratio}
\end{align}
where $\flowExtensionOpt^{\sectorOpt}(\flowVec)$ is flow extension ratio achieved by the optimal sectorization as the solution to {\opt}.
\end{enumerate}
%%
% it holds that:
% (i) $\frac{ \flowExtensionOpt^{\sectorApprox}(\flowVec)}{ \flowExtensionOpt^{\sectorOpt}(\flowVec)} \geq \frac{2}{3}$, and (ii) $\frac{1}{\flowExtensionOpt^{\sectorOpt}(\flowVec)} \ge \frac{1}{\flowExtensionOpt^{\sectorApprox}(\flowVec)} - max_{\edge \in \setEdge^\sector}\{ \flow_{\edge} \} $.
\end{theorem}
%%

% \mynote{Correctness, convergence, and complexity}
% \mynote{the critical threshold is different for different $n$, right?}

%%
\begin{proof}
\emph{First, we prove (i)}.
We show that the decision returned by $\decisionProb(\decisionThreshold)$ has a monotonic property with respect to $\decisionThreshold$, i.e., there exists a \emph{critical threshold}, $\decisionThresholdCrit \in \mathbb{R}_{+}$, such that
\begin{align}
\decisionProb(\decisionThreshold) =
\begin{cases}
(\texttt{No}, \varnothing), & \forall \decisionThreshold \in [0, \decisionThresholdCrit), \\
(\texttt{Yes}, \cdot), & \forall \decisionThreshold \in [\decisionThresholdCrit, +\infty).
\end{cases}
\label{eqn:decision-problem-monotonicity}
\end{align}
It is easy to see that $\decisionProb(\decisionThreshold)$ outputs \texttt{No} for $\decisionThreshold = 0$ and \texttt{Yes} for a sufficiently large $\decisionThreshold$ with a non-zero flow $\flowVec$. Let $\decisionThresholdCrit$ be the smallest $\decisionThreshold$ such that $\exists \sectorApproxAlgo_{n} \in \setSector_{n}(\numSector_n)$ with which the output of {\eqref{eqn:decisition-problem-node-n}} is $(\texttt{Yes}, \sectorApproxAlgo_{n})$.
% \panos{I think we should say something like $\decisionProb_{n}(\decisionThresholdCrit, \delta(n)) = (\texttt{Yes}, \cdot)$, because we need the first YES and not a specific sectoring - naming it v might be confusing. Or just say we need the first T that gives yes as the answer and then: let $\sectorApproxDist_{n}$ be such that  $\decisionProb_{n}(\decisionThresholdCrit, \delta(n)) = (\texttt{Yes}, \sectorApproxDist_{n})$. Right?}.
As a result, we have $\decisionProb( \decisionThreshold) = (\texttt{No}, \varnothing),\ \forall \decisionThreshold \in [0, \decisionThresholdCrit)$. In addition, $\sectorApproxAlgo_{n}$ and $\decisionThresholdCrit$ satisfy
\begin{align*}
\max\limits_{v \in \{n_{1}^{\sectorApproxAlgo_{n}}, \cdots, n_{\numSector_{n}}^{\sectorApproxAlgo_{n}}\}} \littlesum\limits_{\edge \in \delta(v)} \flow_{\edge} \leq \decisionThreshold,\ \forall \decisionThreshold \in [\decisionThresholdCrit, +\infty).
\end{align*}
%%
% Furthermore, from the former we get that $\max\limits_{v \in \{n_{1}^{\phi_{n}}, \cdots, n_{\numSector_{n}}^{\phi_{n}}\}} \littlesum\operatornamewithlimits{}_{\edge \in \delta(v)} \flow_{\edge} \le \decisionThresholdCrit$ and from the latter we get that, $\max\limits_{v \in \{n_{1}^{\sector_{n}}, \cdots, n_{\numSector_{n}}^{\sector_{n}}\}} \littlesum\operatornamewithlimits{}_{\edge \in \delta(v)} \flow_{\edge} > T, \quad \forall \sector_n \in \setSector_n, \forall T \in [0, \decisionThresholdCrit)$.
Therefore, for a network flow $\flowVec_{n}$ incident to node $n$, we can set
\begin{align}
\decisionThresholdCrit & = \min_{ \sector_{n} \in \setSector_{n}(\numSector_n) } \Big\{ \max_{v \in \{n_{1}^{\sector_{n}}, \cdots, n_{\numSector_{n}}^{\sector_{n}}\}} \littlesum_{\edge \in \delta(v)} \flow_{\edge} \Big\},
% \sectorApproxAlgo_{n} & = \sectorApproxDist_{n}(\flowVec) = \arg \min_{ \sector_{n} \in \setSector_n  } 
% \max_{v \in \{n_{1}^{\sector_{n}}, \cdots, n_{\numSector_{n}}^{\sector_{n}}\}} \littlesum_{\edge \in \delta(v)} \flow_{\edge}
\label{eqn:proof-decision-threshold-critical}
\end{align}
and the sectorization of node $n$ corresponding to $\decisionThresholdCrit$ is equivalent to $\sectorApproxDist_{n}$ in {\optApproxN}, i.e.,
\begin{align}
\hspace*{-10pt}
\sectorApproxAlgo_{n}(\flowVec) = \arg \min_{ \sector_{n} \in \setSector_{n}(\numSector_n) } \Big\{ \max_{v \in \{n_{1}^{\sector_{n}}, \cdots, n_{\numSector_{n}}^{\sector_{n}}\}} \littlesum_{\edge \in \delta(v)} \flow_{\edge} \Big\}
\stackrel{\eqref{eqn:opt-approx-n}}{=}  \sectorApproxDist_{n}(\flowVec_{n}).
\label{eqn:proof-sector-approx-algo}
\end{align}
Due to the monotonicity of {\eqref{eqn:decision-problem-monotonicity}}, the critical value $\decisionThresholdCrit$ can be found via a binary search within the interval $[\max_{\edge \in \delta(n)} \flow_{\edge}, \sum_{\edge \in \delta(n)} \flow_{\edge}]$ with a sufficiently small $\epsilon$, which is a parameter that controls the trade-offs between accuracy and convergence of Algorithm~\ref{algo:find-sector}. Then, $\sectorApproxAlgo_{n}(\flowVec_{n})$ can be obtained via ${\decisionProb}(\decisionThresholdCrit)$.

\emph{Next, we prove (ii)}.
For a network flow $\flowVec$, consider the sectorizations $\sectorOpt(\flowVec)$ and $\sectorApprox(\flowVec)$ as the solution to {\opt} and {\optApprox}, respectively. It is easy to see from their definitions that $\flowExtensionOpt^{\sectorApproxAlgo}(\flowVec) \leq \flowExtensionOpt^{\sectorOpt}(\flowVec) \Rightarrow \frac{\flowExtensionOpt^{\sectorApproxAlgo}(\flowVec)}{\flowExtensionOpt^{\sectorOpt}(\flowVec)} \leq 1$. From {\eqref{eqn:opt}}--{\eqref{eqn:opt-approx}} and Lemma \ref{lemma:sasaki_approx}, we have
\begin{align}
\label{equation_approx_optimality}
\textstyle
\flowExtensionOpt^{\sectorOpt}(\flowVec) \leq \flowExtensionQOpt^{\sectorOpt}(\flowVec) \leq \flowExtensionQOpt^{\sectorApprox}(\flowVec)
~\textrm{and}~
\frac{2}{3} \cdot \flowExtensionQOpt^{\sectorApprox}(\flowVec) \leq  \flowExtensionOpt^{\sectorApprox}(\flowVec),
\end{align}
which further implies that
\begin{align}
\textstyle
\frac{2}{3} \cdot \flowExtensionOpt^{\sectorOpt}(\flowVec) \leq
\frac{2}{3} \cdot \flowExtensionQOpt^{\sectorApprox}(\flowVec) \leq  \flowExtensionOpt^{\sectorApprox}(\flowVec).
\end{align}
In addition, since $\sectorApproxAlgo(\flowVec) = \sectorApproxDist(\flowVec) = \sectorApprox(\flowVec)$ (see Theorem~\ref{thm:main}~(\emph{i}) and Lemma~\ref{lem:opt-approx-n-equivalence}), we can conclude that
\begin{align}
\textstyle
\frac{2}{3} \cdot \flowExtensionOpt^{\sectorOpt}(\flowVec) \leq  \flowExtensionOpt^{\sectorApprox}(\flowVec) = \flowExtensionOpt^{\sectorApproxAlgo}(\flowVec)
\Rightarrow
\frac{\flowExtensionOpt^{\sectorApproxAlgo}(\flowVec)}{\flowExtensionOpt^{\sectorOpt}(\flowVec)} \geq \frac{2}{3},
\end{align}
and Theorem~\ref{thm:main}~(\emph{ii}) follows directly.
\end{proof}
\begin{remark}[Complexity of {\algoName}]
Recall from the proof of Theorem~\ref{thm:main} that $\decisionThresholdCrit \in [ \max_{\edge \in \delta(n)} \flow_{\edge}, \sum_{\edge \in \delta(n)} \flow_{\edge} ]$. Let $\Theta := (\sum_{\edge \in \delta(n)} \flow_{\edge} - \max_{\edge \in \delta(n)} \flow_{\edge})/\epsilon$, the binary search process  of {\algoName} will terminate in $O(\log_{2}\Theta)$ iterations, and each iteration has a complexity of $O(|\delta(n)|^2)$ (see Algorithm~\ref{algo:decision-problem}). Therefore, the complexity of {\algoName} is $O(|\delta(n)|^2 \cdot \log_{2}\Theta) = O(\numNode^{2} \cdot \log_{2}\Theta)$.
% Towards calculating the complexity of algorithm \algoName, for every node $n$, we assume that instead of using the error $\epsilon$ found in \algoName, we want to find the value of $\decisionThresholdCrit$ up to a precision that depends on the precision we use to store the flows in every link of $\graph$. Let us denote this precision as $\mathcal{A}$.
\end{remark}

% Furthermore, from definition it holds that  $\flowExtensionOpt^{\sector}(\flowVec)
% \le \flowExtensionQOpt^{\sector}(\flowVec), \forall \sector \in \setSector$.
% Consider Lemma \ref{lemma:sasaki_approx} % and inequalities \ref{equation_approx_optimality}, \ref{equation_definition} 
% we deduce that $\flowExtensionOpt^{\sectorOpt}(\flowVec) \leq \flowExtensionQOpt^{\sectorOpt}(\flowVec)$, and $\frac{2}{3} \flowExtensionQOpt^{\sectorApprox}(\flowVec) \leq  \flowExtensionOpt^{\sectorApprox}(\flowVec)$. 
% The result follows from the above inequalities combined with inequality \ref{equation_approx_optimality}.

%% remark begins
\begin{remark}
\label{remark:bound2}
Based on Lemma~\ref{lemma:sasaki_approx}, we can obtain another lower bound of the approximation ratio of {\algoName} given by:
\begin{align}
\textstyle
\frac{\flowExtensionOpt^{\sectorApproxAlgo}(\flowVec)}{\flowExtensionOpt^{\sectorOpt}(\flowVec)}
\geq \frac{1}{1 + \flowExtensionOpt^{\sectorOpt}(\flowVec) \cdot \max_{\edge \in \setEdge} \flow_{\edge} }
\geq \frac{1}{1 + \flowExtensionQOpt^{\sectorApproxAlgo}(\flowVec)\cdot \max_{\edge \in \setEdge} \flow_{\edge} }
:= LB^{\sectorApproxAlgo}(\flowVec).
\label{eqn:bound_2}
\end{align}
\end{remark}
%% remark ends

\begin{proof}
    Lemma \ref{lemma:sasaki_approx}(ii) for $\sigma = \pi$ gives, 
    \begin{equation}
        \label{eq:remark_eq1}
        \frac{1}{\flowExtensionQOpt^{\pi}(\flowVec)} \geq  \frac{1}{\flowExtensionOpt^{\pi}(\flowVec)} - \max_{\edge \in \setEdge} \flow_{\edge}.
    \end{equation} 
    By definition of sectorization $\pi$ and Lemma \ref{lemma:sasaki_approx}(i), we have $\mu^{\pi}(\flowVec) \ge \mu^{\sector^{*}}(\flowVec) \ge \lambda^{\sector^{*}}(\flowVec)$ and hence $\frac{1}{\lambda^{\sigma^{*}}(\flowVec) } \ge \frac{1}{\mu^{\pi}(\flowVec)}$. Therefore, \eqref{eq:remark_eq1} becomes
    \begin{equation}
        \label{eq:remark_eq2}
          \frac{1}{\flowExtensionOpt^{\pi}(\flowVec)} - \max_{\edge \in \setEdge} \flow_{\edge} \le \frac{1}{\lambda^{\sigma^{*}}(\flowVec)} \Leftrightarrow \frac{\lambda^{\sigma^{*}}(\flowVec)}{\flowExtensionOpt^{\pi}(\flowVec)}   \le 1 + \lambda^{\sigma^{*}}(\flowVec) \cdot \max_{\edge \in \setEdge}  \flow_{\edge}
        \nonumber
    \end{equation}
    which gives the final result:
    \begin{equation}
        \frac{\flowExtensionOpt^{\sectorApproxAlgo}(\flowVec)}{\flowExtensionOpt^{\sectorOpt}(\flowVec)} \ge \frac{1}{1 + \flowExtensionOpt^{\sectorOpt}(\flowVec) \cdot \max_{\edge \in \setEdge} \flow_{\edge} }.
        \nonumber
    \end{equation}
    The second inequality of \eqref{eqn:bound_2} follows using the inequality $\mu^{\pi}(\flowVec) \ge \mu^{\sector^{*}}(\flowVec) \ge \lambda^{\sector^{*}}(\flowVec)$.
    
\end{proof}

This bound is useful since, from the definition of ${\flowExtensionOpt^{\sectorOpt}(\flowVec)}$, i.e., the extension ratio of $\flowVec$ for it ``hit" the boundary of $\polytope_{H^{\sectorOpt}}$, the values of $\flowExtensionOpt^{\sectorOpt}(\flowVec)$ and $\max_{\edge \in \setEdge^{\sector}} \flow_{\edge}$ can not be large simultaneously. Although $\flowExtensionOpt^{\sectorOpt}(\flowVec)$ is analytically intractable, since $\flowExtensionOpt^{\sectorOpt}(\flowVec) \leq \flowExtensionQOpt^{\sectorOpt}(\flowVec) \leq \flowExtensionQOpt^{\sectorApproxAlgo}(\flowVec)$, we can derive another lower bound of the approximation ratio, $LB^{\sectorApproxAlgo}(\flowVec)$, which depends on $\flowExtensionQOpt^{\sectorApproxAlgo}(\flowVec)$. Note that $LB^{\sectorApproxAlgo}(\flowVec)$ is tractable since $\max_{\edge \in \setEdge^{\sector}} \flow_{\edge}$ is independent of the sectorization and $\flowExtensionQOpt^{\sectorApproxAlgo}(\flowVec)$ can be explicitly computed {\eqref{eqn:flow-extension-Q}}. In other words,
\begin{align}
\textstyle
\max \big\{ \frac{2}{3}, LB^{\sectorApproxAlgo}(\flowVec) \big\}
\leq \frac{\flowExtensionOpt^{\sectorApproxAlgo}(\flowVec)}{\flowExtensionOpt^{\sectorOpt}(\flowVec)}
\leq 1.
\end{align}
In Section~\ref{ssec:evaluation-random-network}, we show that for small values of $\numSector_{n}$, $LB^{\sectorApproxAlgo}(\flowVec)$ is indeed a much tighter bound than $2/3$ in {\eqref{eqn:thm-approx-ratio}}.

\vspace{0.5ex}
\noindent\textbf{Discussions and Applications.}
The distributed {\algoName} algorithm approximates the optimal sectorization, $\sectorOpt(\flowVec)$, which maximizes the flow extension ratio, $\flowExtensionOpt^{\sector}(\flowVec)$, under given $\flowVec$ and $\numSectorVec \in \mathbb{Z}_{+}^{\numNode}$. The choice of a sectorization should be based on a network flow, $\flowVec$ as the polytope $\polytope_{\graphAux}$ can be augmented in different flow directions depending on $\sector$. Hence, different sectorizations favor different sets of network flows. Some discussions about the proposed optimization framework:

\begin{itemize}[leftmargin=*, topsep=3pt, itemsep=3pt]
% \item
% \emph{Dynamic Sectorization based on Backpressure}.
% With an (unknown and/or time-varying) arrival rate matrix $\arrivalRateMat \in \textrm{int}(\capRegion(\graphNet))$ and a given sectorization $\sector \in \setSector(\numSectorVec)$, the dynamic backpressure algorithm will converge to and return a network flow $\flowVec \in \polytope_{\graphAux}$. Using the proposed framework, one can find the sectorization that approximates the best sectorization with respect to $\flowVec$. The rationale behind this is that the sectorized network will be able to maintain arrival rates proportionally higher than $\arrivalRateMat$. Moreover, when ${\flowExtensionOpt^{\sectorApprox}(\flowVec)}$ is analytically tractable, it can provide information about how much $\flowVec$ can be extended until it intersects with the boundary of the matching polytope $\polytope_{H^{\sectorApprox}}$. Therefore, the proposed framework can enable \emph{dynamic sectorization} of the network to adapt to every network flow $\flowVec$ obtained by the backpressure algorithm, including in scenarios with time-varying arrival rates, $\arrivalRateMat$. This sectorization technique is demonstrated in Section ~\ref{sec:evaluation}.
%%
\item
\emph{Known Arrival Rates}.
In the case with single-hop traffic, the capacity region of a network $\graphNet$ is given by $\capRegion(\graphNet) =  \convexHull(\scheduleSet_{\graphNet})$ (see {\eqref{eqn:cap-region}}). With a known arrival rate matrix $\arrivalRateMat$, {\algoName} augments the capacity region with respect to the required $\arrivalRateMat$. Similarly, in the case of multi-hop traffic, with a known $\arrivalRateMat$, one can first obtain a feasible multi-commodity network flow $\flowVec$ that supports $\arrivalRateMat$, and then augment the polytope $\polytope_{\graphAux}$ according to this $\flowVec$.
\item
\emph{Varying the Number of Sectors}.
With proper (minor) modifications, {\algoName} can also return the minimum number of sectors $\numSector_n$ for every node $n$ such that a given network flow, $\flowVec$, can be maintained by the network, i.e., $\flowVec \in \polytope_{H^{\sectorApprox}}$. This is due to the distributed nature the proposed optimization framework, {\optApproxN}. Compared to previous work (e.g.,~\cite{hajek1988link}) whose objective is to obtain the minimum rate required for a network to support a flow vector $\flowVec$, our proposed framework also provides a method to support $\flowVec$ via efficient sectorization based on available resources.
\iffullpaper
\item
\emph{Changing the Effective Topology of the Network}.
A natural question that arises after developing the proposed optimization framework is whether the network can be sectorized in a way to embed a specific useful structural property in its auxiliary graph, such that the underlying links between nodes can be changed. Since in the sectorized network, the routing and scheduling schemes can be run conceptually in the auxiliary graph rather than in the connectivity graph, by choosing specific sectorizations the network architect might be able to induce properties in the graph to boost the performance of the algorithms to be executed in the network. For example, some sectorization schemes could potentially divide the network into conceptual non-communicating clusters or make the auxiliary graph bipartite.
\else
\fi
\iffullpaper
% \item \emph{Inducing Bipartiteness in $\graphAux$.}
% As a special case of changing the effective topology, in the following section, we present the class of \emph{Even Homogeneous Sectorizations}, under which the auxiliary graph of a network, $\graphAux$, is guaranteed to be bipartite. This is of particular interest because of two main reasons. 
% First, recall from Section~\ref{ssec:auxiliary-background-matching-polytopes} that the matching polytope and fractional matching polytope of a bipartite graph $\graph_{bi}$ are equivalent, i.e., $\polytope_{\graph_{bi}} = \polytopeQ_{\graph_{bi}}$. A sectorized network with a bipartite auxiliary graph, $\graphAux$, (and an unsectorized network with a bipartite connectivity graph) have a significantly simplified explicit capacity region characterization given by its fractional matching polytope $\polytopeQ_{\graphAux}$, since $\polytope_{\graphAux} = \polytopeQ_{\graphAux}$ (see {\eqref{eqn:polytope-Q}} and Lemma~\ref{lem:schedule-matching}).
% Second, the maximum weight matching (MWM) in $\graphAux$ used as part of the backpressure algorithm, described in Section~\ref{sec:flow-extension}, can be obtained more efficiently when $\graphAux$ is a bipartite graph (see~\cite{schrijver2003combinatorial} for a detailed survey). Moreover, the MWM in a bipartite graph can be obtained in a distributed manner~\cite{wattenhofer2004distributed}, which is suitable for wireless networks where nodes have only local information.
\end{itemize}
\else
\fi

\iffullpaper
\section{Bipartite Auxiliary Graphs \& Even Homogeneous Sectorization}
\label{sec:bipartite}

% \iffullpaper
% IF FULL PAPER
% \else
% IF NOT FULL PAPER
% \fi

As a special case of changing the effective topology using sectorization, in this section, we present the class of \emph{Even Homogeneous Sectorizations}, under which the auxiliary graph of a network, $\graphAux$, is guaranteed to be bipartite. 

% We present the motivation and benefits of sectorizing a network such that its auxiliary graph is  \emph{bipartite}, and a class of sectorizations that guarantee bipartite auxiliary graphs.

% This is of particular interest because of two main reasons. 
% First, recall from Section~\ref{ssec:auxiliary-background-matching-polytopes} that the matching polytope and fractional matching polytope of a bipartite graph $\graph_{bi}$ are equivalent, i.e., $\polytope_{\graph_{bi}} = \polytopeQ_{\graph_{bi}}$. A sectorized network with a bipartite auxiliary graph, $\graphAux$, (and an unsectorized network with a bipartite connectivity graph) have a significantly simplified explicit capacity region characterization given by its fractional matching polytope $\polytopeQ_{\graphAux}$, since $\polytope_{\graphAux} = \polytopeQ_{\graphAux}$ (see {\eqref{eqn:polytope-Q}} and Lemma~\ref{lem:schedule-matching}).
% Second, the maximum weight matching (MWM) in $\graphAux$ used as part of the backpressure algorithm, described in Section~\ref{sec:flow-extension}, can be obtained more efficiently when $\graphAux$ is a bipartite graph (see~\cite{schrijver2003combinatorial} for a detailed survey). Moreover, the MWM in a bipartite graph can be obtained in a distributed manner~\cite{wattenhofer2004distributed}, which is suitable for wireless networks where nodes have only local information.

\subsection{Motivation for Inducing Bipartiteness in $\graphAux$}

% Note that the mathematical formulation and analysis of a sectorized network, $\graphNet$, rely on the underlying structure of its auxiliary graph, $\graphAux$. One question that naturally arises is whether the network can be sectorized in a way to embed a specific useful structural property in its auxiliary graph. If the answer to this question is positive, network sectorization can not only lead to an improved capacity, but also change the effective underlying connections between nodes in the network.

For a network $\graph$ and $\numSectorVec \in \mathbb{Z}_{+}^{\numNode}$, let $\setSector_{bi}(\numSectorVec)$ be the set of network sectorizations such that the auxiliary graph $\graphAux$ is a bipartite graph, $\forall \sector \in \setSector_{bi}(\numSectorVec)$. We are particularly interested in the network sectorizations $\setSector_{bi}(\numSectorVec)$ because of two main reasons.
First, recall from section \ref{ssec:auxiliary-background-matching-polytopes} that the matching polytope and fractional matching polytope of a bipartite graph $\graph_{bi}$ are equivalent, i.e., $\polytope_{\graph_{bi}} = \polytopeQ_{\graph_{bi}}$. A sectorized network with a bipartite auxiliary graph, $\graphAux$, (and an unsectorized network with a bipartite connectivity graph) have a significantly simplified explicit capacity region characterization given by its fractional matching polytope $\polytopeQ_{\graphAux}$, since $\polytope_{\graphAux} = \polytopeQ_{\graphAux}$ (see {\eqref{eqn:polytope-Q}} and Lemma~\ref{lem:schedule-matching}).
Second, the maximum weight matching (MWM) in $\graphAux$ used as part of the backpressure algorithm, described in section \ref{sec:flow-extension}, can be obtained more efficiently when $\graphAux$ is a bipartite graph (see~\cite{schrijver2003combinatorial} for a detailed survey). Moreover, the MWM in a bipartite graph can be obtained in a distributed manner~\cite{wattenhofer2004distributed}, which is suitable for wireless networks where nodes have only local information. In that way, the routing can be achieved in a distributed manner.

One intuitive way to construct a network sectorization $\sector \in \setSector_{bi}(\numSectorVec)$ is to set $\numSector_{n} = |\delta(n)|, \forall n$ and put each link in $\delta(n)$ in a separate sector, i.e., $\sector_{n} = \{ \vert \link_{1} \vert \link_{2} \vert \cdots \vert \link_{|\delta(n)|}\}, \forall n$. However, we might not be able to afford that many sectors for every node. Next, we present a class of network sectorizations under which the auxiliary graph of the sectorized network is \emph{guaranteed} to be bipartite for any even $K_1 = \dots = K_N$, therefore enjoying the benefits described above.

% An observation towards the aforementioned goal, is that for sufficiently large number of sectors, every node can put every emanating link in a different sector, something that will trivially make the auxiliary graph bipartite. \tingjun{not trivial for me...} \panos{Basically the network will consist of links that will have two nodes each. Everything will be isolated.}

%%%%%
%%%%%
\subsection{Even Homogeneous Sectorization}

We now present the class of \emph{Even Homogeneous Sectorizations}, denoted by $\setHomSector(\numSector) \subseteq \setSector(\numSectorVec)$ for $\numSector \in \mathbb{Z}_{+}$ and $\numSectorVec = (K, \cdots, K) \in \mathbb{Z}^{\numNode}_{+}$. For a network $\graphNet$ with sectorization $\sector \in \setHomSector(\numSector)$:
(\emph{i}) every node has the same even number of sectors, i.e., $\numSector_{n} = \numSector, \forall n$ with an even number $\numSector$,
(\emph{ii}) all the $\numSector$ sectors of a node have an equal FoV of $360^{\circ}/\numSector$, and
(\emph{iii}) at least one sectoring axis (and hence all sectoring axes) are parallel across the nodes.
Although (\emph{ii}) and (\emph{iii}) may seem restrictive, based on the notion of equivalent sectorization described in section \ref{ssec:auxiliary-equivalent-sectorizations}, our results that hold for all sectorizations in $\setHomSector(\numSector)$ also hold for their corresponding equivalent sectorizations in $\setSector(\numSectorVec)$.

% \subsubsection{Bipartiteness \& Complexity Decrease in Even Homogeneous Sectorizations}

The following lemma shows that for a sectorized network $\graphNet$ with $\sector \in \setHomSector(\numSector)$, its auxiliary graph $\graphAux$ is bipartite.

%% lemma begins
\begin{lemma}
\label{lemma:k_disconnected_bipartite_clusters}
For a sectorized network $\graphNet$ with an Even Homogeneous Sectorization $\sector \in \setHomSector(\numSector)$ and $\numSector \in \mathbb{Z}_{+}$, its auxiliary graph $\graphAux$ is bipartite. Moreover, $\graphAux$ consists of $\numSector/2$ isolated bipartite graphs.
\end{lemma}
%% lemma ends

% We define the sectoring axis between sector $K$ and sector $1$ with $\nodeSectorAxis{n}{K}$.

%% proof begins
\begin{proof}
Without loss of generality, we assume that the parallel sectoring axes across different nodes are labeled with the same index, i.e., for a given $k \in [\numSector]$, $\{ \nodeSectorAxis{n}{k}: \forall n \}$ is a set of parallel sectoring axes. Recall from section \ref{ssec:model-network} that for node $n$, we call the boundary of two adjacent sectors $k$ and $(k+1)$ \emph{a sectoring axis} and denote it by $\nodeSectorAxis{n}{k}$. We define the sectoring axis between sector $K$ and sector $1$ with $\nodeSectorAxis{n}{K}$. We also assume that for node $n$, its sectoring axes are indexed from $k = 1$ to $\numSector$ clockwise.
%, and $\nodeSector{n}{\numSector} = \nodeSector{n}{0}$ due to the cyclic nature of $\sector_{n}$. 
Using simple geometric calculations, it can be shown that under even homogeneous sectorization $\sector \in \setHomSector(\numSector)$, the $k^{\textrm{th}}$ sector of node $n$, $\nodeSector{n}{k}$, can only share a link in $\setLink^{\sector}$, with the $k'^{\textrm{th}}$ sector of node $n'$, $\nodeSector{n'}{k'}$, $\forall n' \in \setVertex, n' \neq n, \forall k \in [\numSector]$, where $k'$ \emph{necessarily} satisfies
\begin{align*}
k' = (k + \numSector/2) \mod \numSector.
\end{align*}
As a result, in the auxiliary graph $\graphAux = (\setVertex^{\sector}, \setEdge^{\sector})$ (see section \ref{ssec:auxiliary-graph-construction}), the vertices in $\setVertex^{\sector}$ associated with the $k^{\textrm{th}}$ sector of node $n$ can only be connected with the vertices in $\setVertex^{\sector}$ associated with the $k'^{\textrm{th}}$ sector of node $n' \in \setNode, n' \neq n$, and vice versa. Therefore, $\graphAux$ is a bipartite graph and consists of $\numSector/2$ completely isolated bipartite graphs.
\end{proof}
%% proof ends

Note that Lemma~\ref{lemma:k_disconnected_bipartite_clusters} has the following implications on the complexity of obtaining the MWM in the backpressure algorithm and hence the stabilization of the network. Let $NB(|\setVertex|, |\setEdge|, W)$ denote time complexity of obtaining the MWM in a general (non-bipartite) graph $\graph = (\setVertex, \setEdge)$, and $B(|\setVertex_{bi}|, |\setEdge_{bi}|, W)$ denote the time complexity of obtaining the MWM in a bipartite graph $\graph_{bi} = (\setVertex_{bi}, \setEdge_{bi})$, where $W := ||\mathbf{w}||_{\infty}$ is the maximum value in the integer-valued edge weight vector $\mathbf{w}$ of the corresponding graph. In our case, the weight on each edge is its backpressure
(see section \ref{ssec:model-capacity}). It is known that with the same values of $|\setVertex|$, $|\setEdge|$, and $W$, $B(|\setVertex|, |\setEdge|, W)$ is asymptotically better than $NB(|\setVertex|, |\setEdge|, W)$ (e.g., see~\cite{schrijver2003combinatorial}). We present the following theorem.

%% theorem begins
\begin{theorem}
\label{thm:speedup}
For an unsectorized network $\graph = (\setNode, \setLink)$ with $|\setNode| = \numNode$ and $|\setLink| = \numLink$, each iteration of the dynamic backpressure algorithm needs an algorithm of time complexity $NB(\numNode, \numLink, W)$ to obtain the MWM in $\graph$.
When the network is sectorized with an Even Homogeneous Sectorization $\sector \in \setHomSector(\numSector)$ to construct $\graphNet$ (with its auxiliary graph $\graphAux$), each iteration of the dynamic backpressure algorithm needs $\numSector/2$ distributed computations of an algorithm of $B(\numNode, \numLink, W)$ time complexity to obtain the MWM in $\graphAux$.
\end{theorem}
%% theorem ends

%% proof begins
\begin{proof}
For a sectorized network $\graphNet$ with even homogeneous sectorization $\sector \in \setHomSector(\numSector)$, in order to compute the asymptotic complexity of each iteration of the dynamic backpressure algorithm, we should think about each bipartite subgraph in the auxiliary graph, $\graphAux$, separately. Lemma~\ref{lemma:k_disconnected_bipartite_clusters} shows that $\graphAux$ can be divided into $\numSector/2$ isolated subgraphs, all of which are bipartitle. In every bipartite subgraph, each partition set consists of all the auxiliary nodes associated with the same sector in the sectorized network $\graphNet$. In particular, each bipartite subgraph is composed of two sets of vertices: the first set consists of the auxiliary nodes that correspond to a sector $k \in [\numSector]$, and the second set consist of the auxiliary nodes that correspond to the sector $(k + \numSector/2) \mod\numSector$. Therefore, each bipartite subgraph consists of $\frac{2 |\setVertex^{\sector}|}{\numSector} = \frac{2 \numSector |\setNode|}{\numSector} = 2 \numNode$ vertices. For the edges, we can deduce that the $|\setLink| = \numLink$ number of edges of $\graphNet$ are divided and distributed across the $\numSector/2$ bipartite subgraphs. The specific distribution of the edges between the subgraphs depends on the network sectorization and the topology of $\graphNet$.

For the dynamic backpressure algorithm in $\graphNet$ in each time slot, we should simply calculate the MWM of these $\numSector/2$ isolated bipartite subgraphs. Let us distinguish these isolated bipartite subgraphs by subscript $b \in [\numSector/2]$. Let $\setVertex_{b}^{\sector}$, $\setEdge_{b}^{\sector}$, and $W_b$ denote the set of nodes, set of edges, and the maximum edge weight in the $b^{\textrm{th}}$ isolated bipartite subgraph of $\graphAux$, respectively. Then, obtaining the MWM in the $b^{\textrm{th}}$ bipartite subgraph required an algorithm with time complexity $B(|\setVertex_{b}^{\sector}|, |\setEdge_{b}^{\sector}|, W_b)$ where $|\setEdge_{n}^{\sector}| \leq \numLink$ and $W_{b} \leq W$. Note that since the $\numSector/2$ bipartite subgraphs are completely isolated, we can calculate the MWM in each of them in parallel.
\end{proof}
%% proof ends

Lemma~\ref{lemma:k_disconnected_bipartite_clusters} and Theorem~\ref{thm:speedup} show that the MWM can be obtained by a distributed algorithm among $K/2$ isolated bipartite subgraphs and, in addition, in every bipartite subgraph, the MWM can be obtained distributedly across the nodes in that subgraph (e.g., according to~\cite{wattenhofer2004distributed}). More importantly, it is implied that sectorizing the network with $\sector \in \setHomSector(\numSector)$ can \emph{simultaneously} improve the network capacity and reduce the complexity required for every iteration of the backpressure algorithm. Since the backpressure algorithm is dynamic and is executed in each time slot, we can achieve significant improvement in the overall efficiency when stabilizing the considered sectorized network. 
% In fact, surprisingly, with enough distributed resources, not only we decrease the complexity compared to the unsectorized network, but by adding sectors we also divide the subgraphs more and more resulting in smaller components that we need to find the maximum weight matchings of.

Having shown that the set of even homogeneous sectorizations, $\mathcal{H}$, is beneficial regarding the complexity and time required for stabilizing the network, we now take a second look at {\opt} and {\optApprox}. From section \ref{ssec:auxiliary-background-matching-polytopes} and Remark~\ref{remark:bipartiteequality}, when the network sectorization is optimized over $\setHomSector(\numSector) \subseteq \setSector(\numSectorVec)$, since $\polytope_{\graphAux} = \polytopeQ_{\graphAux}$, it holds that for a given network flow $\flowVec$ and $\forall \numSector \in \mathbb{Z}_{+}$,
\begin{align*}
\flowExtensionOpt^{\sector}(\flowVec) & = \flowExtensionQOpt^{\sector}(\flowVec), \forall \sector \in \setHomSector(\numSector),
\end{align*}
and therefore, the network sectorization optimized over the set $\setHomSector(\numSector)$, denoted by $\widehat{\sector}(\flowVec)$, is given by
\begin{align*}
\widehat{\sector}(\flowVec) := \arg \max_{\sector \in \setHomSector(\numSector)} \flowExtensionOpt^{\sector}(\flowVec) = \arg \max_{\sector \in \setHomSector(\numSector)} \flowExtensionQOpt^{\sector}(\flowVec).
\end{align*}
This reveals the trade-offs between the optimality of the obtained sectorization (i.e., $\widehat{\sector}$ compared with $\sectorOpt$) and the time complexity required for obtaining the MWM in the dynamic backpressure algorithm (i.e., optimization over the set of $\setHomSector(\numSector)$ compared with over the set of $\setSector(\numSectorVec)$).

% Since $\setHomSector(\numSector) \subseteq \setSector(\numSectorVec), \forall \numSector \in \mathbb{Z}_{+}$ and $\numSectorVec = (K, \cdots, K) \in \mathbb{Z}^{\numNode}_{+}$, a natural question is whether we can solve problem {\opt} which is equivalent with {\optApprox}, over the set $\setHomSector(\numSector)$, i.e., solve the problem:
% %%
% \begin{align*}
% \arg \max_{\sector \in \setHomSector(\numSector)}~ \flowExtensionQOpt^{\sector}(\flowVec).
% \end{align*}
% A little thought, reveals that due to the observation of section section \ref{ssec:auxiliary-equivalent-sectorizations}, and to the small size of the set $\setHomSector(\numSector)$, a brute-force like algorithm could solve the aforementioned problem. Since we decrease the optimization set, we expect a decrease in the optimal value. However, as we will prove in the next section, for every connectivity graph $\graph$,  the corresponding auxiliary graph, $\graphAux$, is bipartite $\forall \numSector \in \mathbb{Z}_{+}$ and $\forall \sector \in \setHomSector(\numSector)$. Furthermore, we show that with enough distributed computing resources, 
% sectorizing the network under $\setHomSector(\numSector)$ decreases the complexity needed for every step of the backpressure algorithm compared to the unsectorized network. 
\else
\fi

\iffullpaper
\section{Joint Routing, Scheduling, and Dynamic Sectorization}
\label{sec:routing}

In Section~\ref{sec:flow-extension}, we presented the throughput-optimal policy for a single sectorization scenario, and in Section~\ref{sec:optimization}, we optimize the choice of sectorization to maximize the capacity in a single direction. Additionally, in Section VII, we introduced the case of Even Homogeneous Sectorization, where all nodes are sectorized uniformly, ensuring that the auxiliary graph is bipartite. This bipartite structure not only simplifies scheduling but also accelerates network stabilization compared to unsectorized networks. These approaches demonstrate the power of sectorization, but they assume a fixed configuration that is not to be changed throughout the network operation.

Furthermore, recent advancements in electronically and/or mechanically steerable antenna arrays might enable the capability to change sectorization over time in a practical manner. This capability enables the network to dynamically leverage capacity increases in multiple flow directions simultaneously,\footnote{Allowing dynamic sectorization enables the network to achieve a capacity region that is the convex hull of the capacities associated with fixed sectorization configurations, effectively expanding the achievable throughput.} leading to an overall enhancement in end-to-end capacity by combining multiple optimal sectorization states instead of committing to a single configuration. This raises a fundamental question: What happens if the sectorizations can be changed dynamically? In such cases, a dynamic sectorization strategy that adapts to the network state is necessary to fully exploit the benefits of sectorized wireless networks.
Below, we present an opportunistic dynamic sectorization approach that leverages our proposed sectorization algorithm (Section~\ref{ssec:opportunistic_dynamic_sectorization_algorithm}), and provide a throughput optimal policy for joint sectorization, scheduling, and routing (Section~\ref{ssec:joint_throughput_optimality}).

\subsection{Opportunistic Dynamic Sectorization via Backpressure-Driven Optimization}
\label{ssec:opportunistic_dynamic_sectorization_algorithm}
This approach first leverages the backpressure algorithm to establish a current network flow. Once this flow is determined,\footnote{The network flow can be estimated by averaging observed link flows over time.} we apply our optimization framework to select the most effective sectorization configuration given the prevailing traffic conditions. Instead of performing sectorization jointly with routing and scheduling in each time slot, this method allows the backpressure algorithm to first indicate a network flow and depending on how frequently we can adjust the sectorization, we can optimize it.

More specifically, starting with an (unknown and/or time-varying) arrival rate matrix $\arrivalRateMat \in \textrm{int}(\capRegion(\graphNet))$ and a given sectorization $\sector \in \setSector(\numSectorVec)$, the dynamic backpressure algorithm will converge to and return a network flow $\flowVec \in \polytope_{\graphAux}$. Using the proposed framework, one can find the sectorization that approximates the best sectorization with respect to $\flowVec$. The rationale behind this is that the sectorized network will be able to maintain arrival rates proportionally higher than $\arrivalRateMat$. Moreover, when ${\flowExtensionOpt^{\sectorApprox}(\flowVec)}$ is analytically tractable, it can provide information about how much $\flowVec$ can be extended until it intersects with the boundary of the matching polytope $\polytope_{H^{\sectorApprox}}$. Therefore, the proposed framework can enable \emph{dynamic sectorization} of the network to adapt to every network flow $\flowVec$ obtained by the backpressure algorithm, including in scenarios with time-varying arrival rates, $\arrivalRateMat$.

While this method is not throughput-optimal, it provides a structured approach to optimizing sectorization periodically based on the current network state and as shown in Section~\ref{sec:evaluation} is faster than the throughput optimal algorithm proposed in the next section.

\subsection{Optimal Joint Sectorization, Routing, and Scheduling}
\label{ssec:joint_throughput_optimality}

In Section~\ref{sec:flow-extension}, we formulated backpressure routing as the joint throughput-optimal scheduling and routing with fixed sectorization. However, with dynamic sectorization, the per time-slot control space includes all possible sectorizations:
\begin{align}
\label{eq:dynamic-sectorization-schedule}
\scheduleVec^{\textrm{BP}}(t) 
& :=
\argmax\limits_{\scheduleVec \in \bigcup\limits_{\sector \in \setSector(\numSectorVec)} \scheduleSet_{\graphNet}}\ \big\{ \bpVec^{\top}(t) \cdot \scheduleVec \big\} 
\nonumber
\\
& =
\arg \max\limits_{\matchingVec \in \matchingSet_{H}^{\numSectorVec}}\ \big\{ \bpVec^{\top}(t) \cdot \matchingVec \big\}.
\end{align}

Here, $\bpVec(t)$ denotes the vector of queue differentials (i.e., backpressure) at time~$t$, 
$\setSector(\numSectorVec)$ represents all possible sectorizations (each node $n$ having $\numSectorVec_n$ sectors),
and $\scheduleSet_{\graphNet}$ is the set of all feasible schedules \emph{under a particular sectorization}.
Equivalently, we can view $\bigcup_{\sector\in\setSector(\numSectorVec)} \scheduleSet_{\graphNet}$
as the set of \emph{all} schedules that can be achieved by any valid sectorization choice, 
subject to the constraint that no node can activate more outgoing or incoming links than its allowed sector count. In the right-hand side of~\eqref{eq:dynamic-sectorization-schedule}, $\matchingSet_{H}^{\numSectorVec}$ denotes all \emph{edge-activation vectors} $\matchingVec$ (i.e., sets of simultaneously active edges) where each node~$n$ can activate at most $K_n$ links total. Each such $\matchingVec$ implicitly defines a feasible sectorization that permits those edges to be active simultaneously. Thus, we can choose an \emph{edge set} first, then retroactively assign a sectorization to realize that edge set in practice. $\matchingSet_{H}^{\numSectorVec}$ corresponds to the set of  \emph{$\bm{b}$-matchings} \cite{schrijver2003combinatorial} (for $\bm{b} :=  \numSectorVec$):
\emph{each node may participate in at most $b_n$ active edges.} The goal becomes selecting the subset of edges with the maximum total backpressure weight $\bpVec^{\top}(t)\cdot \matchingVec$ under the node-degree constraints.

A natural way to represent the above optimization is through the following Integer Linear Program (ILP). 
We introduce binary variables \(\schedule_{nm}\in\{0,1\}\) indicating whether the link~\((n,m)\) is active. 
Let \(w_{nm} = \bpVec_{nm}(t)\) be the backpressure weight for~\((n,m)\). 
Then the per-slot optimization becomes:
\begin{align}
\label{eq:ilp-dynamic-sector}
\max_{\{\schedule_{nm}\}} \quad 
& \sum_{(n,m)\,\in\,\setLink} w_{nm}\,\schedule_{nm}
\\[2pt]
\text{subject to} \quad
& \sum_{m \,:\, (n,m)\,\in\,\setLink}\,
  \schedule_{nm} 
  ~\le~ K_n,
  \quad
  \forall \,n\in\setNode,
\nonumber
\\[2pt]
& \schedule_{nm} \,\in\,\{0,1\},
  \quad
  \forall\,(n,m)\in\setLink.
\nonumber
\end{align}
The constraint \(\sum_{m}\schedule_{nm}\le K_n\) ensures that each node \(n\) cannot activate more links than its sector count, while the objective maximizes the total backpressure-weighted sum. In essence, \eqref{eq:ilp-dynamic-sector} describes a \emph{maximum weighted b-matching} in the graph of unsectorized nodes~\(H\). 
% This problem's constraints are \emph{totally unimodular}, so solving the \emph{LP relaxation} of \eqref{eq:ilp-dynamic-sector} already yields an \emph{integral} optimal solution and therefore suffices \cite{schrijver2003combinatorial}. 
We refer the reader to~\cite{schrijver2003combinatorial} for a detailed discussion of b-matchings in combinatorial optimization.
Solving \eqref{eq:ilp-dynamic-sector} exactly in each slot provides a \emph{throughput-optimal} dynamic sectorization policy. In practice, however, large-scale networks render per-slot exact solutions computationally prohibitive. Consequently, many systems could employ \emph{approximation methods} (e.g., LP-based heuristics, iterative algorithms, or greedy selection) and/or \emph{periodic re-sectorization} (solving the LP/Max-Weight at finite intervals~\cite{georgiadis2006resource}), preserving near-optimal throughput while substantially reducing overhead.

\else
\fi

\section{Evaluation}
\label{sec:evaluation}
% I changed these sections of the evaluation. I changed how we implement the random graphs to be more realistic as we talked about - you can compare with the previous version.}

We now evaluate the sectorization gain and the performance of the distributed approximation algorithm via simulations.
We focus on:
(\emph{i}) an example 7-node network, and
(\emph{ii}) random networks with varying number of nodes, number of sectors per node, and network flows.
For each network and a given network flow, $\flowVec$, we consider:
\begin{itemize}[leftmargin=*, topsep=3pt, itemsep=3pt]
\item
$\sectorApproxAlgo_{n}(\flowVec_{n})$: the sectorization of node $n$ returned by the distributed approximation algorithm, {\algoName} (Algorithm~\ref{algo:find-sector} in Section~\ref{ssec:optimization-algorithm}), and $\sectorApproxAlgo(\flowVec) = (\sectorApproxAlgo_{n}(\flowVec_{n}): \forall n \in \setNode)$ is the sectorization of all nodes.
\item
$\flowExtensionQOpt^{\sectorApproxAlgo}(\flowVec)$ and $\flowExtensionQOpt^{\varnothing}(\flowVec)$: the approximate flow extension ratios for the sectorized and unsectorized networks, respectively (Section~\ref{sec:flow-extension}).
\item
$\sectorizationGainQ^{\sectorApproxAlgo}(\flowVec)$: the approximate sectorization gain achieved by $\sectorApproxAlgo(\flowVec)$ (Section~\ref{sec:flow-extension}).
\end{itemize}
%%

%%%%%
%%%%%
\subsection{An Example 7-node Network}
\label{ssec:evaluation-example-network}

%% figure begins
\begin{figure}[!t]
\centering
% \vspace{-\baselineskip}
\includegraphics[width=0.95\columnwidth]{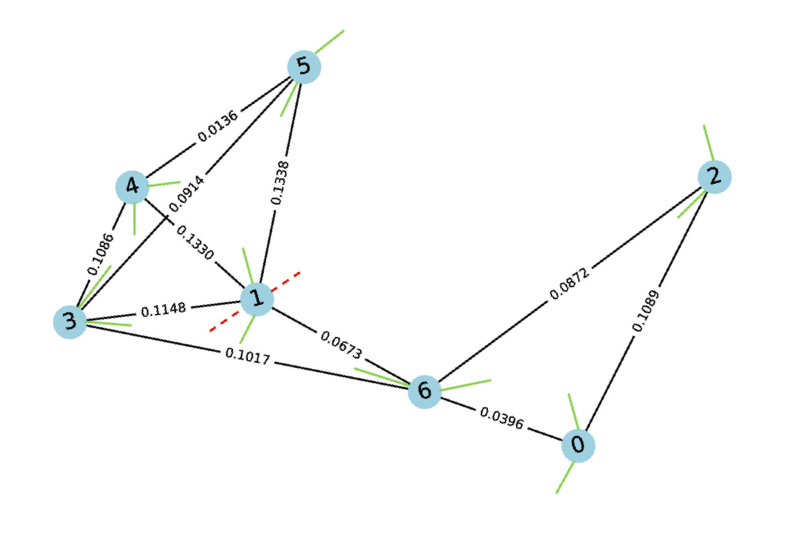}
\vspace{-3mm}
\caption{An example 7-node network: the connectivity graph with the network flow $\flowVec$ labeled on each edge. The green lines indicate the node sectorization, $\sectorApproxAlgo(\flowVec)$, obtained via {\algoName} with $\numSector_{n} = 2, \forall n$ with a sectorization gain of 1.83. The red dashed lines indicate a ``misconfigured'' sectorization of node 1 in the bottleneck phenomenon.}
\label{fig:example-network}
\vspace{-3mm}
\end{figure}
%% figure ends

%% figure begins
\begin{figure*}[!t]
\centering
\subfloat[]{
\includegraphics[height=1.5in]{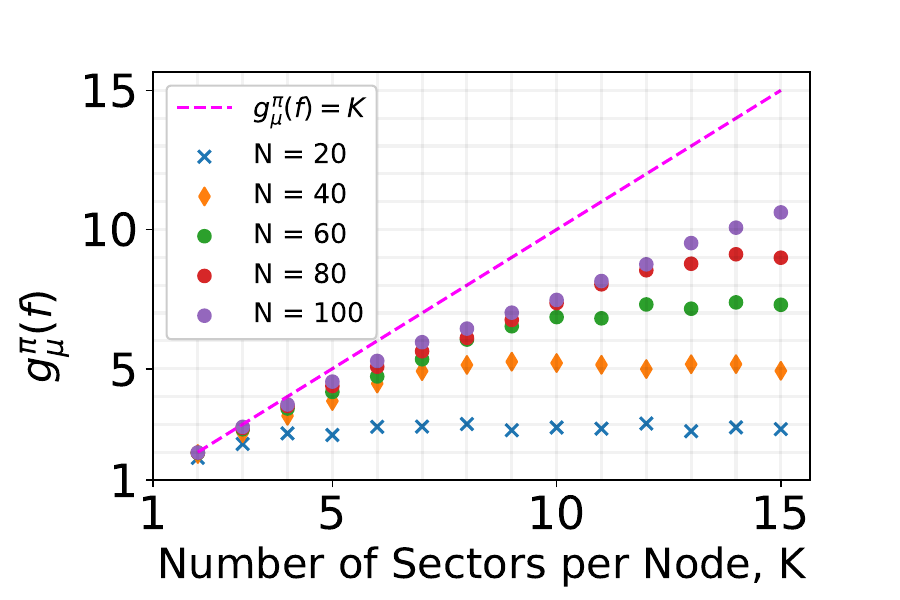}
\label{fig:eval-random-network-varying-N}}
% \hspace*{2mm}
\subfloat[]{
\includegraphics[height=1.5in]{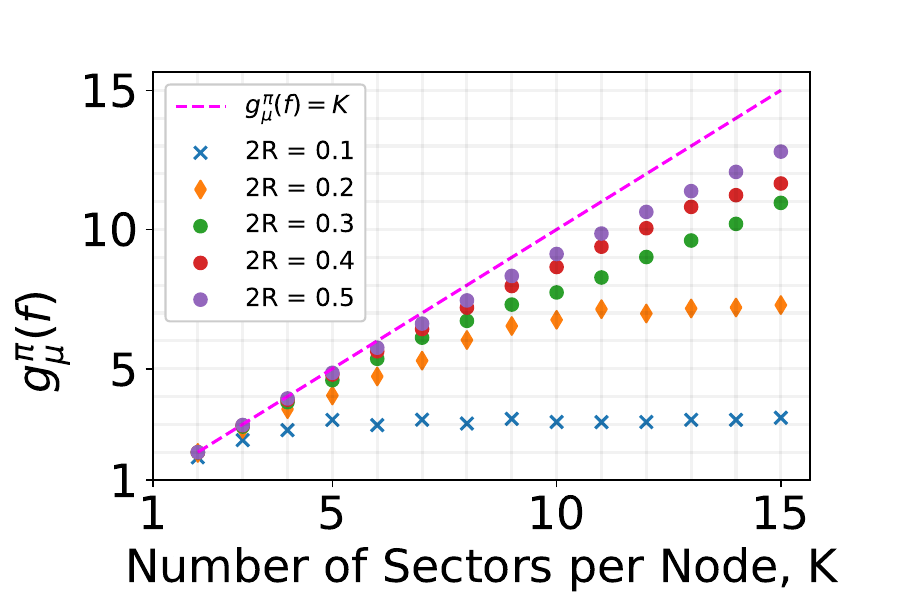}
\label{fig:eval-random-network-varying-density}}
% \hspace*{2mm}
\subfloat[]{
\includegraphics[height=1.5in]{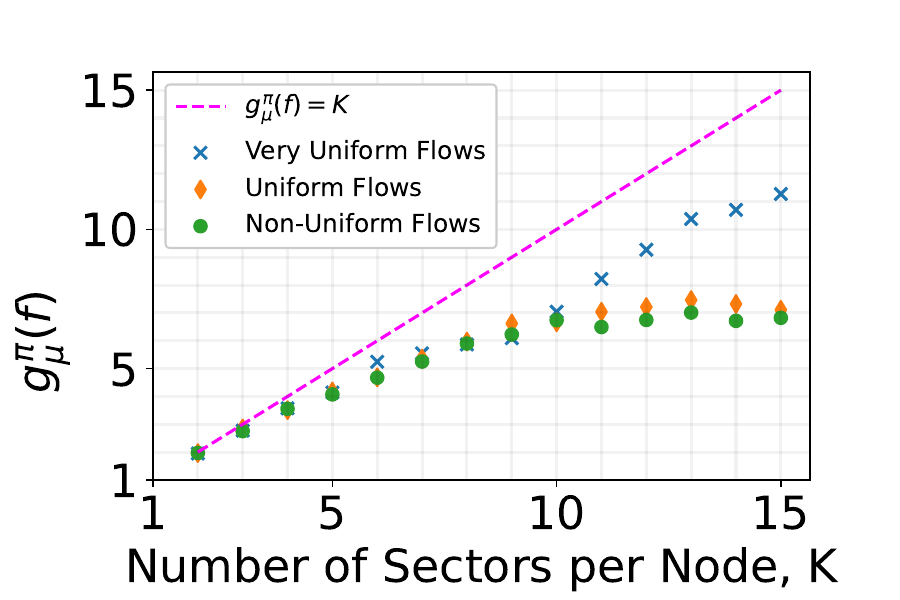}
\label{fig:eval-random-network-varying-flow}}
\vspace{-1mm}
\caption{The approximate sectorization gain, $\sectorizationGainQ^{\sectorApproxAlgo}(\flowVec)$, as a function of the number of sectors per node, $\numSector$, in random networks:
(a) varying number of nodes $\numNode \in \{20, 40, 60, 80, 100\}$ in networks with $2\nodeRange = 0.2$ and \emph{Uniform} flows,
(b) varying communication range $2\nodeRange \in \{0.1, 0.2, 0.3, 0.4, 0.5\}$ with \emph{Uniform} flows and $\numNode = 60$, and
(c) varying network flows (\emph{Non-uniform}, \emph{Uniform}, \emph{Very Uniform}) in networks with $\numNode = 60$ and $2\nodeRange = 0.2$.}
\label{fig:eval-random-network}
\vspace{-2mm}
\end{figure*}
%% figure ends

We consider a 7-node network, whose connectivity graph is shown in Fig.~\ref{fig:example-network}, with a network flow $\flowVec$ labeled on each edge and $\theta_{\textrm{th}} = 15.8^{\circ}$ (see Section~\ref{ssec:model-interference}). For tractability and illustration purposes, we set $\numSector_{n} = 2, \forall n$, and the green lines in Fig.~\ref{fig:example-network} indicate the sectorization $\sectorApproxAlgo(\flowVec)$ returned by {\algoName}.
For this relatively small network, we can explicitly compute the flow extension ratios for both the sectorized network, $\flowExtensionOpt^{\sectorApproxAlgo}(\flowVec) = \flowExtensionQOpt^{\sectorApproxAlgo}(\flowVec) = 4.06$, and the unsectorized network, $\flowExtensionOpt^{\varnothing}(\flowVec) = \flowExtensionQOpt^{\varnothing}(\flowVec) = 2.22$. Therefore, the approximate sectorization gain $\sectorizationGainQ^{\sectorApproxAlgo}(\flowVec)$ is equal to the explicit sectorization gain, i.e., $\sectorizationGainQ^{\sectorApproxAlgo}(\flowVec) = \sectorizationGainP^{\sectorApproxAlgo}(\flowVec) = 1.83$, which is close to $\numSector_{n} = 2$.

Note that optimizing the sectorization of each node under a given $\flowVec$ is critical, since the misplacement of the sectoring axes of even a single node can largely affect the achievable sectorization gain. We call this effect the \emph{bottleneck phenomenon} in sectorized networks, as illustrated by the following example. Considered the optimized sectorization $\sectorApproxAlgo(\flowVec)$ shown by the green lines in Fig.~\ref{fig:example-network}. If only the sectorization of node 1 is ``misconfigured" to be the red dashed lines, the sectorization gain is decreased from 1.83 to 1.22. This is also intuitive since with this misconfiguration, all three edges incident to node 1 with the highest flows are served by the same sector.

Since nodes 1, 3, and 6 in Fig.~\ref{fig:example-network} have a maximum node degree of 4, we also obtain a sectorization $\sectorApproxAlgo(\flowVec)$ by running {\algoName} with $\numSector_{n} = 4$. As expected, in the optimized sectorization for nodes 1, 3, and 6, one sectoring axis is put between every pair of adjacent edges. With this $\sectorApproxAlgo(\flowVec)$, we can also explicitly compute the flow extension ratios for the sectorized network $\flowExtensionOpt^{\sectorApproxAlgo}(\flowVec) = \flowExtensionQOpt^{\sectorApproxAlgo}(\flowVec) = 7.47$.
% \panos{here I put the sentence of the remark to the technical report.} 
\iffullpaper
This also confirms Remark \ref{remark:bipartiteequality} since the corresponding auxiliary graph $\graph^{\sectorApproxAlgo}$ is bipartite under $\sectorApproxAlgo(\flowVec)$. 
\else
\fi
With $\flowExtensionOpt^{\varnothing}(\flowVec) = \flowExtensionQOpt^{\varnothing}(\flowVec) = 2.22$, the sectorization gain is $\sectorizationGainQ^{\sectorApproxAlgo}(\flowVec) = \sectorizationGainP^{\sectorApproxAlgo}(\flowVec) = 3.36$. This example 7-node network demonstrates the performance and flexibility of {\algoName} for optimizing the deployment and configuration of sectorized networks based on the network flows.

%%%%%
\subsection{Random Networks}
\label{ssec:evaluation-random-network}

We now consider networks with randomly generated connectivity graphs, $\graph$. In particular, for each generated random geometric graph, $\numNode$ nodes are placed uniformly at random in a unit square area, and two nodes are joined by an edge if the distance between them is less than $2\nodeRange$.
%
% Unless otherwise stated, for the rest of the evaluation results, we assume that the connectivity graph, $\graphNet$, of a given instance of a sectorized mmWave network, is given by a random geometric graph with standard deviation $std(\graphNet)$ and mean $0$ inside a unit cube. In every instance, the neighbors of a node $n \in \setVertex$ are the nodes that lay inside a circle of radius $R = 0.19$. By fixing the radius, $R$, we can tune the density of the network by the parameter $std(\graphNet)$. In every section, we use the Monte Carlo method to evaluate every metric, using the average of $1000$ different instances of a random network.
%
We are interested in the effects of the following parameters of a random network on the sectorization gain:
\begin{itemize}[leftmargin=*, topsep=0pt]
\item
\textbf{Number of Nodes}, $\numNode$:
We consider random networks with different sizes of $\numNode \in \{20, 40, 60, 80, 100\}$, and the network density increases with larger values of $\numNode$.
\end{itemize}
\begin{itemize}[leftmargin=*, topsep=0pt]
\item
\textbf{Number of Sectors Per Node}, $\numSector_{n}$: We assume all nodes have an equal number of sectors, $\numSector_{n} = \numSector, \forall n$, with $\numSector \in \{2, 3, \dots, 15\}$.
\end{itemize}
\begin{itemize}[leftmargin=*, topsep=0pt]
\item
% \textbf{Node Range}, $\nodeRange$:
% With a given number nodes, $\numNode$, the density of a random network can be tuned by the range of every node, $\nodeRange$. In particular,
% with each value of $\numNode$, we consider
% % \emph{Sparse}, \emph{Normal Density}, and \emph{Dense} networks with
% $\nodeRange = 0.15, 0.3, \textrm{and}~0.5$, respectively.
\textbf{Communication Range}, $2\nodeRange$:
With a given number nodes, $\numNode$, the connectivity of the network can be tuned by the communication range between two nodes, $2\nodeRange$. We consider
% \emph{Sparse}, \emph{Normal Density}, and \emph{Dense} networks with
$2\nodeRange \in \{0.1, 0.2, \dots, 0.5\}$.
\end{itemize}
\begin{itemize}[leftmargin=*, topsep=0pt]
\item
\textbf{Uniformity of Network Flows}, $\phi$: For a network flow $\flowVec$, we define its \emph{uniformity} by $\phi := \max_{\edge} \flow_{\edge} / \min_{\edge} \flow_{\edge}$, i.e., $\flowVec$ is more uniform if its $\phi$ is closer to $1$. For a given value of $\phi$, random network flows $\flowVec$ can be generated as follows. First, each element of $\flowVec' = (\flow_{\edge}')$ is independently drawn from a uniform distribution between $[1, \phi]$. Then, $\flowVec$ is set to be $\flowVec'$ after normalization, i.e., $\flowVec = \flowVec' / |\flowVec|$. We consider \emph{Non-uniform}, \emph{Uniform}, and \emph{Very Uniform} network flows with $\phi = 1000, 10, \textrm{and}~1.1$, respectively.
\end{itemize}
For random networks with a large number of nodes, we only consider $\sectorizationGainQ^{\sectorApproxAlgo}(\flowVec)$ since it is computationally expensive to obtain $\sectorizationGainP^{\sectorApproxAlgo}(\flowVec)$, which is the true sectorization gain achieved by $\sectorApproxAlgo(\flowVec)$. However, from Lemma~\ref{lemma:sasaki_approx}, $\sectorizationGainQ^{\sectorApproxAlgo}(\flowVec)$ provides good upper and lower bounds on $\sectorizationGainP^{\sectorApproxAlgo}(\flowVec)$. The performance evaluation for each point is based on 1,000 instances of the random networks and their corresponding $\sectorApproxAlgo(\flowVec)$ obtained by {\algoName}.

%%%%%
% \subsubsection{Varying Number of Nodes, $N$}
\vspace{0.5ex}
\noindent\textbf{Varying Number of Nodes, $N$}.
Fig.~\ref{fig:eval-random-network}\subref{fig:eval-random-network-varying-N} plots the approximate sectorization gain, $\sectorizationGainQ^{\sectorApproxAlgo}(\flowVec)$, as a function of the number of sectors per node, $\numSector$, in a network with $2\nodeRange = 0.2$ and uniform flows ($\phi = 10$), with varying number of nodes, $\numNode$. Observe that $\sectorizationGainQ^{\sectorApproxAlgo}(\flowVec)$ increases sublinearly with respect to $\numSector$, and it approaches the identity line of $\sectorizationGainQ^{\sectorApproxAlgo}(\flowVec) = \numSector$ as $\numNode$ increases, which is as expected. Note that with a practical value of $\numSector$ (e.g., $\numSector \leq 6$), these networks can achieve $\sectorizationGainQ^{\sectorApproxAlgo}(\flowVec)$ that is almost equal to the number of sectors per node, $\numSector$. In addition, as the value of $\numSector$ increases, $\sectorizationGainQ^{\sectorApproxAlgo}(\flowVec)$ deviates from $\numSector$, which reveals a tradeoff point between the achievable sectorization gain ($\sectorizationGainQ^{\sectorApproxAlgo}(\flowVec)$) and complexity of network deployments ($\numNode$ and $\numSector$). In fact, for given parameters $\numNode$, $2\nodeRange$, and $\phi$, i.e., for a given density of the network, there exists a number of sectors that saturates the gain $\sectorizationGainQ^{\sectorApproxAlgo}(\flowVec)$. This is because after a sufficiently large number of sectors, the auxiliary graph of the network breaks down to isolated pairs of nodes.
\iffullpaper
This fact also indicates that this phase transition phenomenon will also hold for the explicit extension ratio, $\sectorizationGainP^{\sectorApproxAlgo}(\flowVec)$, since the graph is becoming eventually bipartite as we increase the number of sectors.
\else
\fi

%%%%%
% \subsubsection{Varying Node Density, $std(\graph)$}
\vspace{0.5ex}
\noindent\textbf{Varying Communication Range, $2\nodeRange$}.
Fig.~\ref{fig:eval-random-network}\subref{fig:eval-random-network-varying-density} plots the approximate sectorization gain, $\sectorizationGainQ^{\sectorApproxAlgo}(\flowVec)$, as a function of the number of sectors per node, $\numSector$, with $\numNode = 60$, \emph{Uniform} flows, $\phi = 10$, and varying communication ranges, $2\nodeRange$. It can be observed that $\sectorizationGainQ^{\sectorApproxAlgo}(\flowVec)$ increases sublinearly with respect to $\numSector$. As expected, with the same number of nodes, $\numNode = 60$, $\sectorizationGainQ^{\sectorApproxAlgo}(\flowVec)$ is closer to $\numSector$ as the range increases. 
% In fact, there is relationship between the parameters $\numNode$ and $2\nodeRange$; they both make the network more dense as they increase. The improved sectorization gains in networks with higher node densities stem from the fact that since each node has a larger number of neighboring nodes (and thus links), a larger value of $\numSector$ and the optimized sectorization can support a larger number of concurrent flows.
In fact, there is a relationship between the parameters $\numNode$ and $2\nodeRange$: they both increase the number of neighbors for every node. The improved sectorization gains stem from the fact that since each node has a larger number of neighboring nodes (and thus links), a larger value of $\numSector$ and the optimized sectorization can support a larger number of concurrent flows.

%%%%%
% \subsubsection{Varying Uniformity of Network Flows, $\phi$}
\vspace{0.5ex}
\noindent\textbf{Varying Uniformity of Network Flows, $\phi$}.
Fig.~\ref{fig:eval-random-network}\subref{fig:eval-random-network-varying-flow} plots $\sectorizationGainQ^{\sectorApproxAlgo}(\flowVec)$ as a function of the number of sectors per node, $\numSector$, in a network with $\numNode = 60$, $2\nodeRange = 0.2$, and varying network flow uniformity (\emph{Non-uniform}, \emph{Uniform}, and \emph{Very Uniform}). Overall, similar trends can be observed as those in Figs.~\ref{fig:eval-random-network}\subref{fig:eval-random-network-varying-N} and~\ref{fig:eval-random-network}\subref{fig:eval-random-network-varying-density}. In general, more uniform network flows lead to improved values of $\sectorizationGainQ^{\sectorApproxAlgo}(\flowVec)$, since non-uniform flows would have some larger flow components than the uniform ones, which on average increases the value of $\flowExtensionQOpt^{\sector}(\flowVec)$ ({see \eqref{eqn:flow-extension-Q}}). In addition, recall the bottleneck phenomenon described in Section~\ref{ssec:evaluation-example-network}, it is more beneficial to divide the flows of a node more equally across its sectors to achieve an improved sectorization gain.

% That means, that we achieve greater flow extension ratio, when we want to augment the network capacity region, towards a flow direction that is not near the boundaries in $\mathbb{R}^{|E|}$.

%% figure begins
\begin{figure}[!t]
\centering
\vspace{-3mm}
\subfloat{
\includegraphics[height=1.2in]{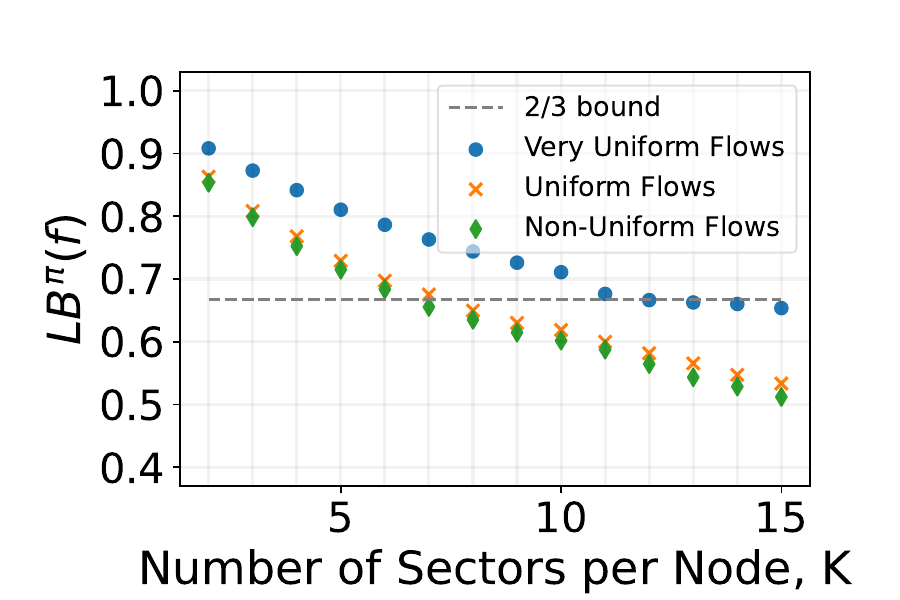}}
\hspace*{-8pt}
\subfloat{
\includegraphics[height=1.2in]{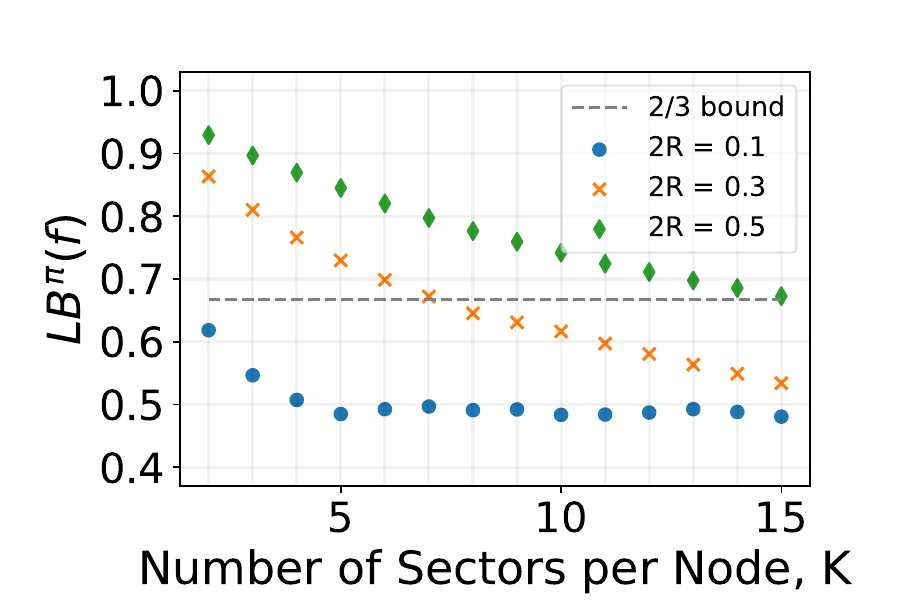}}
% \vspace{-1mm}
\caption{The lower bound of approximation ratio of {\algoName}, $LB^{\sectorApproxAlgo}(\flowVec)$ in {\eqref{eqn:bound_2}}, as a function of the number of sectors per node, $\numSector$, with $\numNode = 60$: (left) varying network flows in a network with $2R = 0.2$, and (right) varying communication ranges with \emph{Uniform} flows.}
\label{fig:eval-approx-bound}
\vspace{-2mm}
\end{figure}
%% figure ends

%%%%%
% \subsubsection{Evaluation of the Lower Bound $LB^{\sectorApproxAlgo}(\flowVec)$}
\vspace{0.5ex}
\noindent\textbf{Evaluation of the Lower Bound, $LB^{\sectorApproxAlgo}(\flowVec)$}.
Using simulations, we also evaluate the lower bound of the approximation ratio of {\algoName}, $LB^{\sectorApproxAlgo}(\flowVec)$ in {\eqref{eqn:bound_2}}, which depends on $\flowExtensionQOpt^{\sectorApprox}(\flowVec)$ and $\flowVec$. Fig.~\ref{fig:eval-approx-bound} plots the value of $LB^{\sectorApproxAlgo}(\flowVec)$ as a function of $\numSector$ with varying uniformity of the network flows and node densities. It can be observed that for small values of $\numSector$, $LB^{\sectorApproxAlgo}(\flowVec)$ is much higher than lower bound of $2/3$ provided by Theorem~\ref{thm:main}. 
In particular, the difference between the bounds increases dramatically and the approximation approaches the optimal for a small number of sectors. % and networks with nodes that have many neighboring nodes.
% In particular, $LB^{\sectorApproxAlgo}(\flowVec) > 0.8$ for $\numSector \leq 5$, and is further improved in dense networks with more uniform flows.
This is because in such networks, the maximum flow is expected to be small and hence $LB^{\sectorApproxAlgo}(\flowVec)$ can be improved when $\flowExtensionQOpt^{\sectorApprox}(\flowVec)$ remains the same.
% Although for greater number of sectors $2/3$ is more than $LB^{\sectorApproxAlgo}(\flowVec)$, algorithm {\algoName} is becoming optimal because the corresponding auxiliary graph is becoming bipartite.

% In that section, we evaluate numerically the bound presented in Remark \ref{remark:bound2}. That calculation is possible for various network settings, since it depends on $\flowExtensionQOpt^{\sectorApprox}(\flowVec)$ which can be calculated efficiently once we have run {\algoName} for every instance of the network. In Fig. \ref{fig:eval-approx-bound} we observe that for small number of sectors for every node, the approximation bound is greater than the theoretical constant bound of $2/3$. Specifically, for the networks studied, using less or equal than $5$ number of sectors for each node, {\algoName} gives an approximation bound at least $80 \%$. Note that in dense networks, where the flows are similar between the links, we can achieve greater approximation ratios. That is reasonable since in such networks we expect the maximum flow to be smaller and hence the bound of Remark \ref{remark:bound2} greater when $\flowExtensionQOpt^{\sectorApprox}(\flowVec)$ is not becoming larger. 

%% figure begins
\begin{figure}[!t]
\centering
\vspace{-2mm}
\subfloat{
\includegraphics[height=1.2in]{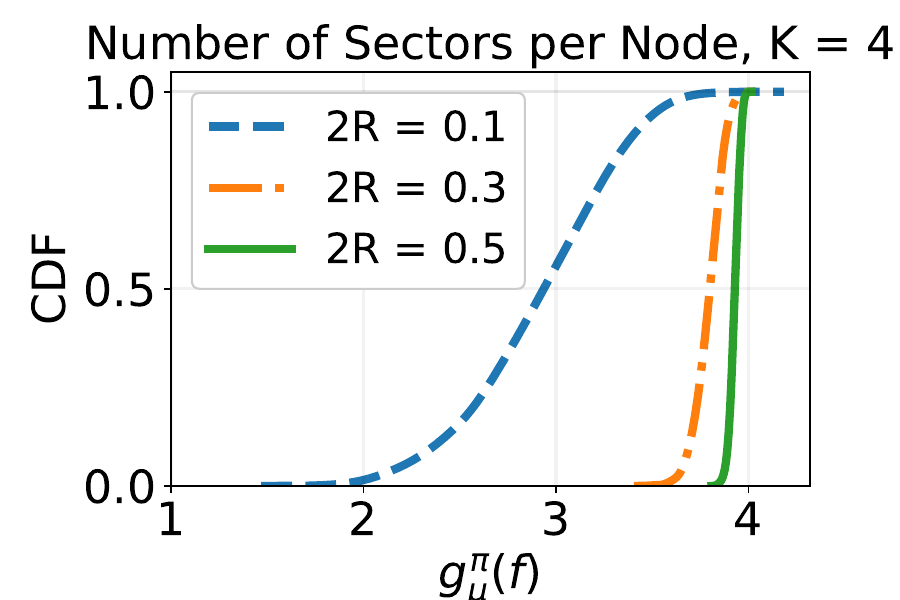}
\label{fig:cdf1}}
\hspace*{-16pt}
\subfloat{
\includegraphics[height=1.2in]{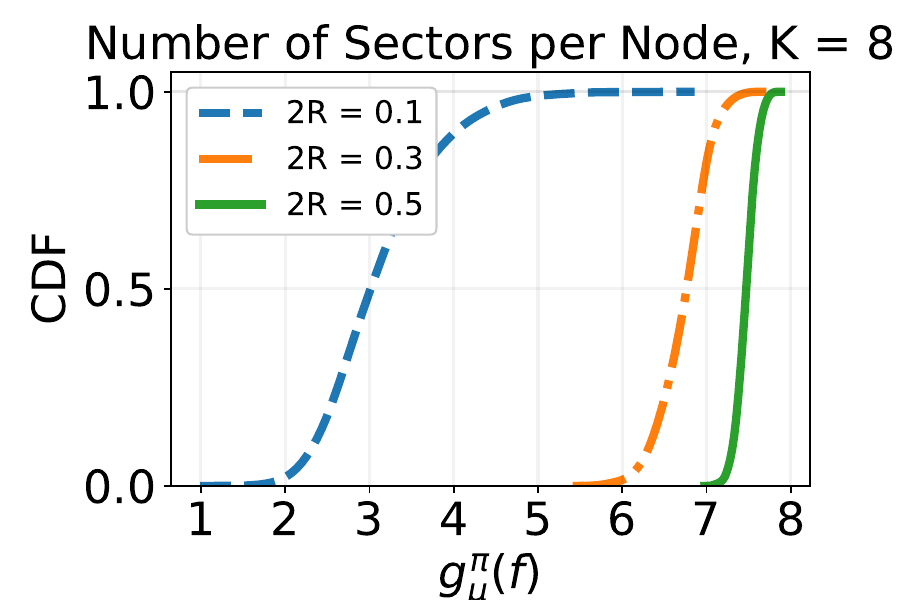}
\label{fig:cdf2}}
% \vspace{-1mm}
\caption{The cumulative distribution function (CDF) of $\sectorizationGainQ^{\sectorApproxAlgo}(\flowVec)$ with $\numNode = 60$ and \emph{Uniform} network flows.}
\label{fig:cdf}
\vspace{-2mm}
\end{figure}
%% figure ends

%%%%%
% \subsubsection{CDF of the Approximate Sectorization Gain, $\sectorizationGainQ^{\sectorApproxAlgo}(\flowVec)$}
\begin{figure*}[!ht]
\centering
\subfloat[]{
\includegraphics[height=1.2in]{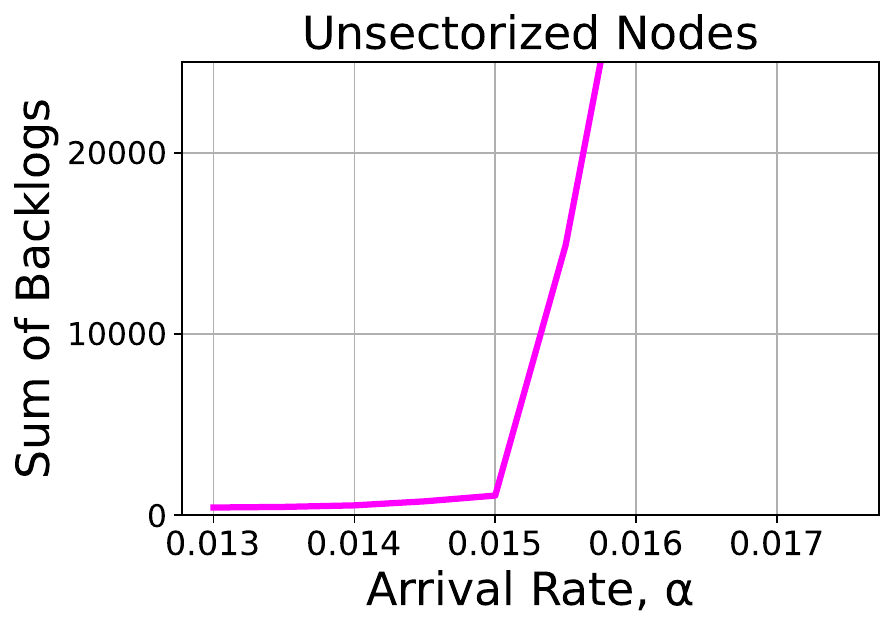}
\label{fig:eval-capacity_unsectorized}}
\hspace*{-8pt}
\subfloat[]{
\includegraphics[height=1.2in]{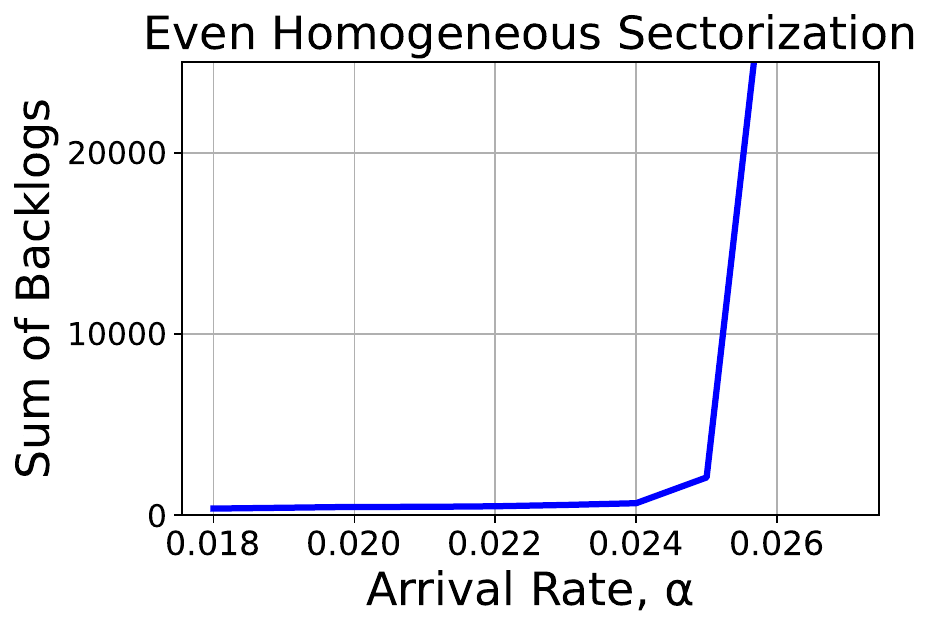}
\label{fig:eval-capacity_bipartite}}
\hspace*{-8pt}
\subfloat[]{
\includegraphics[height=1.2in]{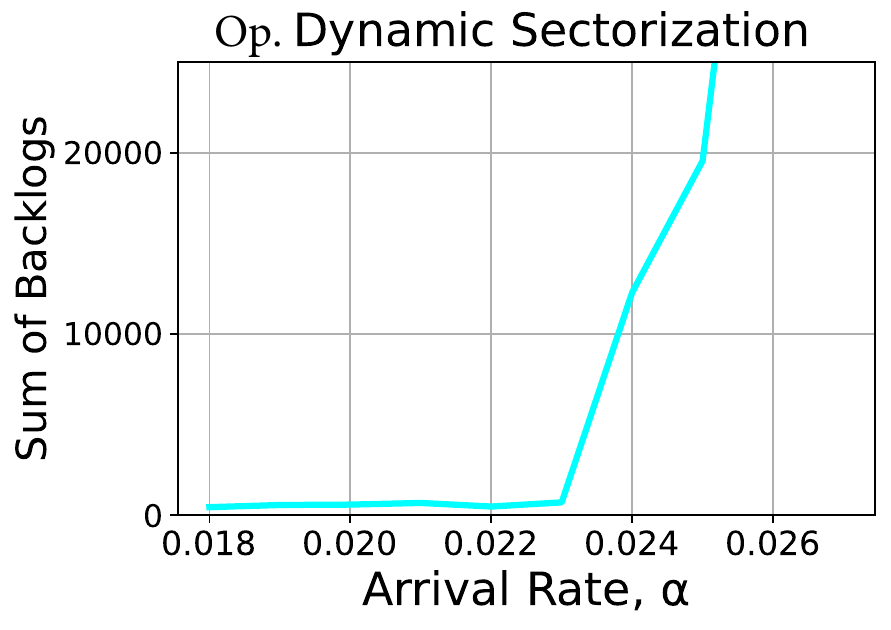}
\label{fig:eval-capacity_dynamic}}
\hspace*{-8pt}
\subfloat[]{
\includegraphics[height=1.2in]{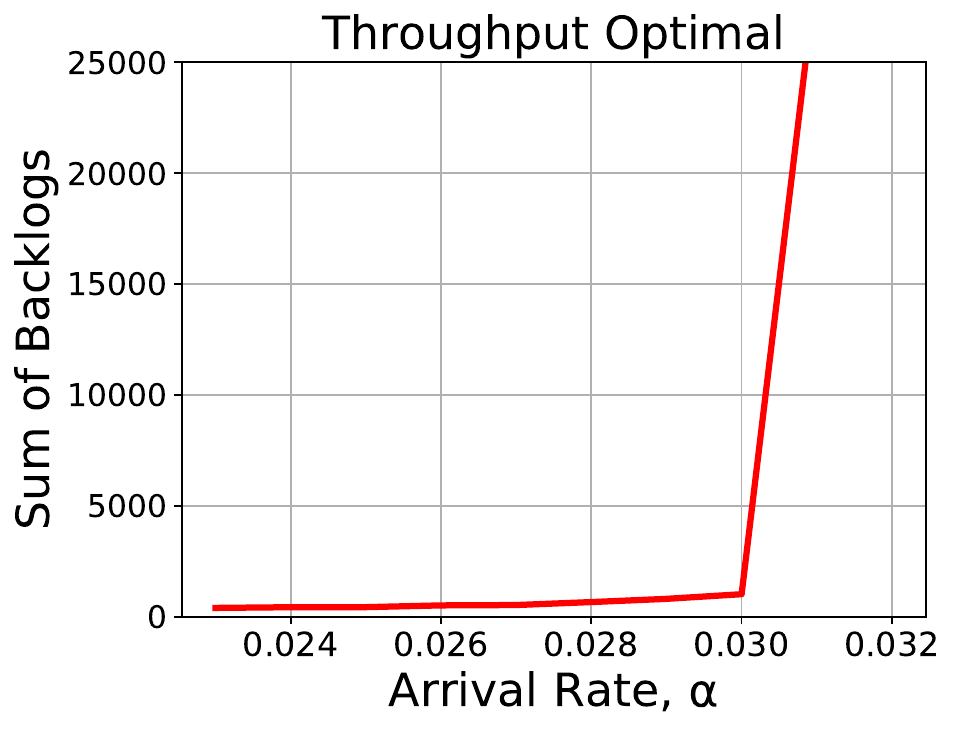}
\label{fig:eval-capacity_bipartite}}

\vspace{-1mm}
\caption{The summation of the backlogs of a 16x16 grid network given a uniform traffic rate that is expressed with the scalar $\alpha$ for (a) unsectorized nodes, (b) Even Homogeneous Sectorization, (c) opportunistic dynamic sectorization, and, (d) throughput optimal routing}.
\label{fig:eval-capacity}
\vspace{-2mm}
\end{figure*}
%% figure ends

\vspace{0.5ex}
\noindent\textbf{CDF of the Approximate Sectorization Gain, $\sectorizationGainQ^{\sectorApproxAlgo}(\flowVec)$}.
Finally, we evaluate the relationship between $\sectorizationGainQ^{\sectorApproxAlgo}(\flowVec)$ and the number of sectors, $\numSector$.  
Fig.~\ref{fig:cdf} plots the CDF of $\sectorizationGainQ^{\sectorApproxAlgo}(\flowVec)$ with $\numNode = 60$ and \emph{Uniform} network flows with varying communication ranges. It can be seen that for networks with $2\nodeRange = 0.3$, $\sectorizationGainQ^{\sectorApproxAlgo}(\flowVec)$ has a median value of 3.7/6.5 for $\numSector = 4/8$, respectively. This demonstrates the (sublinear) gain introduced by node sectorization, and this gain approaches the number of sectors per node, $\numSector$, as the underlying is more connected.
% \panos{let's delete the following since it is not true now:} In addition, across a total number of 1,000 instances of random network topologies and flows, $\sectorizationGainQ^{\sectorApproxAlgo}(\flowVec)$ is always greater than 3 and 5 for $\numSector = 4$ and 8, respectively, demonstrating the (sublinear) gain introduced by node sectorization.

% In the previous section, we presented an average approximation value for the flow extension ratio, i.e., $\frac{\flowExtensionQOpt^{\sectorApprox}(\flowVec)}{\flowExtensionQOpt^{\emptyset}(\flowVec)}$, using the Monte Carlo method. In this section, we show the distribution of this ratio throughout the $1000$ iterations of the Monte Carlo. In Fig. \ref{fig:cdf1}, where we sectorize every node with up to $4$ sectors, we note that most of the times the ratio is greater than $3.8$. Specifically, for the sparse network case, and for $4$ sectors, the first quartile is calculated to be $3.79$, which is $94.75 \%$ of the maximum possible. The same behavior is observed in Fig. \ref{fig:cdf2}, where we sectorize every node with up to $8$ sectors. \tingjun{update figures, then come back to this last part.}
%% figure begins

\iffullpaper
\subsection{Dynamic Routing: Sectorization and Induced Bipartiteness}

\begin{figure}[!t]
\centering
% \vspace{-\baselineskip}
\includegraphics[width=0.95\columnwidth]{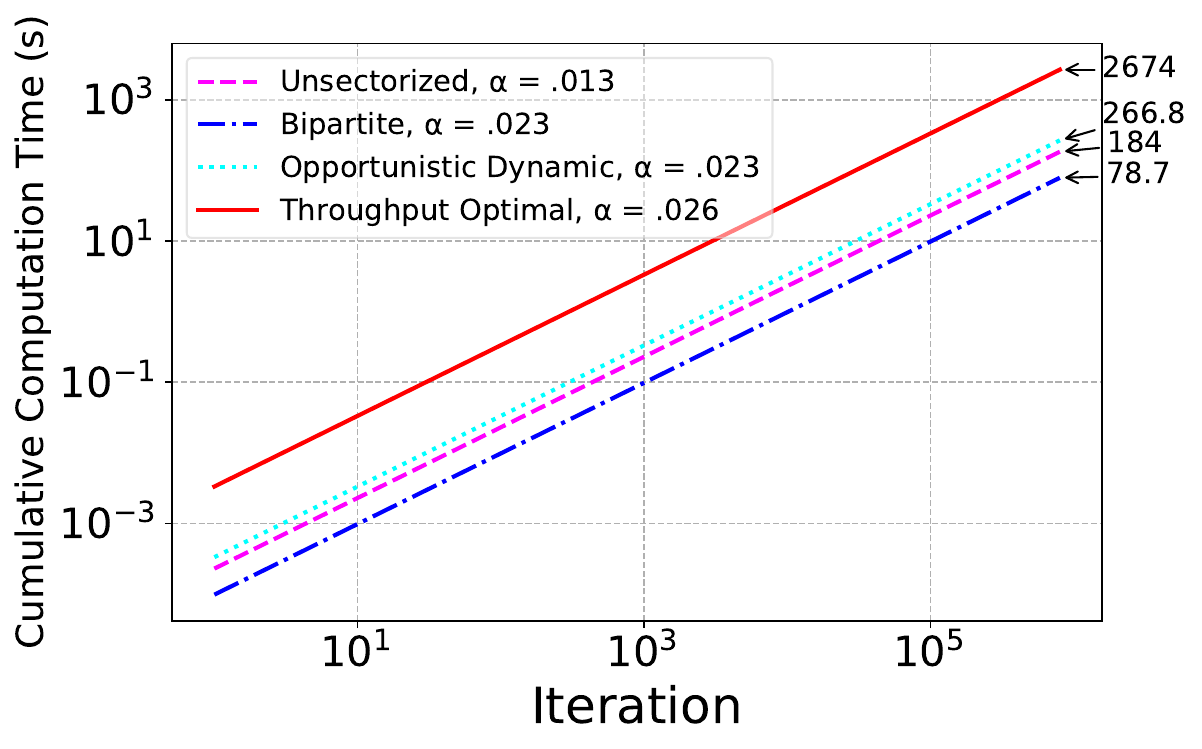}
\vspace{-3mm}
\caption{Cumulative time needed (in seconds) for the calculation of the max-weight in every iteration of the backpressure algorithm.}
\label{fig:eval-bipartite-time}
\vspace{-2mm}
\end{figure}

% The backpressure routing algorithm, described in Section~\ref{sec:model}, is throughput optimal for a general multi-hop network and can determine the routing and scheduling choices in a sectorized network by executing a polynomial time algorithm in every time slot. 
In this subsection, we study the capacity and evaluate the behavior of backpressure in a regular grid network with $\numNode = 16$ nodes and $\numSector = 2$ for each node. In particular, we assume that the nodes are placed on a square grid and that every node has a range such that it can communicate with its four incident neighbors as well as with the four nodes on its diagonals.

% We consider and compare three different techniques for sectorizing the network:
% (\emph{i}) unsectorized nodes,
% (\emph{ii}) Even Homogeneous Sectorization, where every node is sectorized with a horizontal line, and 
% (\emph{iii}) dynamic sectorization, as described in Section~\ref{sec:optimization}.
% In the second case, the network is sectorized in the initialization of the network using a simple horizontal sectorization to every node, which creates two sectors, one north and one south. Note that this sectorization is agnostic to the current flows, however it makes the auxiliary graph bipartite (see Section ~\ref{sec:bipartite}). In the third case, each node in the network is dynamically sectorized based on the current flows of the links using the procedure described in the discussion in the end of Section ~\ref{sec:optimization}. 

We compare four sectorization techniques for the network: (\emph{i}) unsectorized nodes, (\emph{ii}) Even Homogeneous Sectorization, where each node is divided with a horizontal line creating north and south sectors, (\emph{iii}) opportunistic dynamic sectorization as described in Section~\ref{ssec:opportunistic_dynamic_sectorization_algorithm}, and, (\emph{iv}) the throughput optimal dynamic sectorization described in Section~\ref{ssec:joint_throughput_optimality}. In Even Homogeneous Sectorization, the network is divided using a simple horizontal line for each node, irrespective of the network state, effectively making the auxiliary graph bipartite (see Section ~\ref{sec:bipartite}). Conversely, dynamic sectorizations adapt to current flow conditions by sectorizing each node based on the ongoing network state, as detailed in Section ~\ref{sec:routing}.

For the following numerical examples, we assume that the network receives uniform traffic and thus every pair of nodes has a commodity that has an arrival rate of $\alpha$. To implement that we assume that the packet arrival process $\{A_n^{(c)}(t)\}_t$ for every commodity $c$ is a Bernoulli process with probability of success $\alpha$. Nevertheless, we assume a dynamic setting where the network is not aware of the arrival rate matrix, $\arrivalRateMat$, and thus there is no information about the flow, $\flow$, that can be used when optimizing the node sectorization accordingly.

\vspace{0.5ex}
\noindent\textbf{Capacity Increase}. 
% We compare the capacity of the grid network under the three aforementioned sectorization techniques and uniform traffic with arrival rate $\alpha$. Similarly to the mechanism employed to compare the matching polytopes for different sectorizations where we have to compare the convex sets given a specific ``direction'', here we assume uniform arrival rates. 
We assess the grid network's capacity employing the four sectorization techniques mentioned earlier, under uniform traffic with an arrival rate of $\alpha$. 
% a sectorization $\sector$, $\capRegion(\graphNet)$.
% explicit sectorization gain for a sectorization $\sigma$ and the flow vector $\flowVec^\alpha$ that backpressure converges given the uniform arrival rate $\alpha$, $\sectorizationGainP^{\sigma}(\flowVec^\alpha)$. To calculate the explicit gain and not the approximate as in the previous sections, 
Specifically, we start from low values for the arrival rate $\alpha$ and increase until the blacklogs ``blow up''. This procedure is depicted in Fig.~\ref{fig:eval-capacity}, where we plot the summation of the 16$\times$16 different backlogs on the last iteration of the simulation with respect to the arrival rate. In Figs.~\ref{fig:eval-capacity}(a) (\emph{unsectorized network}), ~\ref{fig:eval-capacity}(b) (\emph{Even Homogeneous Sectorization}), and ~\ref{fig:eval-capacity}(d) (\emph{Throughput Optimal}) it can be observed that as the arrival rate exceeds the boundary of the capacity regions, the sum of the backlogs on the last iteration increases abruptly \cite{tassiulas1990stability}. In Fig.~\ref{fig:eval-capacity}(c) (\emph{Opportunistic Dynamic Sectorization}), note that the capacity boundary does not emerge by a single ``knee'' point as is the case for the other three techniques. In the latter, for $\alpha \in [0.23, 0.25]$, though the network seems stable, requests accumulate in backlogs in a faster rate. The Even Homogeneous Sectorization's performance closely aligns with that of the opportunistic dynamic sectorization, demonstrated by their capacity ratios to the unsectorized case, approximately $0.025/0.015 = 1.67$. The chosen Even Homogeneous Sectorization for this network is notably effective, as horizontal division fairly distributes uniform traffic across both sectors. In the case of the throughput optimal policy, we observe in Fig.~\ref{fig:eval-capacity}(d) that the capacity is \emph{two} times the one of the unsectorized network--which is the maximum we can expect using two sectors.  

% although the network appears to be stable, the backlogs accumulate requests until backpressure converges to a flow and thus our framework to the close to optimal sectorization. The Even Homogeneous Sectorization closely matches the performance of dynamic sectorization, as evidenced by their capacity ratios to the unsectorized case being roughly 0.025/0.015 = 1.67. In fact, the Even Homogeneous Sectorization picked for this network is expected to perform well indeed, since a horizontal sectorization divides the uniform traffic in the links in a fair manner between the two sectors. 

% This outcome is anticipated to be relatively independent of the node count, owing to the grid's symmetric nature.

\vspace{0.5ex}
\noindent\textbf{Bipartiteness \& Speed-up}. Next, we validate experimentally Theorem~\ref{thm:speedup}. In the previous section, we verified the capacity increase that the Even Homogeneous Sectorization introduce to the network. However, in Section~\ref{sec:bipartite}, we proved that the auxiliary graph becomes bipartite, something that introduces speedup in the dynamic routing schemes. To validate this without using very sophisticated algorithms for the problem of maximum weighted matchings in bipartite graphs (e.g.,~\cite{schrijver2003combinatorial}) we use Python's NetworkX library~\cite{hagberg2008exploring}. This library provides a distinct function for identifying maximum weighted matchings in bipartite graphs, differing from the algorithm for non-bipartite graphs by generating a complete graph and addressing the assignment problem instead. Despite not being as efficient as more recent algorithms, we observe an improvement in the network's dynamic routing performance using this method.

Specifically, in Fig.~\ref{fig:eval-bipartite-time} we plot the cumulative time needed to calculate the max-weight up until a given time slot. That is, in a time slot $t$ we plot the time the max-weight operations took cumulatively for time slots $1,2,\cdots,t$. For fairness, for every sectorization technique we set the network to receive arrival rate close to the boundary of the corresponding capacity region (see previous section). Therefore, the Even Homogeneous Sectorization case will be burdened with arrival rate that is $1.67$ times greater than the unsectorized case. Nevertheless, note from the figure that with the Even Homogeneous Sectorization we decrease the time needed for the computations by a factor of $2.34$ compared to the unsectorized case at time-slot $200,000$. Essentially, in other words, although we increase the capacity region by a factor of approximately $1.67$ and we make the connectivity graph more complex, we can decrease the time needed for the scheduling/routing operation by a factor of $2.34$. The corresponding factors compared to the opportunistic dynamic sectorization is $3.39$, and to the throughput optimal is $34$. These ratios can be further increased with more sectors since, as proved, with sufficient distributed computing power, the more we sectorize the network, the more we decrease the time needed for its stabilization. 

Fig.~\ref{fig:eval-bipartite-time} also shows that the throughput optimal policy described in Section~\ref{ssec:joint_throughput_optimality} needs significantly longer time to find the optimal sectorization and schedule compared to the opportunistic one (x10 times slower). 
\else 

\fi

\section{Future Directions}
\label{sec:future_directions}

\noindent\textbf{Analytical Trade-Offs and Choice of \(\boldsymbol{K}\).}
Our sectorization algorithm provides a provable approximation ratio and bounds but does not yield a closed-form relationship between the flow extension ratio and key parameters such as the number of nodes, network topology, flow properties, or the sector count \(K\). Consequently, we do not offer explicit guidelines on how to select \(K\) in practice, aside from our numerical exploration in Figure~7, which suggests diminishing returns as \(K\) grows. Investigating a more rigorous, closed-form analysis of these parameter trade-offs is a natural avenue for future work and would be highly valuable for network designers.

\vspace{5pt} \noindent\textbf{Even Homogeneous Sectorization Guarantees.}
In Section~\ref{sec:bipartite}, we showed how an Even Homogeneous Sectorization enforces a bipartite auxiliary graph, simplifying scheduling and speeding up backpressure stabilization. However, a more thorough characterization of its performance guarantees—especially when restricted to such sectorizations—would offer deeper insights.

\vspace{5pt} \noindent\textbf{Sectorizations to Alter Effective Topology.}
Finally, while we focus on even homogeneous sectorization in this work, other sectorization designs could be used to \emph{embed} advantageous topological properties that benefit additional algorithmic tasks, such as graph coloring or clustering. By selectively adding (or removing) certain links from sectors, a sectorization can reconfigure the network's ``effective" connectivity, potentially leading to improved performance for functions beyond scheduling and routing. Investigating such systematic approaches—where the sectorization is chosen to induce a desired topology and/or properties—remains an intriguing future direction.

% \input{tex/model_old}
% \input{tex/auxiliary_old}

% \input{tex/examples}
% \mynote{After the motivating examples, a few sections about different main results}
% \panos{Maybe fix everything else first, to see where is the best placement for that section.}

%%%%%
\section{Conclusion}
\label{sec:conclusion}

In this paper, we considered wireless networks employing sectorized infrastructure nodes that form a multi-hop mesh network for data forwarding and routing. We presented a general sectorized node model and characterized the capacity region of these sectorized networks. We defined the flow extension ratio and sectorization gain of these networks, which quantitatively measure the performance gain introduced by node sectorization as a function of the network flow. We developed an efficient distributed algorithm that obtains the node sectorization with an approximation ratio of 2/3. 
\iffullpaper
Further, we introduced a special class of sectorizations that boost the performance of dynamic routing schemes while increasing the capacity region.
\else 
\fi
We also introduced a throughput-optimal scheme that jointly optimizes routing, scheduling, and sectorization for networks capable of adapting in real time. 
We evaluated the proposed algorithm and achieved sectorization gain in various network scenarios via extensive simulations.

% \section*{Acknowledgements}
% This work was supported in part by NSF grants AST-2232458, CNS-2211944, CNS-2128638 and 2128530 RING, and by ARO grant W911NF2110325 MURI .

%%%%%
% \bibliographystyle{ACM-Reference-Format}
\bibliographystyle{ieeetr}
% \interlinepenalty=10000
\bibliography{mmwave_sectorization}

\begin{IEEEbiography}
[{\includegraphics[width=1in,height=1.25in,clip,keepaspectratio]{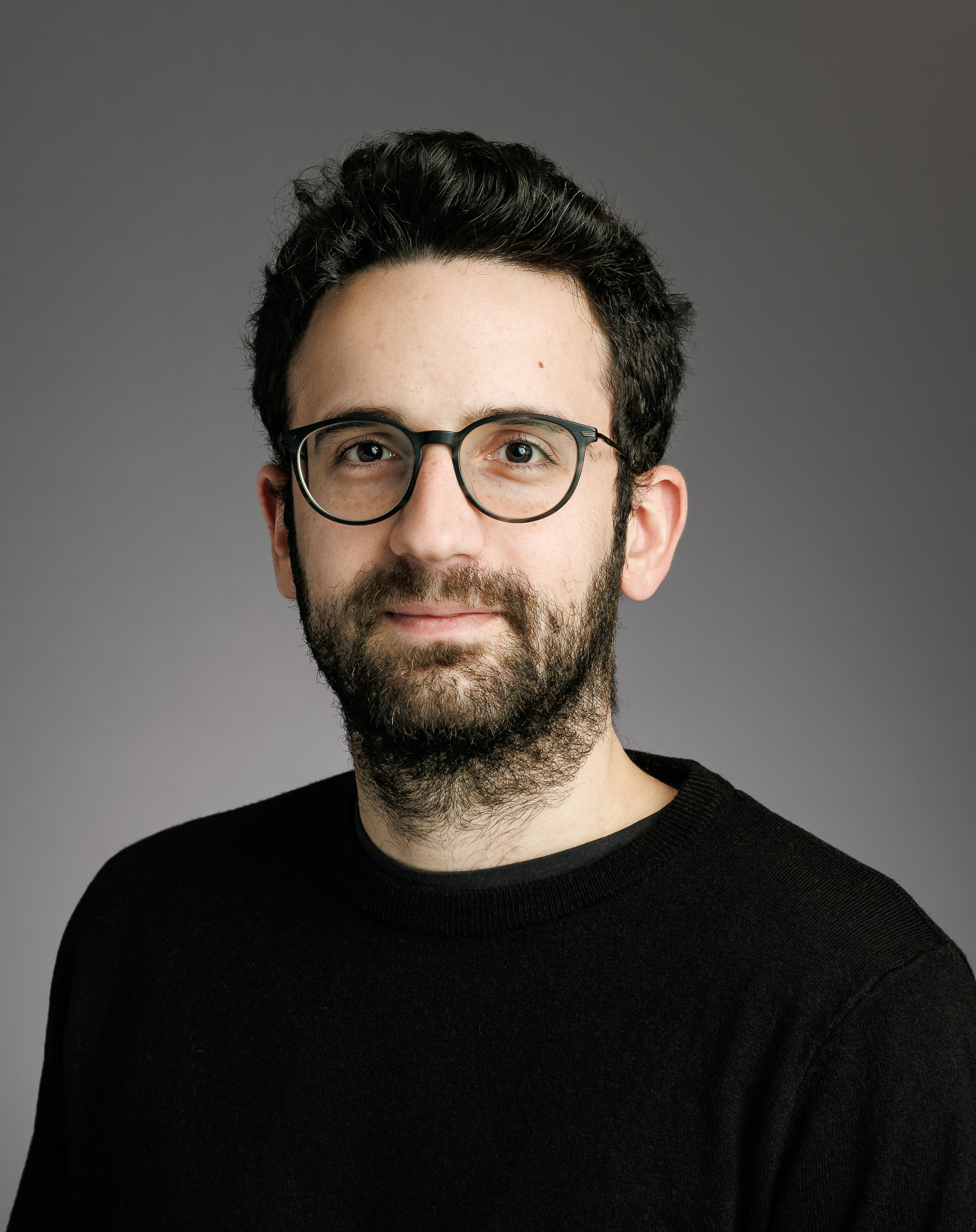}}]
{Panagiotis Promponas (Graduate Student Member, IEEE)} received his Diploma in Electrical and Computer Engineering (ECE) from the National Technical University of Athens (NTUA), Greece, in 2019. He is currently a Ph.D. student in the Department of Electrical and Computer Engineering at Yale University. His primary scientific interests include resource allocation in constrained interdependent systems and the optimization of algorithms. Specifically, he focuses on the optimization and modeling of  quantum networks and wireless networking systems. He is also the recipient (co-author) of the Best Paper Award at the 12th IFIP WMNC 2019.
\end{IEEEbiography}

\begin{IEEEbiography}
[{\includegraphics[width=1in,height=1.25in,clip]{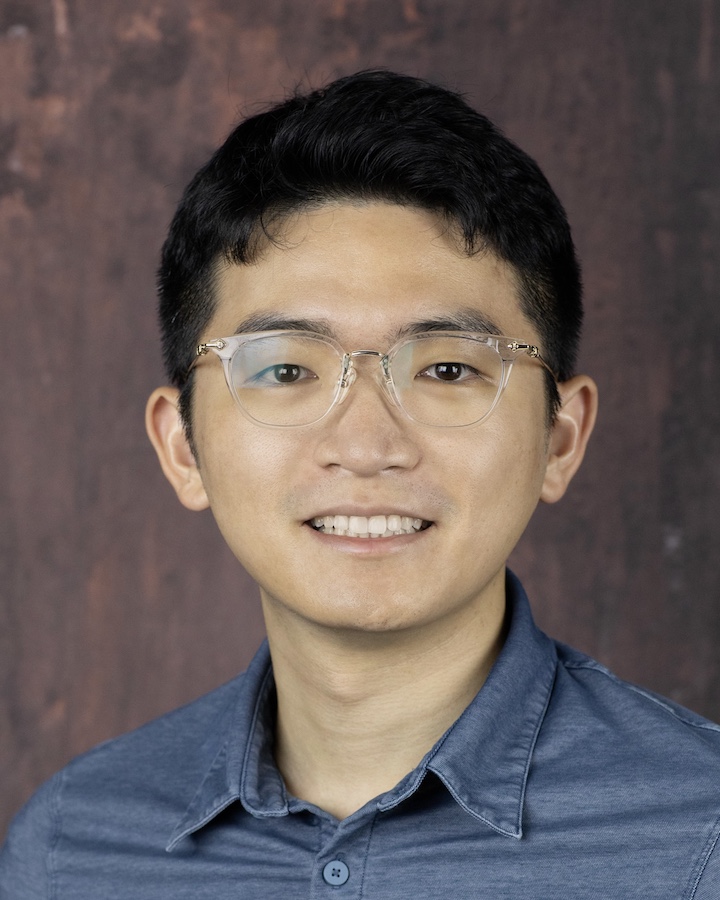}}]
{Tingjun Chen (Member, IEEE)}
is an Assistant Professor of Electrical \& Computer Engineering and Computer Science at Duke University. He received the B.Eng. degree in electronic engineering from Tsinghua University in 2014, the Ph.D. degree in electrical engineering from Columbia University in 2020, and was a Postdoctoral Associate with Yale University from 2020 to 2021. His research interests are in the area of networking and communications with a specific focus on next-generation wireless, mobile, and optical networks, as well as Internet-of-Things systems. He received the NSF CAREER Award, Google Research Scholars Award, IBM Academic Award, Facebook Fellowship, Wei Family Private Foundation Fellowship, Columbia Engineering Morton B. Friedman Memorial Prize for Excellence, Columbia University Eli Jury Award and Armstrong Memorial Award, the ACM SIGMOBILE Doctoral Dissertation Award Runner-Up, and the ACM CoNEXT’16 Best Paper Award.
\end{IEEEbiography}

\begin{IEEEbiography}
[{\includegraphics[width=1in,height=1.25in,clip,keepaspectratio]{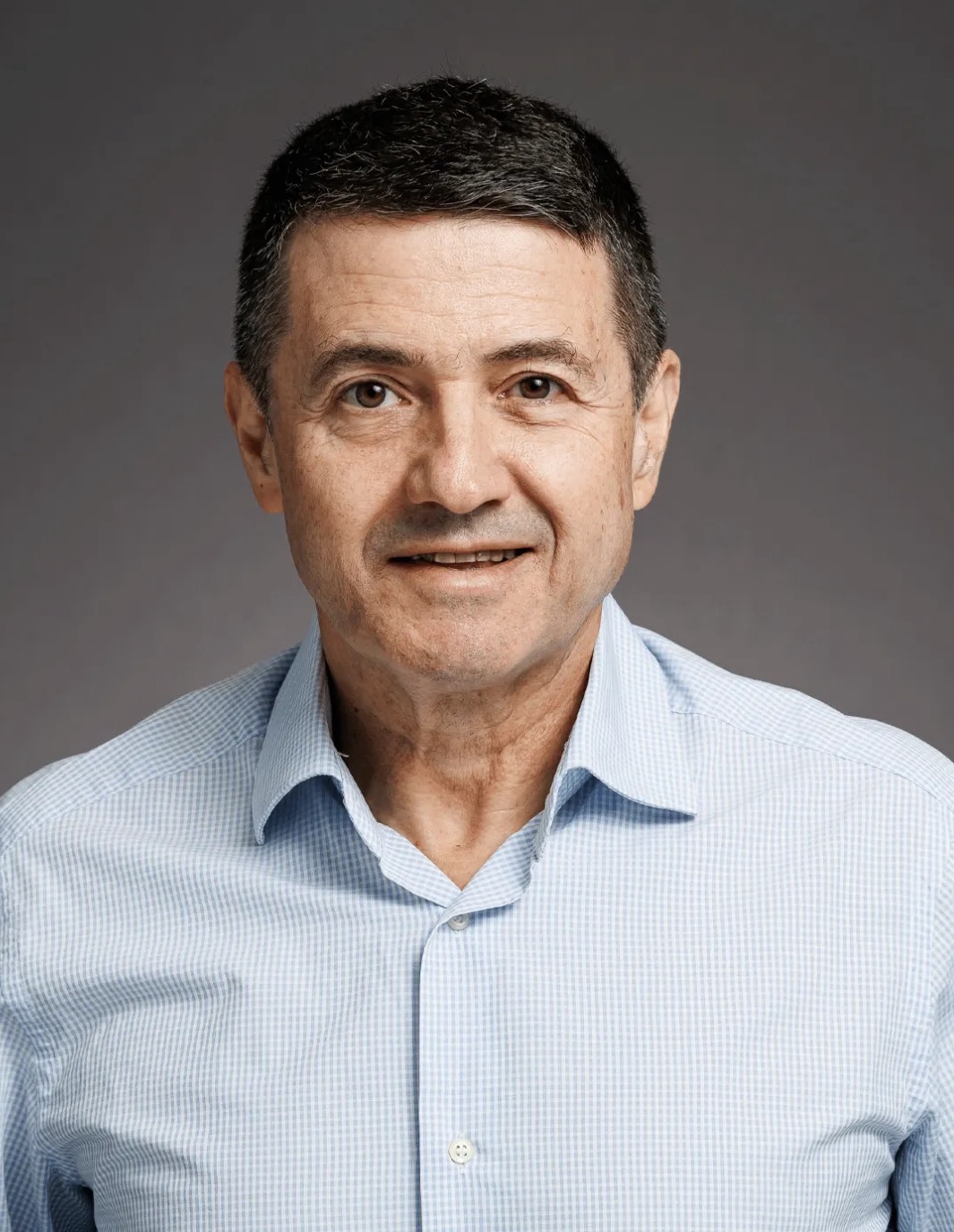}}]
{Leandros Tassiulas (Fellow, IEEE)} is the John C. Malone Professor of Electrical Engineering at Yale University, where he served as the department head 2016-2022.  His current research is on intelligent services and architectures at the edge of next generation networks including Internet of Things, sensing \& actuation in terrestrial and non-terrestrial environments and quantum networks. He worked in the field of computer and communication networks with emphasis on fundamental mathematical models and algorithms of complex networks, wireless systems and sensor networks. His most notable contributions include the max-weight scheduling algorithm and the back-pressure network control policy, opportunistic scheduling in wireless, the maximum lifetime approach for wireless network energy management, and the consideration of joint access control and antenna transmission management in multiple antenna wireless systems. Dr. Tassiulas is a Fellow of IEEE (2007) and of ACM (2020) as well as a member of Academia Europaea (2023). His research has been recognized by several awards including the IEEE Koji Kobayashi computer and communications award (2016), the ACM SIGMETRICS achievement award 2020, the inaugural INFOCOM 2007 Achievement Award ``for fundamental contributions to resource allocation in communication networks", several best paper awards including the INFOCOM 1994, 2017 and Mobihoc 2016, a National Science Foundation (NSF) Research Initiation Award (1992), an NSF CAREER Award (1995), an Office of Naval Research  Young Investigator Award (1997) and a Bodossaki Foundation award (1999). He holds a Ph.D. in Electrical Engineering from the University of Maryland, College Park (1991) and a Diploma of Electrical Engineering from Aristotle University of Thessaloniki, Greece. He has held faculty positions at Polytechnic University, New York, University of Maryland, College Park and University of Thessaly, Greece.
\end{IEEEbiography}

\end{document}